\documentclass[conference,10pt,onecolumn]{IEEEtran}

\IEEEoverridecommandlockouts
\usepackage{textcomp}
\usepackage{stfloats}
\usepackage{verbatim}
\usepackage{cite}
\usepackage{amsmath,amssymb,amsfonts}
\usepackage{algorithmic}
\usepackage{graphicx}
\usepackage{float}
\usepackage{textcomp}
\usepackage{xcolor}
\usepackage{bbm}
\usepackage{array}
\usepackage{rotating}
\usepackage{gensymb}
\usepackage{comment}
\usepackage[acronym]{glossaries}
\usepackage{multirow}
\usepackage{lscape}
\usepackage{makecell}
\usepackage[inline]{enumitem}
\usepackage{svg}
\usepackage{hyperref}
\hypersetup{hidelinks}
\usepackage[font=small,labelsep=period,justification=centering]{caption}
\usepackage{subcaption}
\usepackage{amsmath}
\usepackage{amsthm}
\usepackage[noabbrev, capitalize]{cleveref}
\usepackage[normalem]{ulem}

\usepackage{etoolbox} 
\makeatletter
\patchcmd{\subsubsection}{\itshape}{\bfseries}{}{}
\makeatother

\usepackage{adjustbox}
\usepackage{tikz}
\usetikzlibrary{shapes,arrows}
\usetikzlibrary{positioning}
\usetikzlibrary{decorations.text}

\newcounter{CorrCounter}
\newcounter{LemmaCounter}

\DeclareMathOperator*{\argmax}{arg\,max}

\newtheoremstyle{mystyle}%
  {1pt}    % espace au-dessus
  {1pt}    % espace en-dessous
  {\itshape} % corps
  {}        % indentation
  {\bfseries} % titre
  {.}       % ponctuation après titre
  {3pt}     % espace après le titre
  {}        % style du titre

\theoremstyle{plain}
\newtheorem{proposition}{Proposition}
\crefname{proposition}{Proposition}{Propositions}
\theoremstyle{plain}
\newtheorem{definition}{Definition}
\crefname{definition}{Definition}{Definitions}

\newcommand{\nocontentsline}[3]{}
\newcommand{\tocless}[2]{\bgroup\let\addcontentsline=\nocontentsline#1{#2}\egroup}

\def\BibTeX{{\rm B\kern-.05em{\sc i\kern-.025em b}\kern-.08em
    T\kern-.1667em\lower.7ex\hbox{E}\kern-.125emX}}

\allowdisplaybreaks

\begin{document}

\title{Resolution-Aliasing Trade-off in Near-Field Localisation}

\author{\IEEEauthorblockN{Baptiste Sambon, Gilles Monnoyer, Luc Vandendorpe, and Claude Oestges \thanks{Baptiste Sambon is a Research Fellow of the Fonds de la Recherche Scientifique - FNRS.}}\\
\IEEEauthorblockA{ICTEAM, UCLouvain - Louvain-la-Neuve, Belgium\\}
email: firstname.lastname@uclouvain.be
}

\maketitle

\thispagestyle{plain}%to add vs delete page numbers: plain vs empty
\pagestyle{plain}

%%%%%%%%%%%%%%%%%%%%%%

\begin{abstract}

Extremely Large-scale MIMO (XL-MIMO) systems operating in Near-Field (NF) introduce new degrees of freedom for accurate source localisation, but make dense arrays impractical. Sparse or distributed arrays can reduce hardware complexity while maintaining high resolution, yet sub-Nyquist spatial sampling introduces aliasing artefacts in the localisation ambiguity function. 
This paper presents a unified framework to jointly characterise resolution and aliasing in NF localisation and study the trade-off between the two. Leveraging the concept of local chirp spatial frequency, we derive analytical expressions linking array geometry and sampling density to the spatial bandwidth of the received field. We introduce two geometric tools--Critical Antenna Elements (CAEs) and the Non-Contributive Zone (NCZ)--to intuitively identify how individual antennas contribute to resolution and/or aliasing. 
Our analysis reveals that resolution and aliasing are not always strictly coupled, e.g., increasing the array aperture can improve resolution without necessarily aggravating aliasing. These results provide practical guidelines for designing NF arrays that optimally balance resolution and aliasing, supporting efficient XL-MIMO deployment.

\end{abstract}


%%%%%%%%%%% Introduction %%%%%%%%%%%

\section{Introduction}
\label{sec:introduction}

Multiple-Input Multiple-Output (MIMO) systems have been widely used in modern wireless communications thanks to their ability to enhance spectral efficiency, reliability, and multi-user support \cite{van_chien_massive_2017, bjornson_designing_2014, paulraj_overview_2004}. 
In this context, there are noticeable trends towards increasing both the array size and the carrier frequency, leading to Massive MIMO and Extremely Large-scale MIMO (XL-MIMO) systems \cite{lu_tutorial_2024,liu_near-field_2023,chen_6g_2024}. These trends extend the Fraunhofer distance--the conventional Far-Field (FF) limit which scales with the square of the array aperture and the carrier frequency--thereby enlarging the Near-Field (NF) region \cite{guerra_near-field_2021, ye_extremely_2024}. 
Unlike in the FF region, where wavefronts can be accurately approximated as planar, the NF region requires accounting for the Spherical Wavefront (SW) propagation. The latter introduces a range-dependent degree of freedom, complementing the traditional angular information available in the FF, thereby enabling source localisation capability. 

In the NF region, the SW propagation induces a non-linear dependence of both phase and amplitude on the exact antenna–source distance, so that each antenna measures a distinct response. As a result, the array effectively captures a wider spatial frequency spectrum that depends on both range and angle, which can be exploited to achieve improved spatial resolution \cite{kosasih_spatial_2025}. The NF resolution is thus inherently linked to the spatial bandwidth of the received field--that is, to the range of spatial frequencies captured across the array aperture. Larger apertures yield wider spatial bandwidths and, consequently, sharpens the localisation ambiguity function \cite{wachowiak_approximation_2025}.

Increasing the array size and carrier frequency is desirable to enhance the resolution, but it implies a higher number of closely spaced antennas to meet the half-wavelength inter-element spacing criterion, required to avoid grating lobes \cite{li_sparse_2025}. Implementing such dense arrays may become impractical due to the resulting complexity and hardware limitation.
Alternative array topologies, such as sparse or distributed arrays \cite{li_sparse_2025, ye_extremely_2024}, reduce the number of antennas by increasing inter-element spacing but induce sub-Nyquist spatial sampling, resulting in aliasing artefacts in the localisation ambiguity function \cite{spawc}. 
In the NF, these artefacts can manifest as multiple peaks, which may be misinterpreted as potential source locations, thereby creating additional ambiguities in the localisation process.

These results lead us to an apparent conflict: for a constrained number of antennas, an improved resolution requires a larger aperture but increases the risk of aliasing owing to higher antenna spacing.
Conversely, mitigating aliasing artefacts requires a denser array, which limits the aperture and thus the achievable resolution.
In FF scenarios, this trade-off has been discussed in \cite{wang_spatial-sampling-based_2025} for Direction-of-Arrival (DoA) estimation.
However, this analysis is limited to linear arrays and only accounts for the inter-element spacing. 
In the NF context, previous studies concentrated their effort on treating these two effects separately, without accounting for their potentially complex interplay.

\medskip

More precisely, the authors in \cite{wachowiak_approximation_2025,lu_tutorial_2024, lei_near-field_2024} related the achievable \emph{resolution} of the localisation ambiguity function to the array aperture. 
These works restricted their focus to dense arrays. Hence, they did not include the impact of spatial sampling that can cause aliasing.

On the other hand, the impact of \emph{aliasing} on the localisation ambiguity function in the NF regime has been thoroughly investigated in \cite{spawc, npj, camsap}.
These studies exploited a chirp-based modelling of steering signals to compactly represent their spectral content through the set of \textit{local spatial frequency} they capture. 
This concept had previously been used to characterise the Degrees of Freedom (DoF) achievable in NF communications with XL-MIMO \cite{ding_degrees_2022, ding_spatial_2024, pan_analysis_2025,kosasih_spatial_2025}.
On this basis, the theoretical framework we recently formalised \cite{npj} and exploited \cite{camsap} provides a systematic methodology to describe the geometry of the NF aliasing artefacts in the ambiguity function.
These contributions resulted in analytical insights into the spatial sampling requirements at the receive array for NF point source uplink localisation. 
However, these formulations remain limited to one-dimensional arrays, and these findings lack practical interpretability with regards to structural system parameters. 
Additionally, they do not address the achievable spatial resolution nor the fundamental trade-off between resolution and aliasing. 
These are the gaps we intend to bridge in the current paper.

Specifically, building upon the framework in \cite{npj}, the current paper introduces new tools to jointly characterise aliasing and resolution in NF localisation scenarios.
This extension enables a comparison of array configurations that underscore the trade-off between the two concepts.
It further reveals decoupling conditions where a larger aperture, improving resolution, can be reached without strengthening the presence of aliasing artefacts, thereby providing deeper insight into the design of NF antenna arrays.
To the best of the authors' knowledge, no prior work has jointly investigated the resolution-aliasing trade-off in NF localisation scenarios, considering a comprehensive set of system parameters and array configurations.

\subsection{Contributions}

Based on the aforementioned works, the contributions of this paper are summarised as follows.

\begin{itemize}
    \item \textbf{Resolution-aliasing trade-off}: The chirp framework is leveraged to extend the aliasing analysis of \cite{npj} to broader configurations (e.g., multidimensional arrays) and to jointly integrate the notions of resolution and aliasing.
    This extended framework enables us to characterise both quantities and to explicitly analyse their trade-off.
    \item \textbf{Geometric insights}: We introduce the concepts of Critical Antenna Elements (CAEs) and Non-Contributive Zone (NCZ), explicitly connecting the array geometry to the above trade-off. These tools highlight which antennas in the array generate aliasing artefacts and formalise locations where adding an antenna leads to a resolution improvement.
    \item \textbf{Design guidelines}: We provide an explicit characterisation of how system parameters (array geometry, spacing, relative source position) govern the resolution-aliasing trade-off, providing recommendations into designing NF antenna arrays that balance resolution and aliasing considerations. 
\end{itemize}

\subsection{Paper Structure}

The remainder of this paper is structured as follows. First, \cref{sec:system_model} introduces the system model. \cref{sec:ambiguity_function_analysis} analyses the ambiguity function, the impact of spatial sampling, the extension of the Aliasing-Free Region (AFR) definition, and the achievable resolution. Then, \cref{sec:chirp} derives closed-form expressions for the non-aliasing conditions and the spatial bandwidth. In \cref{sec:geometrical_tools}, two geometrical tools are introduced: the Critical Antenna Elements and the Non-Contributive Zone. Finally, \cref{sec:point_scatterer}  and \cref{sec:polar_arrays} characterise the resolution-aliasing trade-off for rectangular and circular arrays, respectively, before concluding in \cref{sec:conclusion}.

\subsection{Notations}

Scalars are written in regular font (e.g., $x$), vectors in bold lowercase (e.g., $\boldsymbol{x}$), and sets in calligraphic uppercase (e.g., $\mathcal{X}$). The boundary of a set $X$ is denoted by $\partial X$. 
Strict (nominal) quantities are denoted with an overbar (e.g., $\bar{K}$), whereas their chirp-based representations are written in roman font (e.g., $\mathrm{K}$).
For any vector $\boldsymbol{x}$, $\|\boldsymbol{x}\|$ and $\angle \boldsymbol{x}$ denote its Euclidean norm and phase, and $x_i$ its $i$-th component. The operators $(\cdot)^*$ and $(\cdot)^T$ denote complex conjugation and transposition, $j=\sqrt{-1}$ is the imaginary unit, “$\cdot$” denotes the dot product, and “$\triangleq$” indicates a definition.

%%%%%%%%%%% System Model %%%%%%%%%%%

\section{System Model}
\label{sec:system_model}

\begin{figure}
    \centering
    \includegraphics[width=0.5\linewidth, trim=140 300 150 100, clip]{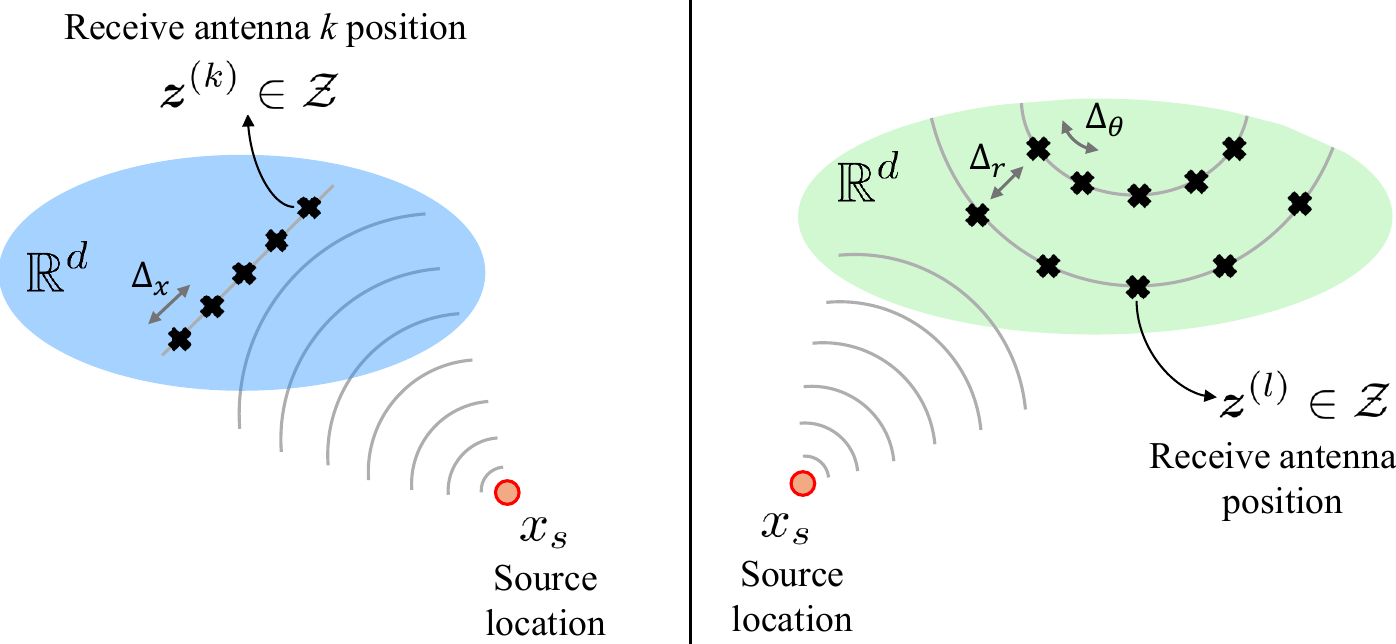}
    \caption{Configuration with a source location $\boldsymbol{x}_s$, receive antenna set $\mathcal{Z}$, in a near-field configuration.}
    \label{fig:bistatic_system_model}
\end{figure}

As illustrated in Figure~\ref{fig:bistatic_system_model}, we consider an uplink scenario where a signal transmitted by a static source positioned at \( \boldsymbol{x}_s \) is captured by a set of \( N \) receiving antennas. The receive array is defined as  
\begin{equation}
    \mathcal{Z} \triangleq \{ \boldsymbol{z}^{(k)} \}_{k=1}^{N} \subset \mathbb{R}^{d},
\end{equation}
where \( \boldsymbol{z} \) denotes the coordinates of the antenna positions and \( d \in \{1, 2, 3\} \) indicates the array dimensionality: \( d=1 \) for linear, \( d=2 \) for planar, and \( d=3 \) for volumetric arrays.
The antennas are assumed to be uniformly spaced along each spatial axis, with inter-element spacing \( \Delta_{i} \), where \( i \) indicates the spatial dimension. In all generality, the antenna coordinates $\boldsymbol{z}$ can be expressed in different coordinate systems, such as Cartesian or spherical, i.e. $i \in \{x, y, z\}$ or $i \in \{r, \theta, \phi\}$.

Following the physical optics approximation \cite{gutierrez-meana_high_2011}, the noise-free signal $s_R$ at a receive antenna located at $\boldsymbol{z}$ from a source at $\boldsymbol{x}_s$ can be expressed as 
\begin{equation}
    s_R(\boldsymbol{z}, t;\boldsymbol{x}_s) = h(\boldsymbol{z}; \boldsymbol{x}_s) \, \, s\left(t - \frac{\| \boldsymbol{x}_s - \boldsymbol{z} \|}{c}\right),
\label{eq:signal}
\end{equation}
with 
\begin{equation}
    h(\boldsymbol{z}, \boldsymbol{x}) \triangleq \frac{e^{-jk \| \boldsymbol{x} - \boldsymbol{z} \|}}{\| \boldsymbol{x} - \boldsymbol{z} \|}, 
\end{equation}
where \( s(t) \) is the transmitted baseband signal, \( k~=~\frac{2\pi}{\lambda} \) is the wavenumber with wavelength $\lambda$ and $c$ is the speed of light. The term \( h(\boldsymbol{z}; \boldsymbol{x}_s) \) models the propagation channel between the source and the receive antenna, which includes both the free-space path loss and the phase shift due to propagation delay. In this paper,  we assume narrowband signals $s(t)$. This assumption allows us to focus on the wave propagation characteristics $h(\boldsymbol{z}, \boldsymbol{x})$ of the received signal without considering the temporal aspects. Note that the extension to wideband temporal signals is left for future work.

%%%%%%%%%%% Ambiguity Function Analysis %%%%%%%%%%%

\section{Ambiguity Function Analysis}
\label{sec:ambiguity_function_analysis}

The ambiguity function ($AF$) is a fundamental tool that characterises the performance and limitations of a sensing system. The AF is commonly used to theoretically assess the ability of a system to distinguish between different spatial locations, revealing both the achievable resolution and the potential ambiguities that may arise due to aliasing artefacts when when estimating the source location $\boldsymbol{x}_s$. This approach enables to derive fundamental insights into the relationship between array geometry, spatial sampling, and localisation performance.
% The AF enables to theoretically assess the ability to ... 
% is commonly used to assess the theoretical ability to ... 

%The ambiguity function ($AF$) plays a key role in estimating the source location \( {\boldsymbol{x}}_s \) from the received signals \( s_R(\boldsymbol{z};\boldsymbol{x}_s) \). %The objective is to identify and locate the source \( {\boldsymbol{x}}_s \), while excluding any artefactual or other responses.  
%Following a Maximum Likelihood (ML) estimation approach \cite{trees_detection_2004}, the $AF$ can be formulated as the outcome of a $\boldsymbol{z}$-matched filtering operation applied to a \textit{tentative} location \( \tilde{\boldsymbol{x}}_s \). This process maximises the Signal-to-Noise Ratio (SNR) in Additive White Gaussian Noise (AWGN) conditions and reflects the likelihood of selecting \( \tilde{\boldsymbol{x}}_s \) as the source location. 
%However, in practical scenarios, aliasing may introduce artefacts in the $AF$, thereby creating ambiguities in source localisation. It is therefore crucial to characterise the $AF$ regions that are impacted by aliasing, distinguish them from artefact-free areas, and investigate how system parameters influence these areas.

%\subsection{Ambiguity Function Construction}

Assuming initially space-continuous antenna arrays, the continuous ambiguity function $AF_c$ matched at a tested location \( \tilde{\boldsymbol{x}}_s \) (e.g. \( \tilde{\boldsymbol{x}}_s = (\tilde{x}_s, \tilde{y}_s, \tilde{z}_s) \in \mathbb{R}^{3} \)  ) is thus expressed as
\begin{align}
    AF_c(\tilde{\boldsymbol{x}}_s ; \boldsymbol{x}_s) 
    &= \int_{\mathcal{Z}} 
    g(\boldsymbol{z} ; \tilde{\boldsymbol{x}}_s, \boldsymbol{x}_s)  
    \, \mathrm{d}\boldsymbol{z}, 
    \label{eq:reconstructed_image_continuous}
\end{align}
where 
\begin{equation}
    g(\boldsymbol{z} ; \tilde{\boldsymbol{x}}_s, \boldsymbol{x}_s) \triangleq  
    h(\boldsymbol{z}, {\boldsymbol{x}}_s) \, h^{*}(\boldsymbol{z}, \tilde{\boldsymbol{x}}_s). 
\end{equation}   

In the following, \cref{sec:discrete_arrays_and_non_aliasing_conditions} discusses the impact of spatial sampling on the $AF$ and derives the non-aliasing conditions while \cref{sec:localisation_resolution} explores the achievable localisation resolution and qualitatively observes the trade-off with aliasing.

\subsection{Discrete Antenna Arrays and Non-Aliasing Conditions}
\label{sec:discrete_arrays_and_non_aliasing_conditions}

In practical scenarios, the antenna arrays are space-discrete, and the integral in \eqref{eq:reconstructed_image_continuous} becomes a summation over the antenna positions \( \boldsymbol{z}^{(k)} \in \mathcal{Z} \), \( k = 1, \ldots, N \). The discrete $AF$ is thus given by
\begin{align}
    AF(\tilde{\boldsymbol{x}}_s ; \boldsymbol{x}_s) 
    &= \sum_{k=1}^{N} 
    g(\boldsymbol{z}^{(k)} ; \tilde{\boldsymbol{x}}_s, \boldsymbol{x}_s). 
    \label{eq:reconstructed_image_discrete}
\end{align}
However, this spatial sampling may lead to aliasing artefacts in the $AF$. In FF, for instance, when the inter-antenna spacing exceeds half a wavelength (i.e., \( \Delta_{i} > \lambda/2 \text{ for some axis } i \)), spatial aliasing arises \cite{jerri_shannon_1977}. 
Such aliasing produces repeated patterns that degrade the $AF$. In particular, these replicas may create spurious peaks at locations where no source is present, thereby introducing additional ambiguities in the $AF$.

\begin{figure}
    \centering
    \makebox[0.4\linewidth][c]{\small $AF(\tilde{\boldsymbol{x}}_s ; \boldsymbol{x}_s)$}
    \makebox[0.2\linewidth][c]{\small \hspace{-6.5em} Zoom}
    %\rotatebox{90}{\makebox[0.3\linewidth][c]{\small \hspace{1cm} $\Delta_x = 4\lambda$ \hspace{1cm} $\Delta_x = 4 \lambda$ \hspace{1cm} $\Delta_x = \lambda/2$}}
    \includegraphics[width=0.45\textwidth, trim=10 60 450 10, clip]{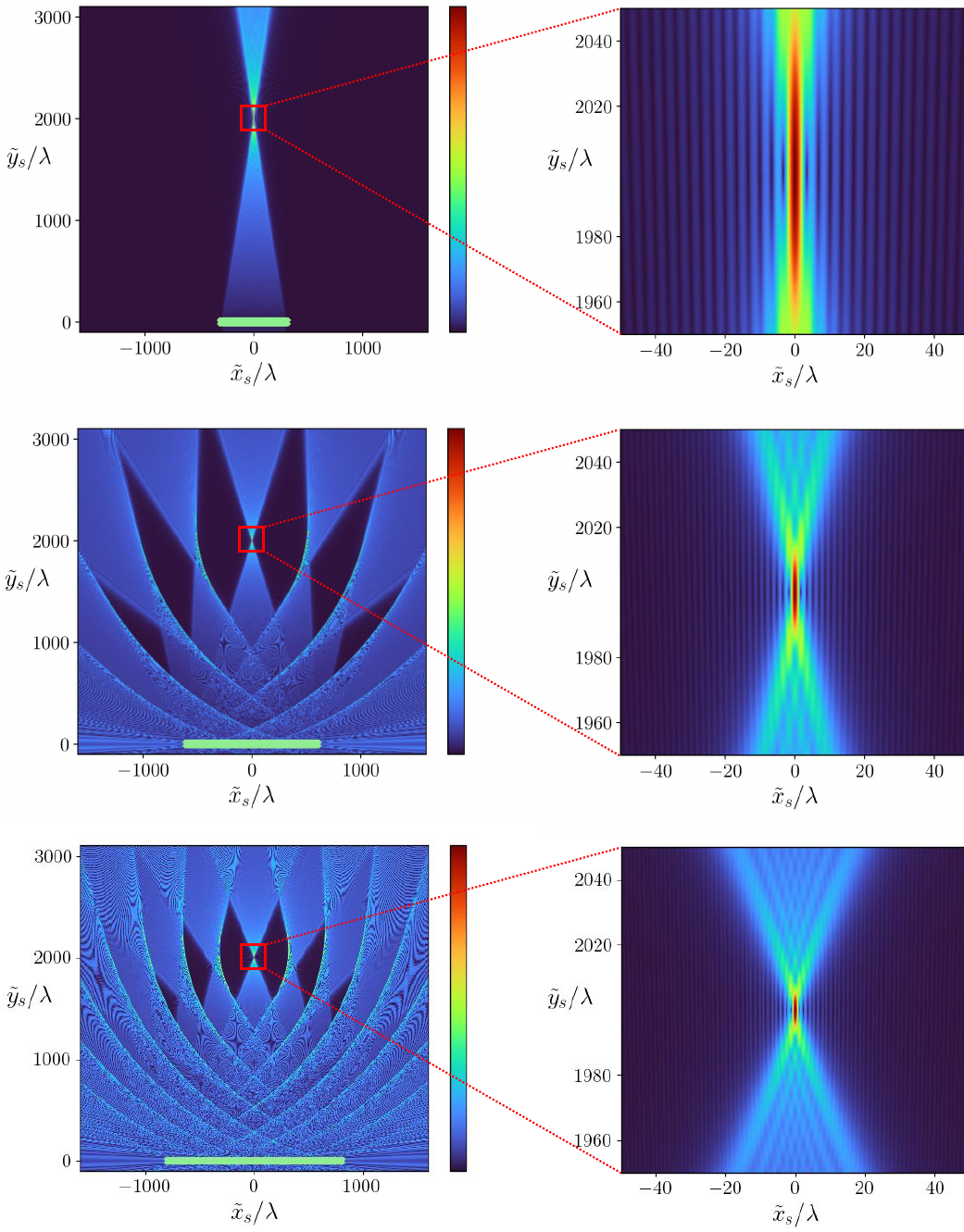} \\
    \includegraphics[width=0.12\linewidth, trim=0 50 405 40, clip]{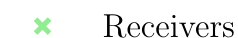} 
    \caption{Ambiguity functions of a point source at \( \boldsymbol{x}_s = (0, 2000) \lambda \) for linear arrays of $N=1200$ antennas over a $600\lambda$ aperture (first row), $N=300$ antennas over a $1200\lambda$ aperture (second row), $N=400$ antennas over a $1600\lambda$ aperture (third row) . }
    \label{fig:mono_aliased_vs_non_aliased}
\end{figure}

\cref{fig:mono_aliased_vs_non_aliased} illustrates this phenomenon for source at \( \boldsymbol{x}_s = (0, 2000)\lambda \). From top to bottom, the number of antennas reduces from \( N=1200\) to \( N=300 \) and $N=400$, while the aperture increases from \( 600\lambda \) to \( 1200\lambda \) and $1600\lambda$. Consequently, the inter-element spacing grows from \( \Delta_x = \lambda/2 \) (first row) to \( \Delta_x = 4\lambda > \lambda/2 \) (second and third rows). 
As a result, the $AF$ transitions from an artefact-free response (top figure) to ones exhibiting significant aliasing artefacts (bottom figures). Notably, the presence of aliasing artefacts complicates the localisation process by creating additional high-spurious responses in the $AF$. 

Due to aliasing, the $AF$ exhibits additional repeated peaks along the angular direction, rather than a single one as in the non-aliased case. This is consistent with the conventional directional cosine rule observed in FF conditions. In contrast, for NF configurations, the repeated patterns follow a curved conical shape due to the ability of NF communication to provide accurate spatial localisation.  

To establish the conditions under which a tentative location $\tilde{\boldsymbol{x}}_s$ remains free from aliasing in the $AF$, we now examine the propagation term $g(\boldsymbol{z} ; \tilde{\boldsymbol{x}}_s, \boldsymbol{x}_s)$ appearing in the integrand of \eqref{eq:reconstructed_image_continuous}. 
For a spatial wavevector \( \bar{\boldsymbol{k}}^{g} \), the Fourier transform $G$ of this term is given by 
\begin{equation}
    G(\bar{\boldsymbol{k}}^{g} ; \tilde{\boldsymbol{x}}_s, \boldsymbol{x}_s) = \sum_{k=1}^{N} 
    g(\boldsymbol{z}^{(k)} ; \tilde{\boldsymbol{x}}_s, \boldsymbol{x}_s) e^{-j \bar{\boldsymbol{k}}^{g} \cdot \, \boldsymbol{z}^{(k)}}.  %\mathrm{d}\boldsymbol{z}_. 
\end{equation}

A key point is that the spectrum \( G \) evaluated at \( \bar{\boldsymbol{k}}^{g}~=~\boldsymbol{0} \) corresponds to the integral of \( g(\boldsymbol{x} ; \tilde{\boldsymbol{x}}_s, \boldsymbol{x}_s) \) over the array aperture and thus to the ambiguity function $AF(\tilde{\boldsymbol{x}}_s, \boldsymbol{x}_s)$, i.e.,
\begin{equation}
    G(\boldsymbol{0} ; \tilde{\boldsymbol{x}}_s, \boldsymbol{x}_s) = \sum_{k=1}^{N} 
    g(\boldsymbol{z}^{(k)} ; \tilde{\boldsymbol{x}}_s, \boldsymbol{x}_s) = AF(\tilde{\boldsymbol{x}}_s, \boldsymbol{x}_s).
\end{equation} 
As illustrated in Figure~\ref{fig:spectral_folding} for an arbitrary triangular spectrum $G$ along the $i$-th dimension (shown in \cref{fig:fig1}), discretising \( \mathcal{Z} \) with a spacing \( \Delta_{i} \) replicates the spectrum \( G \) along this same dimension with a period \( {2\pi}/{\Delta_{i}} \)  (see \cref{fig:fig2}) \cite{dudgeon_multidimensional_1984}. These periodic replicas induce aliasing (spectral folding) in the $AF$ at the tentative location $\tilde{\boldsymbol{x}}_s$ whenever they overlap the component at \( \bar{\boldsymbol{k}}^{g} = \boldsymbol{0} \) (as shown in \cref{fig:fig4}), as this component corresponds $AF(\tilde{\boldsymbol{x}}_s; \boldsymbol{x}_s) $.
In this setting, the conventional Nyquist sampling criterion \cite{dudgeon_multidimensional_1984, shannon_communication_1949} which requires all replicas to be disjoint, must be reconsidered: relaxing this criterion by a factor two can still yield accurate reconstruction (see \cref{fig:fig3}), as aliasing away from the origin in the spatial-frequency domain does not degrade the $AF$. Accordingly, the condition ensuring that $\tilde{\boldsymbol{x}}_s$ remains aliasing-free is\footnote{As stated, this condition is less restrictive than the Nyquist criterion, which requires non-overlapping spectral replicas, i.e. $2\bar{K}_{i}(\tilde{\boldsymbol{x}}_s, \boldsymbol{x}_s) < 2\pi/\Delta_{i}$.}
\begin{equation}
    \bar{K}_{i}(\tilde{\boldsymbol{x}}_s, \boldsymbol{x}_s) \le 2\pi/\Delta_{i} \quad \forall i, 
\label{eq:aliasing_conditions}
\end{equation}
where \( \bar{K}_{i}(\tilde{\boldsymbol{x}}_s, \boldsymbol{x}_s) \) is the maximum spatial frequency of the spectrum \( G \) along the \( i \)-th dimension strictly defined as 
\begin{equation}
    \bar{K}_{i}(\tilde{\boldsymbol{x}}_s, \boldsymbol{x}_s) 
    \triangleq 
    \max \bigl\{\, |\bar{k}_i^{g}| \; : \; 
    G(\bar{\boldsymbol{k}}^{g} ; \tilde{\boldsymbol{x}}_s, \boldsymbol{x}_s) \neq 0 \, \bigr\},
\end{equation}
with \( \bar{k}_i^{g} \triangleq (\bar{\boldsymbol{k}}^{g})_i \) being the \( i \)-th component of the spatial wavevector~\( \bar{\boldsymbol{k}}^{g} \). Conversely, if \eqref{eq:aliasing_conditions} is violated, aliasing artefacts will appear in the $AF$ at \( \tilde{\boldsymbol{x}}_s \).

\begin{figure}
    \centering
    % Figure 1
    \begin{subfigure}[t]{0.5\linewidth}
        \centering
        \includegraphics[width=0.65\linewidth, trim=190 390 190 140, clip]{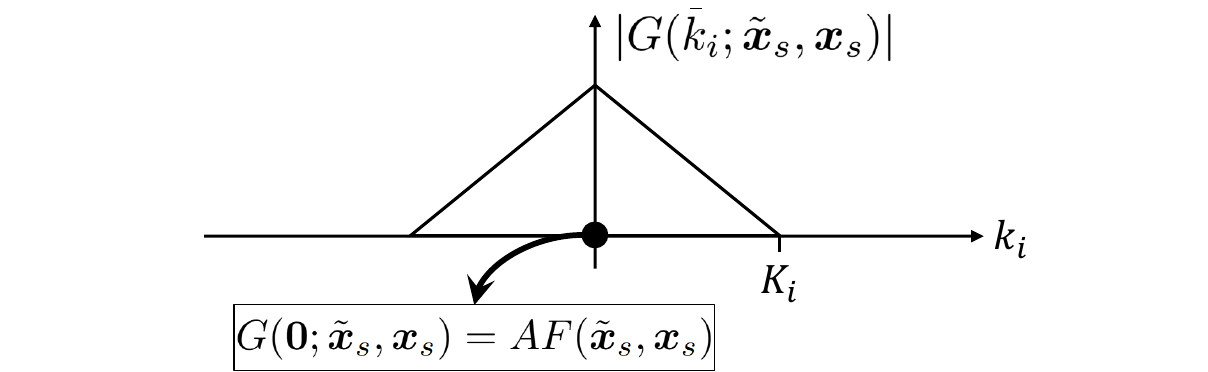}
        \caption{Magnitude of arbitratry spectrum $G(\bar{\boldsymbol{k}}^{g} ; \tilde{\boldsymbol{x}}_s, \boldsymbol{x}_s)$ along the \( i \)-th dimension, i.e. $G(\bar{k}_i^{g} ; \tilde{\boldsymbol{x}}_s, \boldsymbol{x}_s)$, with maximum spatial frequency \( \bar{K}_i \). }
        \label{fig:fig1}
    \end{subfigure}
    
    \vspace{0.5em}
    % Figure 2
    \begin{subfigure}[t]{0.5\linewidth}
        \centering
        \includegraphics[width=\linewidth, trim=10 550 10 10, clip]{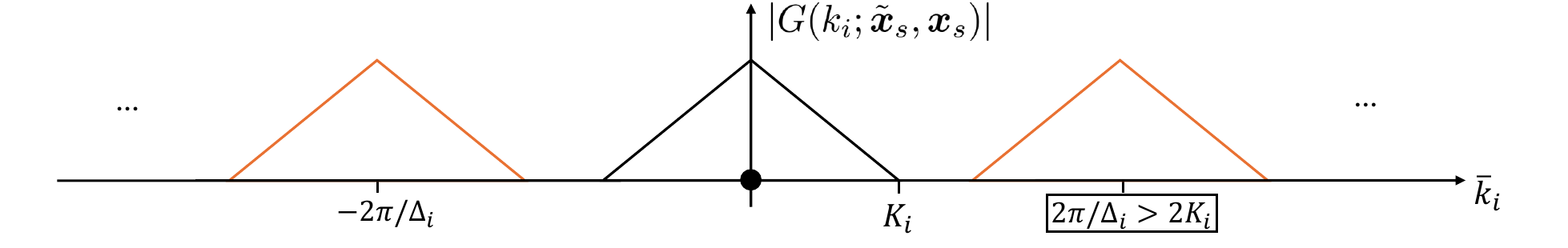}
        \caption{$\Delta_i$ is sufficiently small to avoid any overlap between spectral replicas.}
        \label{fig:fig2}
    \end{subfigure}
    
    \vspace{0.5em}
    % Figure 3
    \begin{subfigure}[t]{0.5\linewidth}
        \centering
        \includegraphics[width=\linewidth, trim=10 345 10 190, clip]{Figures_needed/Nyquist_Spectrum_analysis_bis_bis_new_notation.pdf}
        \caption{Nyquist sampling criterion not respected: replicas overlap. However, no aliasing occurs in the $AF$ since the $\bar{\boldsymbol{k}}^{g}=\boldsymbol{0}$ component is preserved.}
        \label{fig:fig3}
    \end{subfigure}
    
    \vspace{0.5em}
    % Figure 4
    \begin{subfigure}[t]{0.5\linewidth}
        \centering
        \includegraphics[width=\linewidth, trim=10 125 10 400, clip]{Figures_needed/Nyquist_Spectrum_analysis_bis_bis_new_notation.pdf}
        \caption{Overlap will cause aliasing in the $AF$ since $\bar{\boldsymbol{k}}^{g}=\boldsymbol{0}$ will no longer be preserved.}
        \label{fig:fig4}
    \end{subfigure}
    
    \caption{Illustration of the spectral folding phenomenon and its impact on aliasing in the $AF$.}
    \label{fig:spectral_folding}
\end{figure}

\subsection{Localisation Resolution}
\label{sec:localisation_resolution}

The resolution of a system can be defined as the minimum distance between two point sources that can be distinguished. It is typically measured as the Full Width at Half Maximum (FWHM), i.e. the distance between two points for which the amplitude falls to 50\% (or \(-3\,\mathrm{dB}\)) of the peak value \cite{ahmed2014electronic}. In the $AF$, this translates to the sharpness of the peak at the true source location \( \boldsymbol{x}_s \). A sharper peak indicates a higher resolution, enabling more accurate source localisation.

The resolution along axis $i$, denoted $\delta x_i(\boldsymbol{x}_s)$, at a source point $\boldsymbol{x}_s$  is linked to the spatial bandwidth \( \bar{B}_i(\boldsymbol{x}_s, \mathcal{Z}) \) along that axis of the received field $h(\boldsymbol{z}, \boldsymbol{x}_s)$ captured by the array $\mathcal{Z}$. %from the source located in $\boldsymbol{x}_s$. 
This relationship can be expressed as\footnote{These derivations neglect the amplitude distribution in the spatial frequency domain. %which can affect the shape of the focused point \cite{ahmed2014electronic}. 
Nevertheless, this simplification remains sufficiently accurate to capture the parameter dependencies \cite{ahmed2014electronic}.} \cite{ahmed2014electronic}
\begin{equation}
    \delta x_i(\boldsymbol{x}_s) = \frac{ 2\pi}{\bar{B}_i(\boldsymbol{x}_s, \mathcal{Z})}.  %\quad i \in \{x, y, z\},   
\label{eq:resolution}
\end{equation}
Thus, larger bandwidth leads to finer resolution, as the system captures more detailed spatial frequency information from the source, while a smaller bandwidth results in coarser resolution. 

Formally, $\bar{B}_i(\boldsymbol{x}_s, \mathcal{Z})$ is thus linked to the support of 
\begin{equation}
    H(\bar{\boldsymbol{k}}^{h} ; \boldsymbol{x}_s, \mathcal{Z}) \triangleq 
    \int_{\mathcal{Z}} h(\boldsymbol{z}, \boldsymbol{x}_s) \, e^{-j \bar{\boldsymbol{k}}^{h} \cdot \boldsymbol{z}} \, \mathrm{d}\boldsymbol{z},  
\label{eq:spatial_spectrum}
\end{equation}
i.e. the spatial spectrum of the received field $h(\boldsymbol{z}, \boldsymbol{x}_s)$, and can be strictly defined as the spread of spatial frequencies $\bar{\boldsymbol{k}}^{h}$ captured across the array $\mathcal{Z}$ along axis \( i \):
\begin{equation}
    \bar{B}_i(\boldsymbol{x}_s, \mathcal{Z}) \triangleq \bar{k}_{i,\max} (\boldsymbol{x}_s, \mathcal{Z}) - \bar{k}_{i,\min}(\boldsymbol{x}_s, \mathcal{Z}),
\label{eq:bandwidth_definition}
\end{equation}
where
\begin{equation}
\left\{
\begin{aligned}
    \bar{k}_{i,\max}(\boldsymbol{x}_s, \mathcal{Z}) &\triangleq 
    \max \bigl\{\, \bar{k}_{i}^{h} :
    H(\boldsymbol{k}^{h}; \boldsymbol{x}_s, \mathcal{Z}) \neq 0  \bigr\}, \\[2mm]
    \bar{k}_{i,\min}(\boldsymbol{x}_s, \mathcal{Z}) &\triangleq 
    \min \bigl\{\, \bar{k}_{i}^{h} :
    H(\bar{\boldsymbol{k}}^{h}; \boldsymbol{x}_s, \mathcal{Z}) \neq 0  \bigr\}.
\end{aligned}
\right.
\label{eq:local_spatial_frequency_difference}
\end{equation}
Although the exact spectral characterisation given by \eqref{eq:bandwidth_definition}-\eqref{eq:local_spatial_frequency_difference} is theoretically rigorous, it is of limited practical use in the NF regime, since the corresponding spectrum generally exhibits unbounded support \cite{npj}.

In \cref{fig:mono_aliased_vs_non_aliased}, the zoomed-in views (right column) show that increasing the array aperture (from top to bottom) improves resolution (i.e., the peak becomes sharper) owing to the larger spatial bandwidth captured. However, this improvement comes at the cost of increased aliasing artefacts, even when the antenna spacing is kept constant (for instance, between the second and third rows of that figure). This example highlights the fundamental trade-off between resolution and aliasing in NF localisation, and the need to jointly account for both in array design.

%%%%%%%%%%% Chirp Framework %%%%%%%%%%%

\section{Chirp-based framework}
\label{sec:chirp}

This section extends the chirp-based framework in \cite{npj} to include the local spatial-frequency characterisation required for resolution analysis, as well as to the derivation of aliasing conditions for multidimensional arrays. The extended framework enables to derive closed-form expressions for the maximum and minimum spatial frequencies, $\bar{k}_{i, max}$ and $\bar{k}_{i, min}$, used to determine the bandwidth $\bar{B}_i$ captured at the receive array and for the maximum spatial frequency $\bar{K}_i$, respectively.

\subsection{Chirp-based Modelling of the Received Field}
\label{sec:resolution_chirp}

A chirp is typically defined as a signal whose frequency varies over time \cite{flandrin_time_2001}, i.e. 
\begin{equation}
    c(t) = A \, e^{j \phi_c(t)} %= A \, e^{j (2\pi f_0 t + \pi \alpha t^2)},
\end{equation}
where \( A \) is the amplitude and $\phi_c(t)$ is the time-varying phase function. The instantaneous frequency \( f(t) \) of the chirp signal can then be derived from the phase \( \phi_c(t) \) as 
\begin{equation}
    f(t) = \frac{1}{2\pi} \frac{\mathrm{d}\phi_c(t)}{\mathrm{d}t}. %= f_0 + \alpha t.
\end{equation}

In our setting, the received field \( h(\boldsymbol{z}; \boldsymbol{x}_s) \) in \eqref{eq:signal} can be interpreted as a spatial chirp when considered as a function of the antenna position \( \boldsymbol{z} \) for a fixed source location \( \boldsymbol{x}_s \). Its phase
\begin{equation}
    \angle h(\boldsymbol{z}; \boldsymbol{x}_s) = \phi_h(\boldsymbol{z}; \boldsymbol{x}_s) \triangleq -k \| \boldsymbol{x}_s - \boldsymbol{z} \|,
\end{equation} 
induces a spatially non-linear phase component that depends on the distance between the source and the antenna element.
Accordingly, the local spatial chirp frequency \( \boldsymbol{\mathrm{k}}^{h}(\boldsymbol{z}; \boldsymbol{x}_s) \) of the received field \( h(\boldsymbol{z}; \boldsymbol{x}_s) \) at a given antenna position \( \boldsymbol{z} \) can be expressed as the gradient of the phase with respect to the spatial coordinates \cite{chassande-mottin_stationary_1998, ding_degrees_2022,ding_spatial_2024}, i.e.,
\begin{equation}
    \boldsymbol{\mathrm{k}}^{h}(\boldsymbol{z}; \boldsymbol{x}_s) \triangleq \nabla_{\boldsymbol{z}} \phi_h(\boldsymbol{z}; \boldsymbol{x}_s) = k \frac{\boldsymbol{x}_s - \boldsymbol{z}}{\| \boldsymbol{x}_s - \boldsymbol{z} \|}.
\label{eq:local_spatial_frequency_h}
\end{equation}
Unlike the usual spatial frequency $\bar{\boldsymbol{k}}^{h}$, which is well defined only for plane-wave propagation, $\boldsymbol{\mathrm{k}}^{h}(\boldsymbol{z};\boldsymbol{x}_s)$ describes a position-dependent, i.e. local, spatial frequency induced by the non-linear phase of a spherical wavefront. It therefore generalises the notion of wave vector by capturing the local, chirp-like behaviour of the steering signal in the near field, at a given antenna position $\boldsymbol{z}.$ Naturally, in the FF, the phase is linear in $\boldsymbol{z}$, so that the local spatial frequency becomes unique and equal to $\bar{\boldsymbol{k}}^{h}$.
Note that the amplitude term $1/\| \boldsymbol{x}_s - \boldsymbol{z} \|$ of $h(\boldsymbol{z}; \boldsymbol{x}_s)$ is neglected when determining the local spatial frequency. As motivated in \cite{npj}, the amplitude contribution is negligible since the phase exhibits a much faster non-linear variation compared to the slowly varying amplitude term.

Equation \eqref{eq:local_spatial_frequency_h} shows that the local spatial frequency vector \( \boldsymbol{\mathrm{k}}^{h}(\boldsymbol{z}; \boldsymbol{x}_s) \) defined at a given antenna position $\boldsymbol{z}$, has a magnitude equal to \( k = 2 \pi/ \lambda \) and points from the antenna toward the source. %The directionality of this vector reflects the fact that the phase variation is most pronounced along the line connecting the source and the antenna element.
\cref{fig:spectrum_resolution} illustrates those local spatial frequencies \( {\boldsymbol{\mathrm{k}}}^{h}(\boldsymbol{z}; \boldsymbol{x}_s) \) (red arrows) at different antenna positions \( \boldsymbol{z} \in \mathcal{Z} \) for a source arbitrarily located at \( \boldsymbol{x}_s = (0,40)\lambda \). As can be seen, the local spatial frequency varies across the array, with different antenna positions experiencing different phase gradients depending on their relative position to the source. Collectively, these vectors characterise the support of the spectrum \( H(\boldsymbol{k}^{h} ; \tilde{\boldsymbol{x}}_s, \boldsymbol{x}_s) \). 
%This variation in local spatial frequency is a key characteristic of the near-field regime, where the wavefront curvature leads to non-uniform phase variations across the array. 
The grey zone denotes the Non-Contributive Zone (NCZ), detailed in \cref{sec:ncz}.

\begin{figure}
    \centering 
    \includegraphics[width=0.5\linewidth, trim=50 20 50 30, clip]{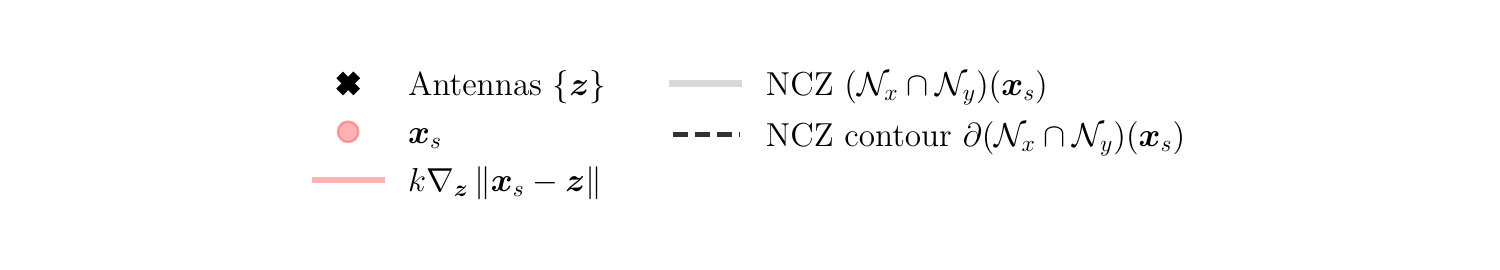} \\

    \includegraphics[width=0.35\linewidth, trim=10 0 0 5, clip]{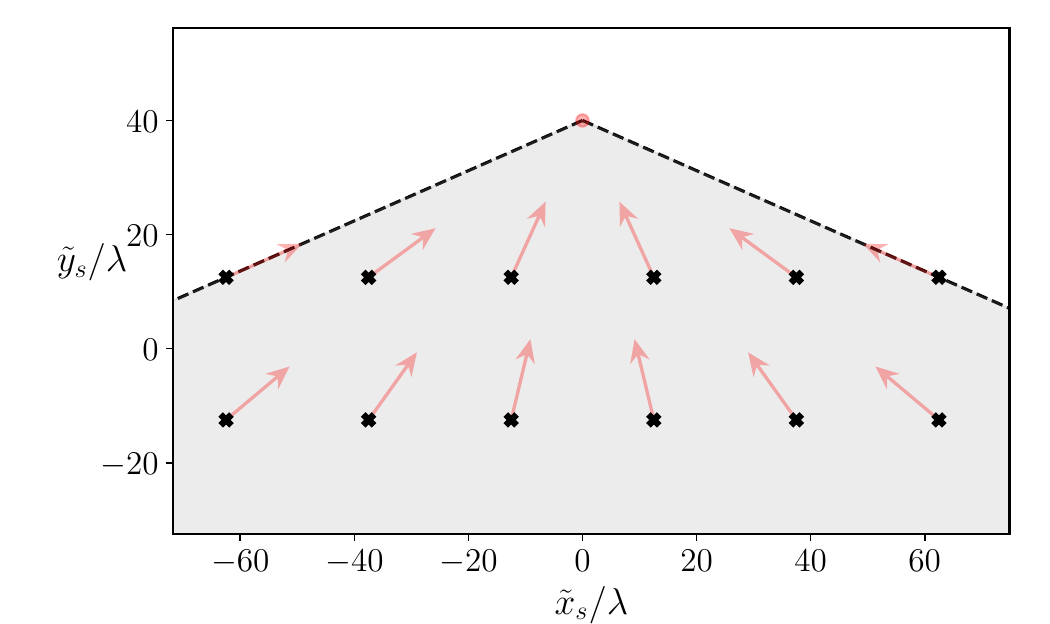}
    \caption{Local spatial frequencies \( \boldsymbol{\mathrm{k}}^{h}(\boldsymbol{z}; \boldsymbol{x}_s) \) of the received field \( h(\boldsymbol{z}; \boldsymbol{x}_s) \) at various antenna positions \( \boldsymbol{z} \) for a source at \( \boldsymbol{x}_s \).}
    \label{fig:spectrum_resolution}
\end{figure}

% Faire le lien avec la bande passante spatiale captée

Using this chirp interpretation, the maximum and minimum spatial frequencies \( \mathrm{k}_{i,\max}(\boldsymbol{x}_s, \mathcal{Z}) \) and \( \mathrm{k}_{i,\min}(\boldsymbol{x}_s, \mathcal{Z}) \) along axis~\( i \) follow from the extrema of the local chirp frequencies across the array \( \mathcal{Z} \) as 
\begin{equation}
\left\{
\begin{aligned}
    \mathrm{k}_{i,\max}(\boldsymbol{x}_s, \mathcal{Z}) &= \max_{\boldsymbol{z} \in \mathcal{Z}} \{ {\mathrm{k}}_{i}^{h}(\boldsymbol{z}; \boldsymbol{x}_s) \}, \\ %[2mm]
    \mathrm{k}_{i,\min}(\boldsymbol{x}_s, \mathcal{Z}) &= \min_{\boldsymbol{z} \in \mathcal{Z}} \{ {\mathrm{k}}_{i}^{h}(\boldsymbol{z}; \boldsymbol{x}_s) \}, 
\end{aligned}
\right.
\label{eq:ki_max_min_chirp}
\end{equation}
where $\mathrm{k}_{i}^{h}(\boldsymbol{z}; \boldsymbol{x}_s) = (\boldsymbol{\mathrm{k}}^{h}(\boldsymbol{z}; \boldsymbol{x}_s))_i$ and $ \boldsymbol{\mathrm{k}}^{h}(\boldsymbol{z}; \boldsymbol{x}_s)$ is computed as \eqref{eq:local_spatial_frequency_h}.
These expressions indicate that the maximum and minimum spatial frequencies along axis \( i \) can be determined by evaluating the local spatial chirp frequency \( \boldsymbol{\mathrm{k}}^{h}(\boldsymbol{z}; \boldsymbol{x}_s) \) at all antenna positions \( \boldsymbol{z} \) in the array \( \mathcal{Z} \) and selecting the extreme values. 
Substituting these expressions into \eqref{eq:bandwidth_definition} yields the chirp-based spatial bandwidth \( \mathrm{B}_i(\boldsymbol{x}_s, \mathcal{Z}) \), which directly determines the resolution via \eqref{eq:resolution}.

\cref{fig:chirp_spectrum_resolution} compares these chirp-derived extrema with the actual support of the spatial spectrum \( H(\bar{\boldsymbol{k}}^{h} ; \tilde{\boldsymbol{x}}_s, \boldsymbol{x}_s) \) for rectangular (left) and a circular (right) arrays. The source and array geometries are arbitrarily chosen. As can be seen, the maximum and minimum chirp frequencies (represented by the vertical and horizontal lines) closely match the actual support of the spectrum \( H(\boldsymbol{k}^{h} ; \tilde{\boldsymbol{x}}_s, \boldsymbol{x}_s) \) in both cases, supporting the relevance of the chirp model for characterising the spatial frequency content of the received field and, consequently, for predicting localisation resolution.

\begin{figure}
    \centering
    % 2x3 Grid with titles aligned with columns

    \includegraphics[width=0.2\linewidth, trim=0 70 20 70, clip]{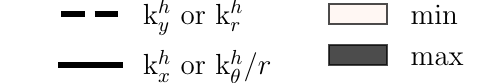}

    \begin{tabular}{cc}
        \tiny
        $\boldsymbol{x}_s = (0;10\lambda)$ & 
        \tiny
        ${\boldsymbol{x}}_s = (-10\lambda;10\lambda)$  \\

        \begin{subfigure}[t]{0.22\textwidth}
            \centering
            \includegraphics[width=\linewidth, trim=20 20 80 19, clip]{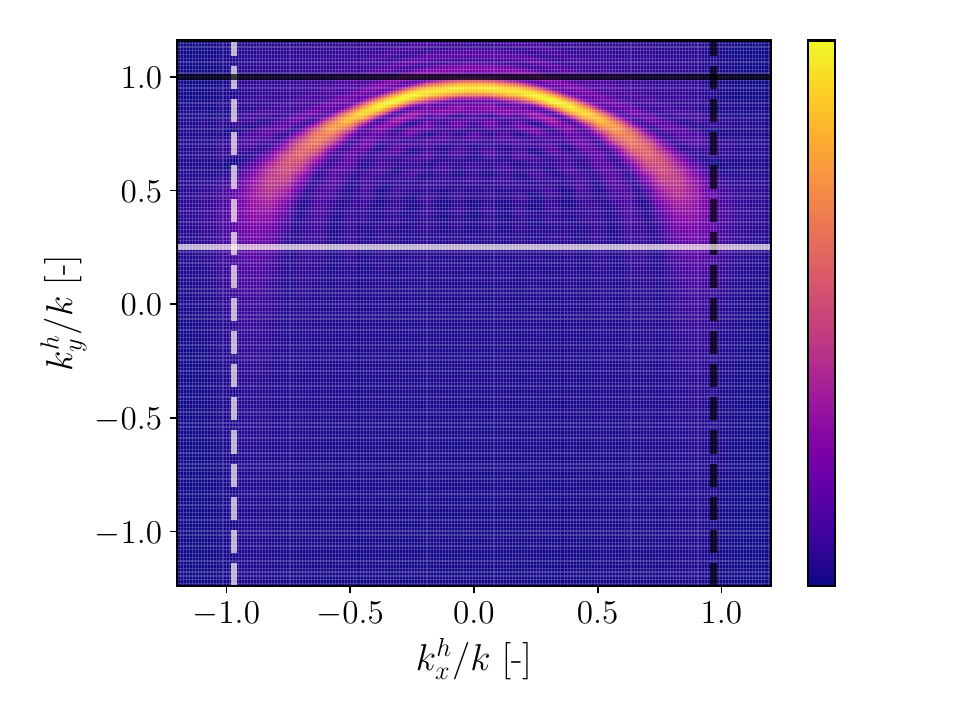}
            %\caption{Caption for figure 1.3}
            %\parbox{\linewidth}{\centering \tiny $R - \tilde{R}$ [m]}
            \caption{Rectangular array: $64\times64$ antennas, $D_x = D_y = 15\lambda$.}
        \label{fig:chirp_spectrum_resolution_cartesian}
        \end{subfigure} & 
        \begin{subfigure}[t]{0.24\textwidth}
            \centering
            \includegraphics[width=\linewidth, trim=20 20 50 19, clip]{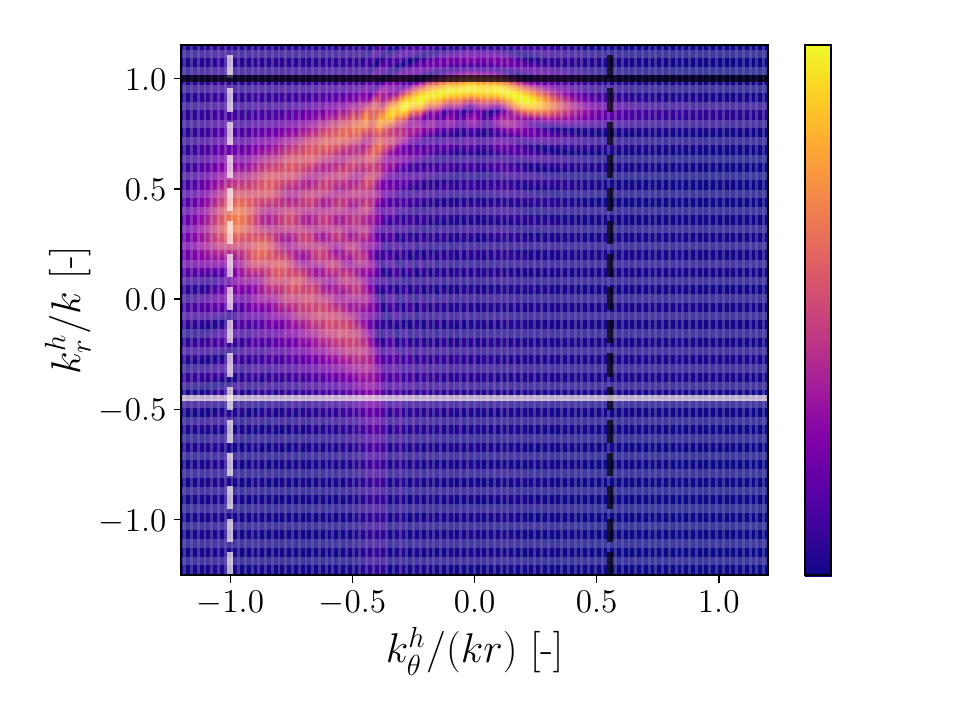}
            %\caption{Caption for figure 1.3}
            %\parbox{\linewidth}{\centering \tiny $R - \tilde{R}$ [m]}
            \caption{Circular array: $N_r=64$ rings (from $5\lambda$ to $15\lambda$), $N_\theta=256$ antennas per ring over a $180^\circ$ aperture.}
        \label{fig:chirp_spectrum_resolution_circular}
        \end{subfigure} 

    \end{tabular}

    \caption{Amplitude of the spectrum $H(\bar{\boldsymbol{k}}^{h} ; \tilde{\boldsymbol{x}}_s, \boldsymbol{x}_s)$ for two-dimensional rectangular and circular arrays. The maximum and minimum chirp frequencies along each axis, derived from \eqref{eq:ki_max_min_chirp}, are indicated by the vertical and horizontal lines.}
    \label{fig:chirp_spectrum_resolution}
\end{figure}

\subsection{Chirp-based Non-Aliasing Conditions}
\label{sec:non_aliasing_conditions_chirp}

Using a similar approach, closed-form expression for the maximum spatial frequency \( \bar{K}_i(\tilde{\boldsymbol{x}}_s, \boldsymbol{x}_s) \) of the ambiguity function $AF$ along each axis \( i \) at a test point $\tilde{\boldsymbol{x}}_s$, can be derived from the local spatial frequency of \( g(\boldsymbol{z} ; \tilde{\boldsymbol{x}}_s, \boldsymbol{x}_s) \). This function can itself be interpreted as a spatial chirp with phase
\begin{equation}
\angle g(\boldsymbol{z}; \tilde{\boldsymbol{x}}_s, \boldsymbol{x}_s) 
= \phi_g(\boldsymbol{z}; \tilde{\boldsymbol{x}}_s, \boldsymbol{x}_s) 
\triangleq \phi_h(\boldsymbol{z}; \boldsymbol{x}_s) 
- \phi_h(\boldsymbol{z}; \tilde{\boldsymbol{x}}_s). 
\end{equation}

Accordingly, the local spatial chirp frequency \( \boldsymbol{\mathrm{k}}^{g}(\boldsymbol{z} ; \tilde{\boldsymbol{x}}_s, \boldsymbol{x}_s) \) of \( g(\boldsymbol{z} ; \tilde{\boldsymbol{x}}_s, \boldsymbol{x}_s) \) at an antenna position \( \boldsymbol{z} \) can be expressed as
\begin{equation}
\begin{aligned}
    \boldsymbol{\mathrm{k}}^{g}(\boldsymbol{z} ; \tilde{\boldsymbol{x}}_s, \boldsymbol{x}_s) 
    &\triangleq \nabla_{\boldsymbol{z}} \phi_g(\boldsymbol{z} ; \tilde{\boldsymbol{x}}_s, \boldsymbol{x}_s) \\
    &= \boldsymbol{\mathrm{k}}^{h}(\boldsymbol{z}; \boldsymbol{x}_s) - \boldsymbol{\mathrm{k}}^{h}(\boldsymbol{z}; \tilde{\boldsymbol{x}}_s) \\
    &= k \left( \frac{\boldsymbol{x}_s - \boldsymbol{z}}{\| \boldsymbol{x}_s - \boldsymbol{z} \|} - \frac{\tilde{\boldsymbol{x}}_s - \boldsymbol{z}}{\| \tilde{\boldsymbol{x}}_s - \boldsymbol{z} \|} \right)
\end{aligned}
\label{eq:local_spatial_frequency_g}
\end{equation}
where \( \boldsymbol{\mathrm{k}}^{h}(\boldsymbol{z}; \boldsymbol{x}_s) \) and \( \boldsymbol{\mathrm{k}}^{h}(\boldsymbol{z}; \tilde{\boldsymbol{x}}_s) \) are given by \eqref{eq:local_spatial_frequency_h}. At any position $\boldsymbol{z}$, the local chirp spatial frequency \( \boldsymbol{\mathrm{k}}^{g}(\boldsymbol{z} ; \tilde{\boldsymbol{x}}_s, \boldsymbol{x}_s) \) is simply the difference between the local spatial frequencies of the received fields from the source \( \boldsymbol{x}_s \) and the tentative location \( \tilde{\boldsymbol{x}}_s \). Each term corresponds to a local vector of norm $k=2\pi/\lambda$ pointing in the direction between the antenna element \( \boldsymbol{z} \) and the points \( \tilde{\boldsymbol{x}}_s \) and \( \boldsymbol{x}_s \), respectively. 

\begin{figure}
    \centering
    \includegraphics[width=0.5\linewidth, trim=50 0 60 0, clip]{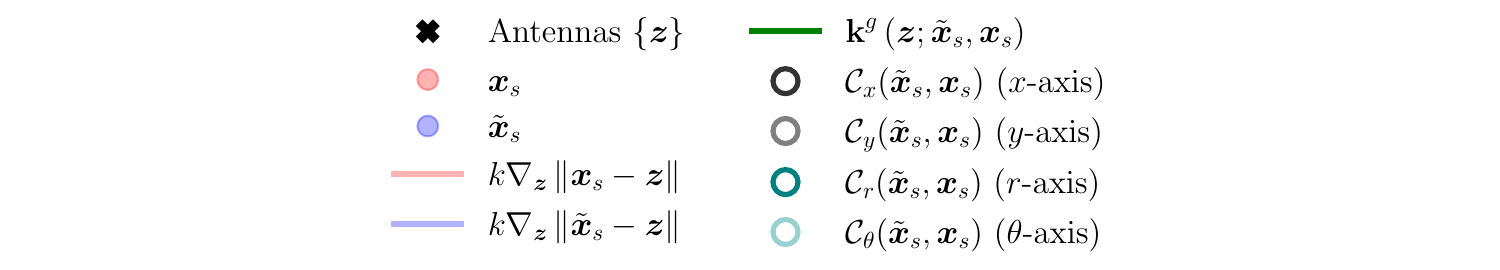}

    \begin{subfigure}{0.4\linewidth}
        \centering
        \includegraphics[width=0.9\linewidth, trim=15 20 0 20, clip]{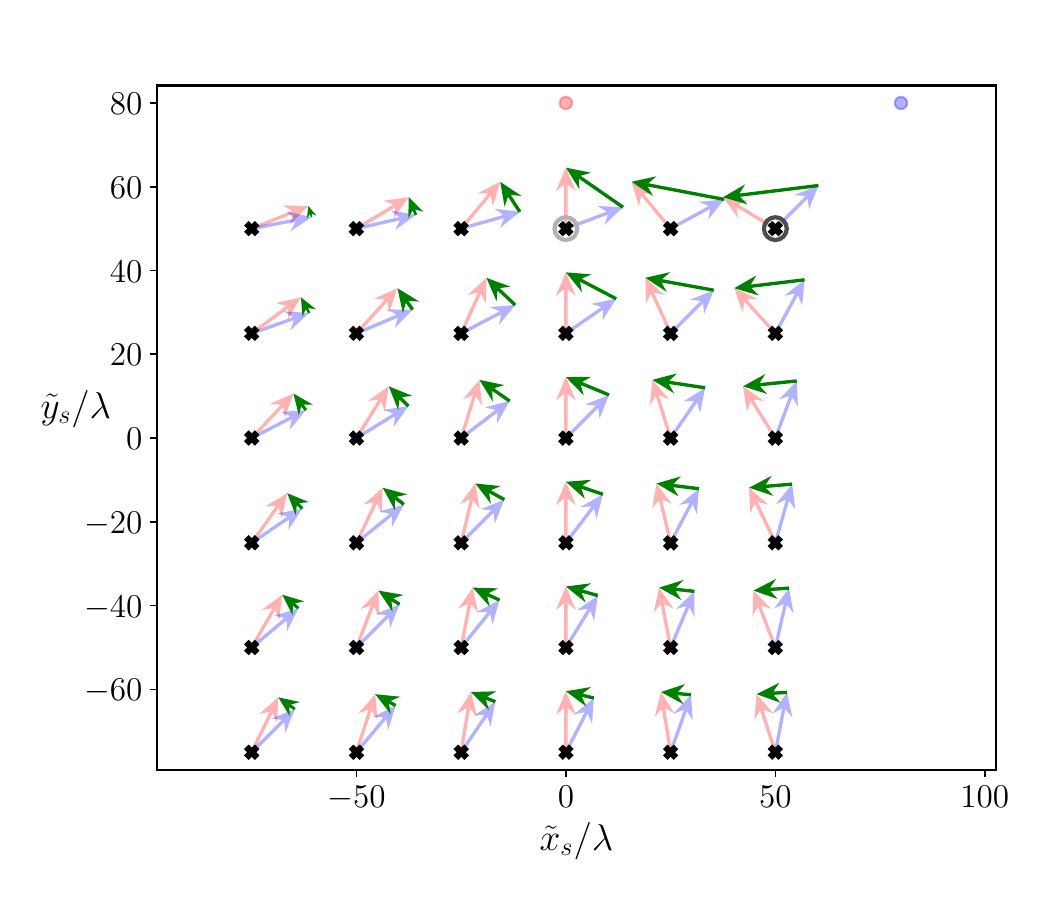}
        \caption{Rectangular geometry.}
        \label{fig:chirp_illustration_cartesian}
    \end{subfigure}

    \begin{subfigure}{0.4\linewidth}
        \centering
        \includegraphics[width=0.9\linewidth, trim=22 20 0 20, clip]{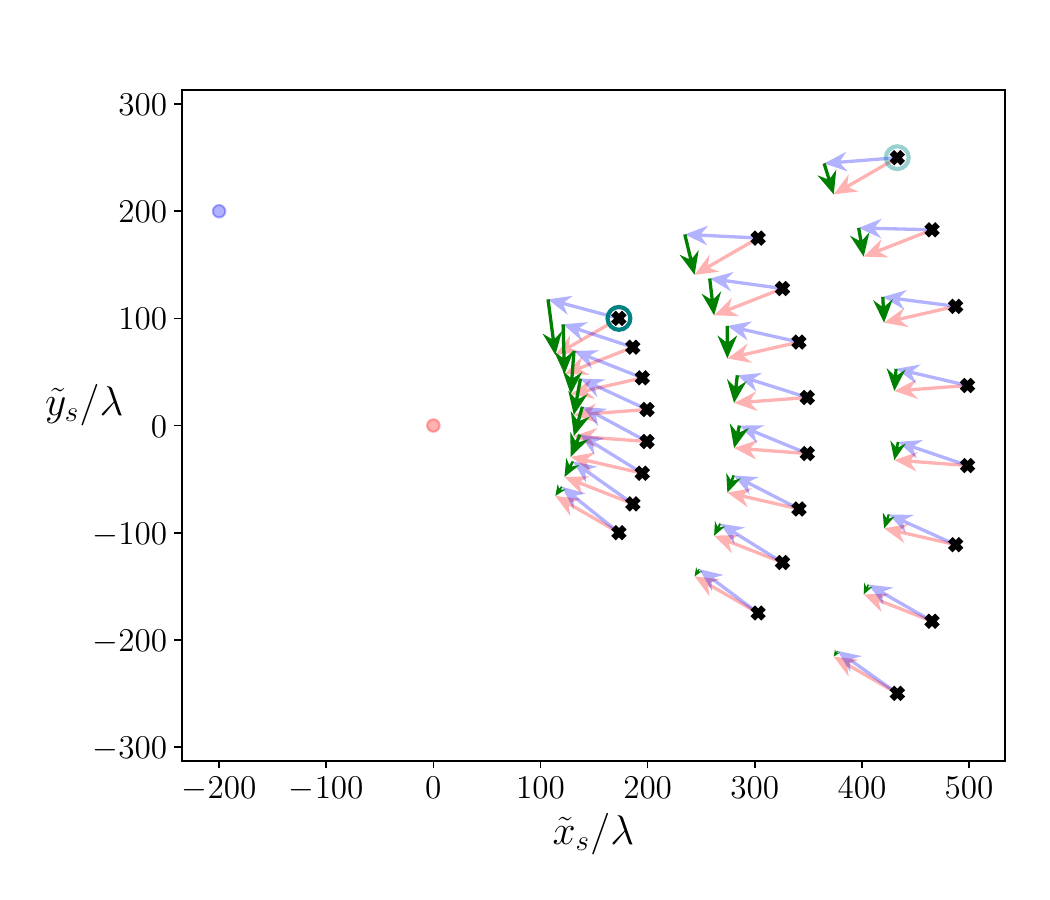}
        \caption{Circular geometry.}
        \label{fig:chirp_illustration_polar}
    \end{subfigure}

    \caption{Illustration of the two terms of the local wavenumber vectors \( \boldsymbol{\mathrm{k}}^{g}(\tilde{\boldsymbol{x}}_s, \boldsymbol{x}_s) \) at the antenna elements \( \boldsymbol{z} \) for a given point source \( \boldsymbol{x}_s \) and tentative location \( \tilde{\boldsymbol{x}}_s \).}
    \label{fig:chirp_illustration}
\end{figure}

Figure~\ref{fig:chirp_illustration} illustrates this for rectangular (Figure~\ref{fig:chirp_illustration_cartesian}) and circular (Figure~\ref{fig:chirp_illustration_polar}) arrays with arbitrary source $\boldsymbol{x}_s$ and tentative location $\tilde{\boldsymbol{x}}_s$. The vectors \( k \nabla_{\boldsymbol{z}} \|  \tilde{\boldsymbol{x}}_s - \boldsymbol{z} \| \) (in blue) and \( k \nabla_{\boldsymbol{z}} \| \boldsymbol{x}_s - \boldsymbol{z} \| \) (in red) at each antenna position $\{ \boldsymbol{z} \in \mathcal{Z} \}$ effectively point towards the points \( \tilde{\boldsymbol{x}}_s \) and \( \boldsymbol{x}_s \), respectively. 
Their difference yields the local wavenumber vector $\boldsymbol{\mathrm{k}}^{g}(\boldsymbol{z} ; \tilde{\boldsymbol{x}}_s, \boldsymbol{x}_s)$ (in green), which representing the spatial frequency content of \( g(\boldsymbol{z} ; \tilde{\boldsymbol{x}}_s, \boldsymbol{x}_s) \) at the element location \( \boldsymbol{z} \). %This visualisation provides an intuitive framework that will be used for analysing and interpreting the results presented in the following sections. 
Circled antennas denote Critical Antenna Elements (CAEs); their role and formal definition will be provided in Section~\ref{sec:critical_antennas}.

Using this chirp-based model, the chirp maximum spatial frequency \( \mathrm{K}_i(\tilde{\boldsymbol{x}}_s, \boldsymbol{x}_s) \) along axis \( i \) is
\begin{equation}
    \mathrm{K}_i(\tilde{\boldsymbol{x}}_s, \boldsymbol{x}_s) = \max_{\boldsymbol{z} \in \mathcal{Z}} \{ {\mathrm{k}}^{g}_{i}(\boldsymbol{z} ; \tilde{\boldsymbol{x}}_s, \boldsymbol{x}_s) \},  
\label{eq:Ki_chirp}
\end{equation}
where \( \boldsymbol{\mathrm{k}}^{g}(\boldsymbol{z} ; \tilde{\boldsymbol{x}}_s, \boldsymbol{x}_s) \) is defined in \eqref{eq:local_spatial_frequency_g}.

% As soon as (22) [...] meets (9) along each dimension i, we assume there is no aliasing. 
As soon as the maximum spatial chirp frequency $\mathrm{K}_i$ in \eqref{eq:Ki_chirp} satisfies the non-aliasing condition \eqref{eq:aliasing_conditions} along each dimension $i$, we assume that no aliasing artefact occurs at the tentative location $\tilde{\boldsymbol{x}}_s$. 
In \cite{npj}, a similar reasoning led to concept of the Aliasing-Free Region (AFR) $\mathcal{S}(\boldsymbol{x}_s)$, which can be redefined here as in Definition~\ref{def:afr}. In this reformulation, $\mathcal{S}(\boldsymbol{x}_s)$ explicitly corresponds to the intersection of the directional AFRs $\mathcal{S}_i(\boldsymbol{x}_s)$ along each axis $i$.

\begin{definition}\label{def:afr}
The Aliasing-Free Region (AFR) $\mathcal{S}(\boldsymbol{x}_s)$ for a monostatic array is given by
\begin{equation}
    \mathcal{S}(\boldsymbol{x}_s) \triangleq \bigcap_{i} \mathcal{S}_i(\boldsymbol{x}_s), 
\end{equation}
where the directional AFR $\mathcal{S}_i(\boldsymbol{x}_s)$ is defined as
\begin{equation}
    \mathcal{S}_i(\boldsymbol{x}_s) \triangleq \{ \tilde{\boldsymbol{x}}_s \in \mathbb{R}^{d} \, | \, \mathrm{K}_{i}(\tilde{\boldsymbol{x}}_s, \boldsymbol{x}_s) \le 2\pi/\Delta_{i} \}. 
\label{eq:def_afr}
\end{equation}
\end{definition}

The AFR is the subset of candidate points \( \tilde{\boldsymbol{x}}_s \) in the $AF$ that satisfy \eqref{eq:aliasing_conditions} and thus remain unaffected by aliasing. 
This formalism precisely identifies the support of the ambiguity function where no aliasing artefacts occur. Moreover, if the true source $\boldsymbol{x}_s$ is known to lie within a given search region, 
the AFR can be evaluated over all possible source points in that region to guarantee aliasing-free reconstruction throughout that entire region.

In \cref{fig:chirp_spectrum}, the maximum spatial frequencies \(\mathrm{K}_i(\tilde{\boldsymbol{x}}_s, \boldsymbol{x}_s)\) from \eqref{eq:Ki_chirp} along each axis are compared with the actual spectrum \(G(\bar{\boldsymbol{k}}^{g} ; \tilde{\boldsymbol{x}}_s, \boldsymbol{x}_s)\) and represented by the vertical and horizontal lines, for given $(\tilde{\boldsymbol{x}}_s;\boldsymbol{x}_s)$ pairs. 
The same two array geometries as in \cref{fig:chirp_spectrum_resolution} are considered: rectangular (\cref{fig:chirp_spectrum_cartesian}) and circular (\cref{fig:chirp_spectrum_circular}) arrays. As can be seen, the maximum chirp frequencies closely match the actual ones of the spectrum \( G(\bar{\boldsymbol{k}}^{g} ; \tilde{\boldsymbol{x}}_s, \boldsymbol{x}_s) \) in both cases. This supports the chirp-based definition of the spatial frequency content of the ambiguity function matched at a given test position $\tilde{\boldsymbol{x}}_s$ and, consequently, the non-aliasing conditions at this point. 

\begin{figure}
    \centering
    % 2x3 Grid with titles aligned with columns

    \includegraphics[width=0.2\linewidth, trim=0 70 20 70, clip]{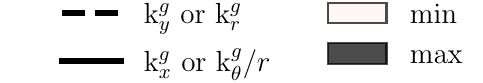}

    \begin{tabular}{cc}
        \tiny
        $\boldsymbol{x}_s = (0;10\lambda)$, $\tilde{\boldsymbol{x}}_s = (10\lambda;10\lambda)$  & 
        \tiny
        $\boldsymbol{x}_s = (-10;10\lambda)$, $\tilde{\boldsymbol{x}}_s = (0\lambda;5\lambda)$  \\

        \begin{subfigure}[t]{0.22\textwidth}
            \centering
            \includegraphics[width=\linewidth, trim=20 20 80 19, clip]{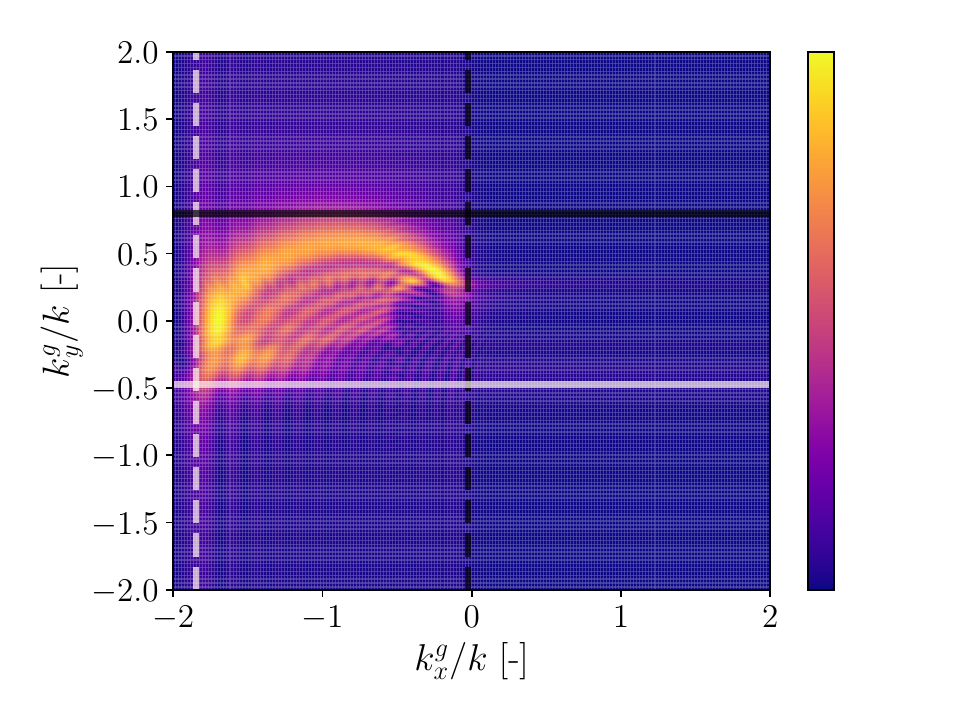}
            %\caption{Caption for figure 1.3}
            %\parbox{\linewidth}{\centering \tiny $R - \tilde{R}$ [m]}
            \caption{Rectangular array: $64\times64$ antennas, $D_x = D_y = 15\lambda$.}
        \label{fig:chirp_spectrum_cartesian}
        \end{subfigure} & 
        \begin{subfigure}[t]{0.24\textwidth}
            \centering
            \includegraphics[width=\linewidth, trim=20 20 50 19, clip]{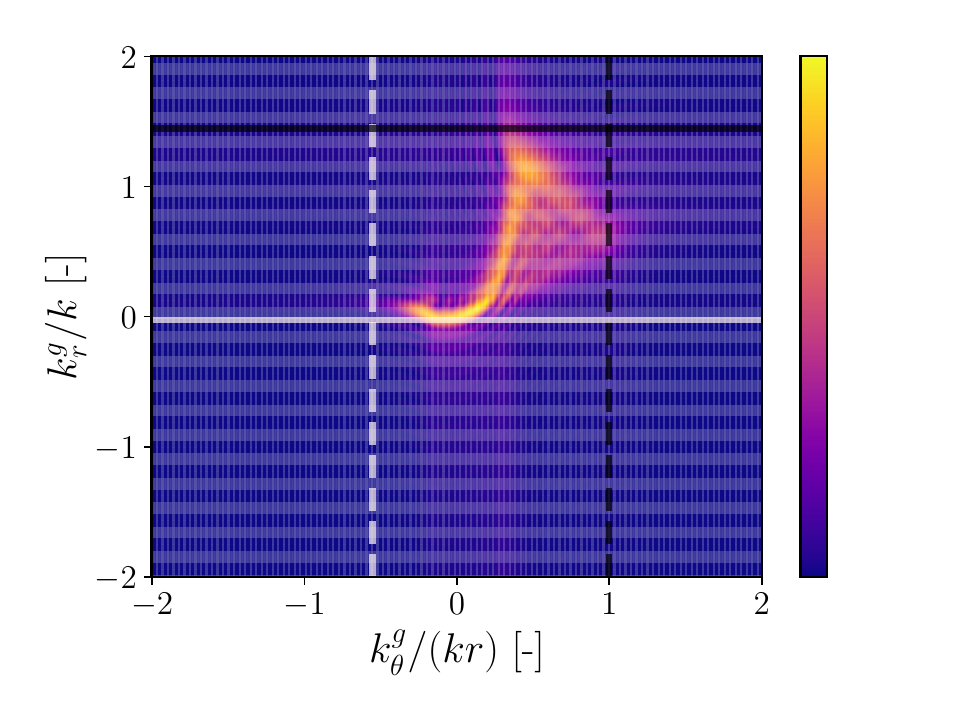}
            %\caption{Caption for figure 1.3}
            %\parbox{\linewidth}{\centering \tiny $R - \tilde{R}$ [m]}
            \caption{Circular array: $N_r=64$ rings (from $5\lambda$ to $15\lambda$), $N_\theta=256$ antennas per ring over a $180^\circ$ aperture.}
        \label{fig:chirp_spectrum_circular}
        \end{subfigure} 

    \end{tabular}

    \caption{Amplitude of the spectrum $G(\bar{\boldsymbol{k}}^{g} ; \tilde{\boldsymbol{x}}_s, \boldsymbol{x}_s)$ for two-dimensional rectangular and circular antenna arrays. Vertical and horizontal lines indicate the maximum and minimum chirp frequencies along each axis, as derived from \eqref{eq:Ki_chirp}.}
    \label{fig:chirp_spectrum}
\end{figure}

To summarise, \eqref{eq:local_spatial_frequency_g} and \eqref{eq:Ki_chirp} provide an effective approximation of the maximum spatial frequency \( \bar{K}_{i}(\tilde{\boldsymbol{x}}_s, \mathcal{S}) \). Combined with the aliasing conditions in \eqref{eq:aliasing_conditions}, this enables practical identification of ambiguity-function locations \( \tilde{\boldsymbol{x}}_s \) that are affected by aliasing and those that are not.

%%%%%%%%%%% Geometrical Tools %%%%%%%%%%%

\section{Geometrical Tools}
\label{sec:geometrical_tools}

This section introduces geometrical tools used to facilitate the analysis and interpretation of the results presented in the following sections. These tools include the Non-Contributive Zone (NCZ) (\cref{sec:ncz}), which identifies regions where antennas do not contribute to the localisation resolution, and the Critical Antenna Elements (CAEs) (\cref{sec:critical_antennas}), which correspond to antennas associated with the maximum spatial frequency \( \mathrm{K}_i(\tilde{\boldsymbol{x}}_s, \boldsymbol{x}_s) \) and thus influence the non-aliasing conditions. 
Both tools are presented separately, but their interplay is illustrated in Sections~\ref{sec:point_scatterer} and \ref{sec:polar_arrays}. These concepts do not strictly coincide, e.g. an antenna can be in the NCZ without being critical, leading to trade-off situations analysed empirically in these sections.

\subsection{Non-Contributive Zone}
\label{sec:ncz}

In \cref{sec:resolution_chirp}, it was shown that, for a given array $\mathcal{Z}$, the spatial bandwidth \( \mathrm{B}_i(\boldsymbol{x}_s, \mathcal{Z}) \) is determined by the extrema of the local spatial frequencies \( \boldsymbol{\mathrm{k}}^{h}(\boldsymbol{z}; \boldsymbol{x}_s) \) along axis \( i \), attained at specific antenna positions within the array. Antennas with intermediate frequencies do not influence the bandwidth nor, consequently, the resolution, and are thus non-contributive. \cref{def:non-contributive_region} formally defines the Non-Contributive Zone along axis \( i \) as the set of antenna positions \( \boldsymbol{z} \) where adding an antenna does not improve the resolution/increase the bandwidth along that axis.

\begin{definition}[Non-contributive Zone]\label{def:non-contributive_region}
    The Non-Contributive Zone (NCZ) along the \( i \)-th axis, noted $\mathcal{N}_{i}$,  is defined as
    \begin{equation}
    \mathcal{N}_{ i}(\boldsymbol{x}_s, \mathcal{Z}) \triangleq \left\{ \boldsymbol{z} : \mathrm{B}_i(\boldsymbol{x}_s, \mathcal{Z} ) = \mathrm{B}_i(\boldsymbol{x}_s, \mathcal{Z} \cup \left\{ \boldsymbol{z} \right\} )\right\}. 
\end{equation}
\end{definition}

In \cref{fig:spectrum_resolution}, the intersection of the NCZs along the $x$- and $y$-axes is shown as the grey area with a black contour. Antennas within this zone do not improve resolution in any direction, whereas antennas outside it are contributive, increasing the bandwidth along at least one axis. The concept of NCZ thus provides a practical guideline for placing antennas to enhance local resolution for a given source \( \boldsymbol{x}_s \).

\subsection{Critical Antenna Elements}
\label{sec:critical_antennas}

The chirp-based framework introduced in \cref{sec:non_aliasing_conditions_chirp} links the maximum spatial frequency \( \mathrm{K}_i(\tilde{\boldsymbol{x}}_s, \boldsymbol{x}_s) \) of the ambiguity function $AF$ at a test point $\tilde{\boldsymbol{x}}_s$, to a maximisation problem among the array elements \( \mathcal{Z} \). This enables the derivation of several properties of the aliasing conditions.
In particular, the inclusion principle from \cite{npj}, stated in \cref{prop:inclusion_principle_2}, provides general insights into how adding an antenna affects the AFR.

\begin{proposition}[Inclusion Principle]
    Given two antenna arrays \( \mathcal{Z}^{(1)} \subseteq \mathcal{Z}^{(2)} \) with the same inter-element spacings \( \Delta_{i} \) along each dimension \( i \), the corresponding aliasing-free regions, $\mathcal{S}^{(1)}$ and $\mathcal{S}^{(2)}$ satisfy
    \begin{equation}
        \mathcal{S}^{(2)}(\boldsymbol{x}_s) \subseteq \mathcal{S}^{(1)}(\boldsymbol{x}_s),
    \end{equation}
    for all sources locations \( \boldsymbol{x}_s \in \mathbb{R}^{d}\).
\label{prop:inclusion_principle_2}
\end{proposition}

%\begin{proof}
%\gmcom{Pas besoin de le démontrer puisque déjà démontré dans NPJ que tu cites. Par ailleurs, pas sur que proposition soit le bon terme, puisque déjà démontré. Proposition = Sera démontré + tard il me semble. A check}
%\gmdel{Since \( \mathcal{Z}^{(1)} \subseteq \mathcal{Z}^{(2)} \), the maximisation problem \eqref{eq:Ki_chirp} $\mathrm{K}_i^{(2)}(\tilde{\boldsymbol{x}}_s, \boldsymbol{x}_s)$ is a relaxed version of that yielding $\mathrm{K}_i^{(2)}(\tilde{\boldsymbol{x}}_s, \boldsymbol{x}_s)$. It follows that $ 
%    \mathrm{K}_i^{(1)}(\tilde{\boldsymbol{x}}_s, \boldsymbol{x}_s) \leq \mathrm{K}_i^{(2)}(\tilde{\boldsymbol{x}}_s, \boldsymbol{x}_s),
%$
%\( \forall \tilde{\boldsymbol{x}}_s, \boldsymbol{x}_s \in \mathbb{R}^{d} \) and \( \forall i  \). Consequently, the aliasing conditions in \eqref{eq:aliasing_conditions} imply that if \( \tilde{\boldsymbol{x}}_s \in \mathcal{S}^{(2)}(\boldsymbol{x}_s) \), then \( \tilde{\boldsymbol{x}}_s \in \mathcal{S}^{(1)}(\boldsymbol{x}_s) \), leading to the conclusion that \( \mathcal{S}^{(2)}(\boldsymbol{x}_s) \subseteq \mathcal{S}^{(1)}(\boldsymbol{x}_s) \).}\
%\end{proof}

This proposition shows that adding antennas (without changing the inter-element spacings) cannot enlarge the AFR; it can only preserve or shrink it. This is because adding antennas may increase the maximum spatial frequency \( \mathrm{K}_i(\tilde{\boldsymbol{x}}_s, \boldsymbol{x}_s) \), thereby tightening the aliasing conditions in \eqref{eq:aliasing_conditions}. It also implicitly suggests that removing antennas is always acceptable for increasing the AFR. Yet, it does not account for the loss of localisation resolution induced by removing antennas, nor does it indicate which antennas should be removed to maximise the AFR while minimally degrading resolution, or under which conditions the AFR remains unchanged or decreases when antennas are added.
To fill these gaps, \cref{def:critical_antenna_element} formally defines the Critical Antenna Elements as the specific antenna elements associated with the maximum spatial frequency \( \mathrm{K}_i(\tilde{\boldsymbol{x}}_s, \boldsymbol{x}_s) \) along axis \( i \) at a given test point \( \tilde{\boldsymbol{x}}_s \). 

\begin{definition}[Critical Antenna Elements]\label{def:critical_antenna_element}
    The Critical Antenna Elements (CAEs) are defined as the set
    \begin{equation}
    \mathcal{C}_{ i} (\tilde{\boldsymbol{x}}_s, \boldsymbol{x}_s) \triangleq
    \argmax_{\boldsymbol{z} \in \mathcal{Z}} \left| \mathrm{k}_{i}^{g}(\boldsymbol{z} ; \tilde{\boldsymbol{x}}_s, \boldsymbol{x}_s) \right|, 
    \label{eq:critical_antenna_element}
    \end{equation}
i.e. the subset of antenna positions $\boldsymbol{z} \in \mathcal{Z}$ where the spatial frequency $\mathrm{k}_i^{g}$ reaches its maximum magnitude.

For a set of candidate points $\tilde{\boldsymbol{x}}_s \in \tilde{\mathcal{X}}_s$, the corresponding CAEs are defined as
\begin{equation}
\mathcal{C}_{i}(\tilde{\mathcal{X}}_s, {\boldsymbol{x}}_s) \triangleq
\bigcup_{\tilde{\boldsymbol{x}}_s \in \tilde{\mathcal{X}}}
\mathcal{C}_{i} (\tilde{\boldsymbol{x}}_s, \boldsymbol{x}_s).
\end{equation}
\end{definition}

Identifying those CAEs is important as they indicate where the maximum spatial frequency along a given direction, namely $\mathrm{K}_i$ in the non-aliasing condition in \eqref{eq:aliasing_conditions}, occurs within the array. 
In particular, determining 
$\mathcal{C}_{ i}(\partial \mathcal{S}(\boldsymbol{x}_s), {\boldsymbol{x}}_s)$, i.e. the CAEs associated with points on the AFR boundary $\partial \mathcal{S}$ where $\mathrm{K}_i(\tilde{\boldsymbol{x}}_s, \boldsymbol{x}_s) = 2\pi/\Delta_i$, directly identifies the antennas that constrain the AFR size.
This concept thus provides a geometric link between the non-aliasing condition and the array geometry, and will be used in the next sections to analyse the influence of key system parameters.

The critical antennas along the $x$- and $y$-axes are highlighted in Figure~\ref{fig:chirp_illustration_cartesian} for the rectangular array geometry. They visibly correspond to the antennas with the greatest spatial variation in the corresponding directions--vertical for $x$ and horizontal for $y$--among the set of local local wavenumber vectors \( \boldsymbol{\mathrm{k}}^{g}(\boldsymbol{z} ; \tilde{\boldsymbol{x}}_s, \boldsymbol{x}_s) \) (shown in green).  
For the circular geometry in Figure~\ref{fig:chirp_illustration_polar}, the critical antennas correspond to two distinct criteria: one exhibiting the largest angular discrepancy between the local wavenumber vectors \( \boldsymbol{\mathrm{k}}^{g}(\boldsymbol{z}; \tilde{\boldsymbol{x}}_s, \boldsymbol{x}_s) \) and the angular direction of the antenna location \( \boldsymbol{z} \) from the origin (critical antenna associated with the angular coordinate \( \theta \)), and the other one where the wavenumber vector exhibits the largest norm (corresponding to the radial direction~\( r \)).

\begin{figure*}
    \centering

    \includegraphics[width=0.9\linewidth, trim=0 40 0 50, clip]{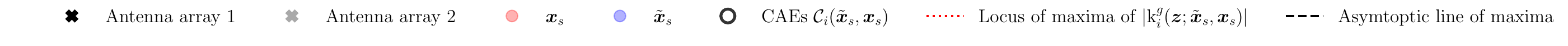}

    % --- Graphe de gauche ---
    \begin{subfigure}{0.18\linewidth}
        \centering
        \includegraphics[width=\linewidth, trim=60 50 60 50, clip]{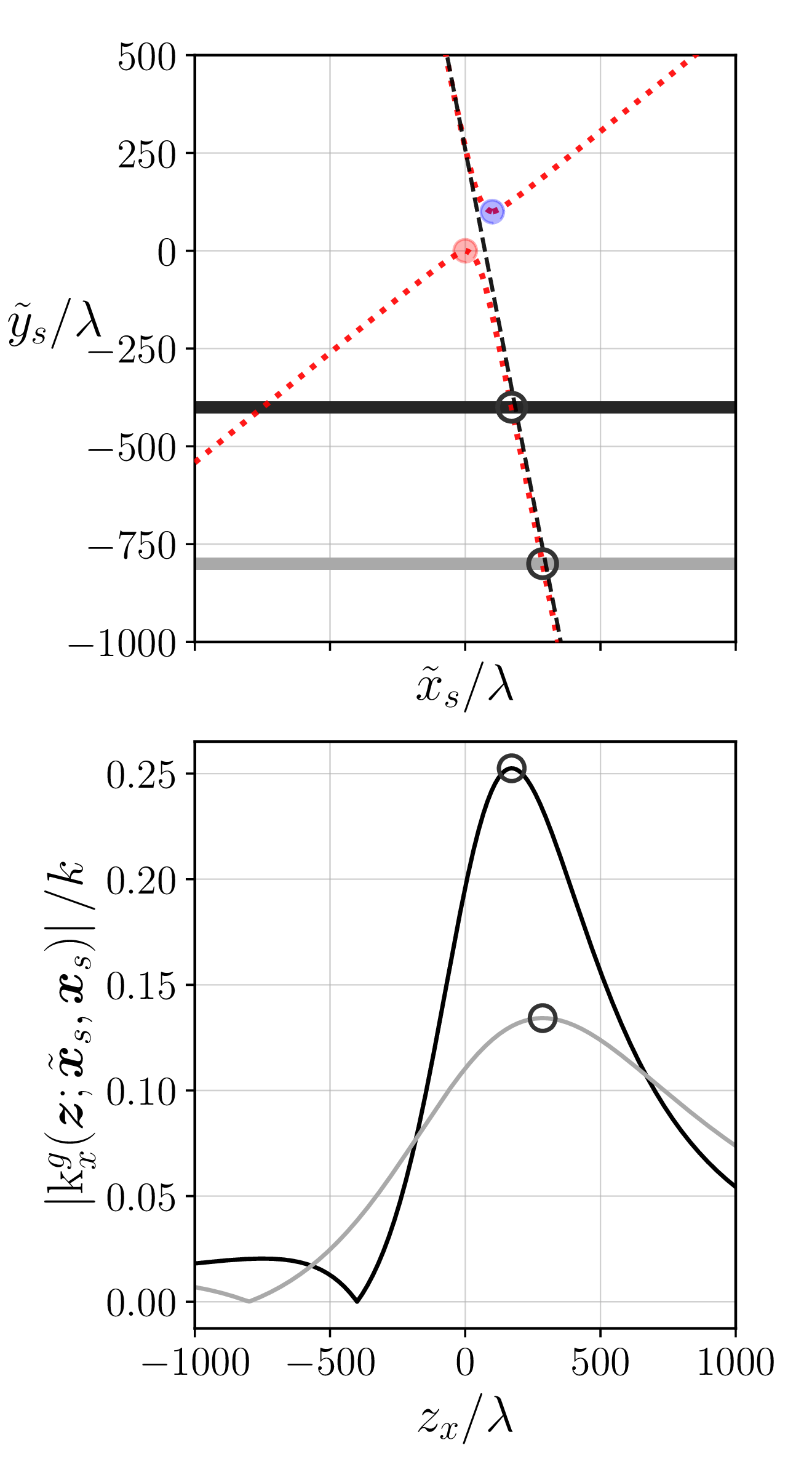}
        \caption{Rectangular array.} 
    \label{fig:critical_antenna_location_cartesian}
    \end{subfigure}
    \hfill
    % --- Deux graphes regroupés ---
    \begin{subfigure}{0.8\linewidth}
        \centering
        \begin{subfigure}{0.235\linewidth}
            \centering
            \includegraphics[width=\linewidth, trim=0 0 40 10, clip]{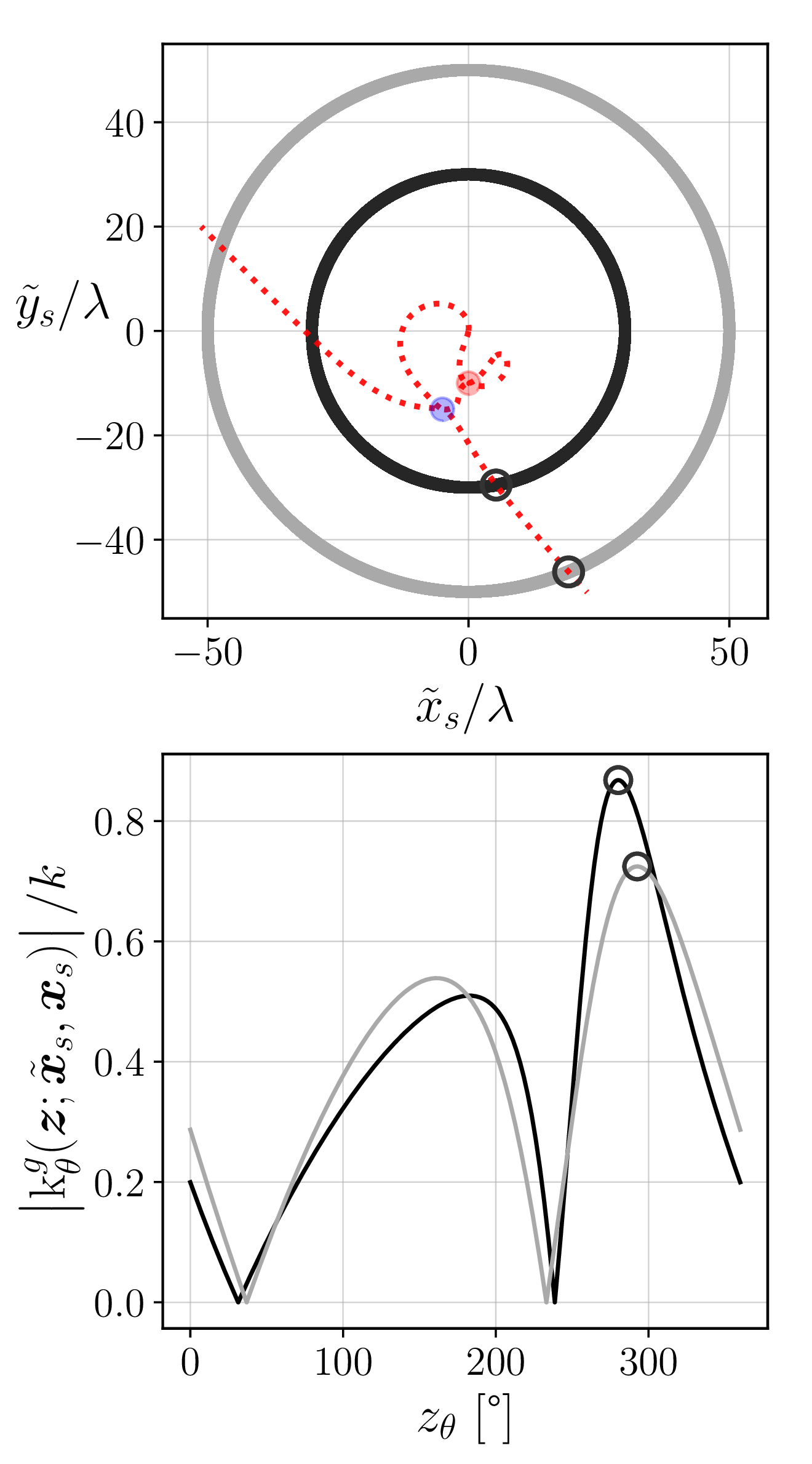}
        \end{subfigure}
        %\hfill
        \begin{subfigure}{0.235\linewidth}
            \centering
            \includegraphics[width=\linewidth, trim=0 0 40 10, clip]{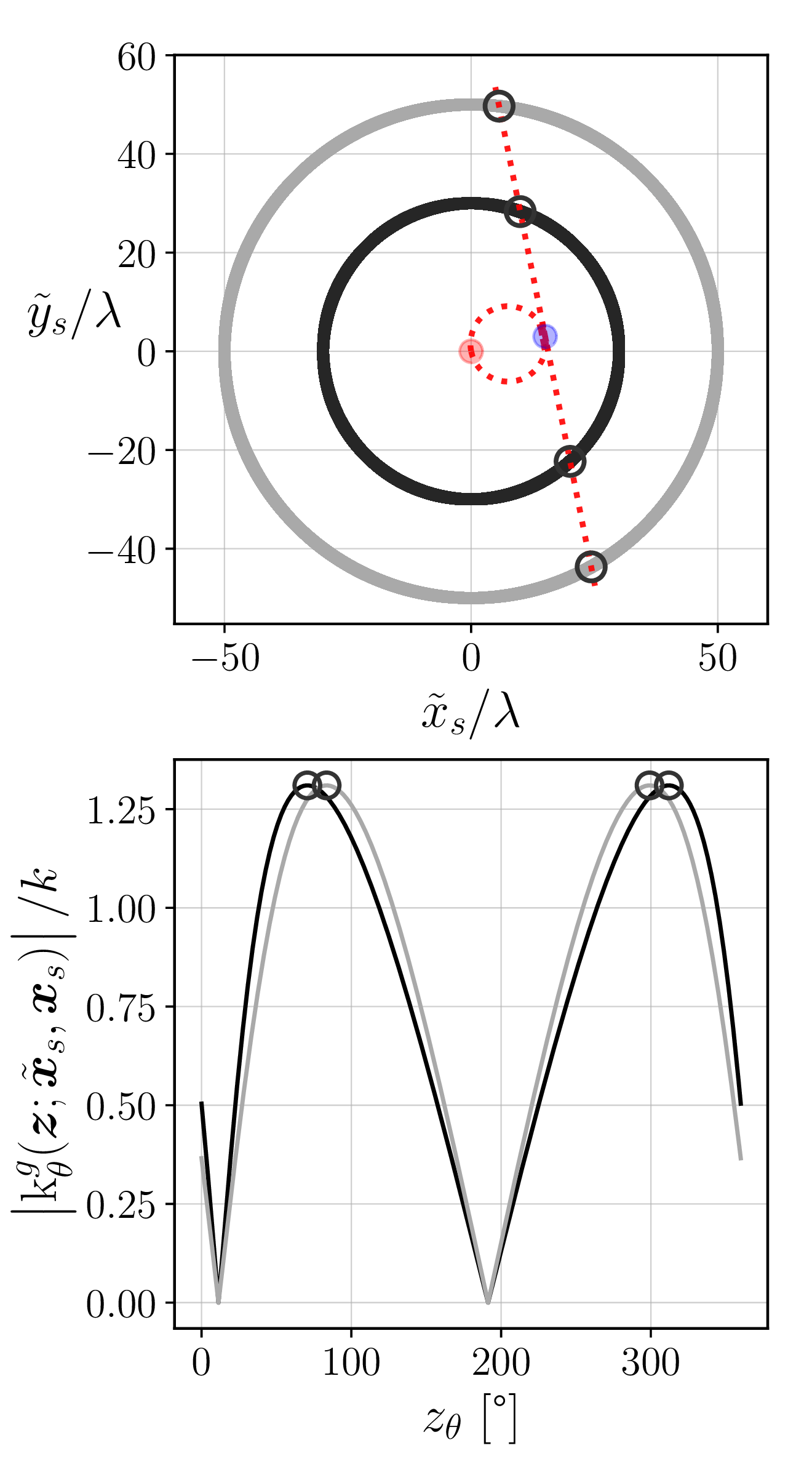}
        \end{subfigure}
        %\hfill
        \begin{subfigure}{0.235\linewidth}
            \centering
            \includegraphics[width=\linewidth, trim=0 0 40 10, clip]{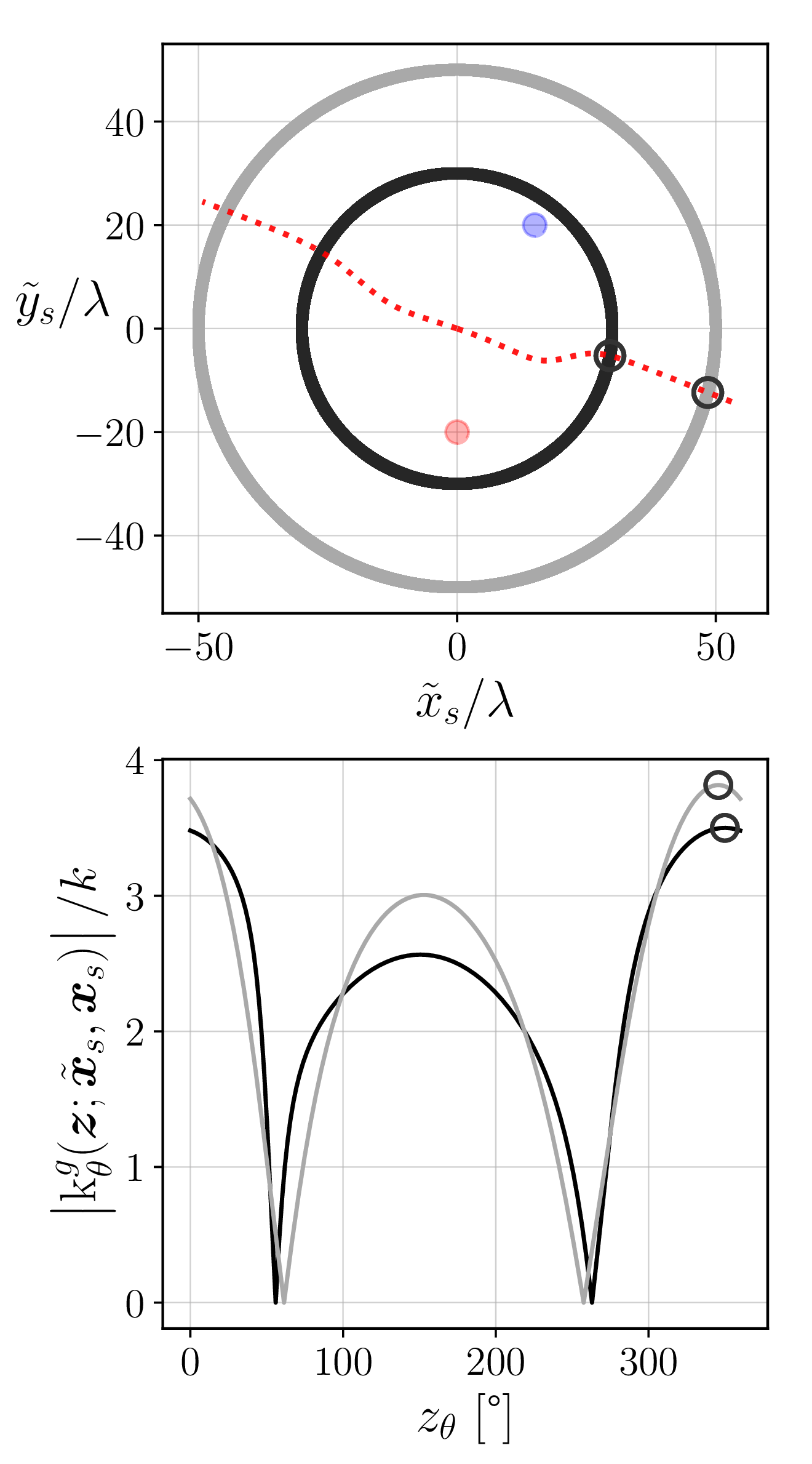}
        \end{subfigure}
        \begin{subfigure}{0.235\linewidth}
            \centering
            \includegraphics[width=\linewidth, trim=0 0 40 10, clip]{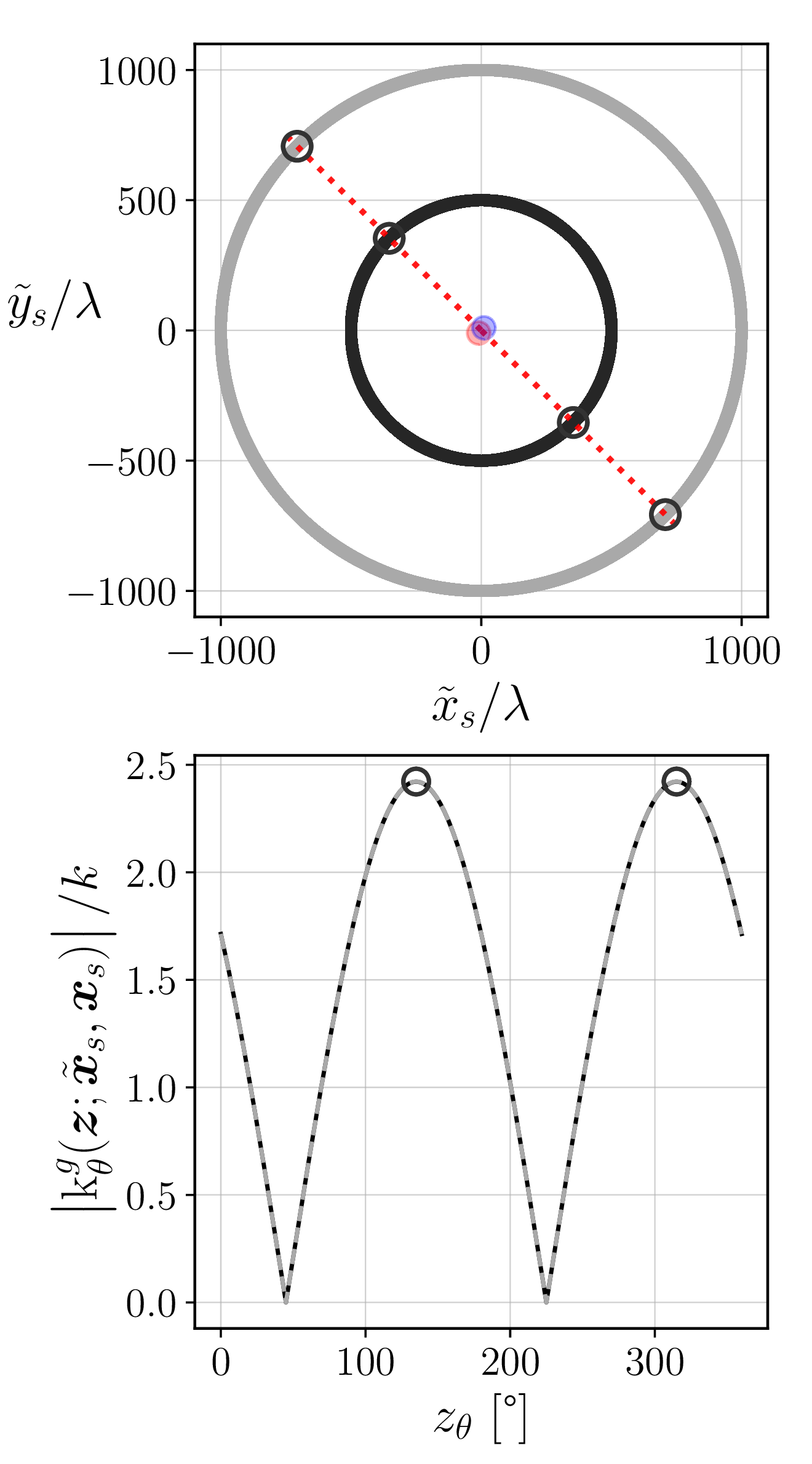}
        \end{subfigure}
        \caption{Circular array (angular direction $\theta$) .} 
    \label{fig:critical_antenna_location_polar}
    \end{subfigure}

    \caption{Critical antenna elements (circled in black) for different given source positions \( \boldsymbol{x}_s \) (red dot) and tentative locations \( \tilde{\boldsymbol{x}}_s \) (blue). Antenna arrays are shown in black and grey; the curves in the lower plot use the same colours as their corresponding arrays.}
\label{fig:critical_antenna_location}
\end{figure*}

The locus of critical antenna positions for a tentative location $\tilde{\boldsymbol{x}}_s$ and a true source point $\boldsymbol{x}_s$ can be further investigated by analysing  how the local frequencies $\mathrm{k}_{i}^{g}(\boldsymbol{z}; \tilde{\boldsymbol{x}}_s, \boldsymbol{x}_s)$ vary across the array. Such analysis is illustrated in Figure~\ref{fig:critical_antenna_location} for continuous arrays. 

The scenarios are shown on the top row, while the bottom row displays the spatial evolution of the local frequencies \( \mathrm{k}_{i}^{g}(\boldsymbol{z}; \tilde{\boldsymbol{x}}_s, \boldsymbol{x}_s) \) along the array, with colours consistent with the array representations above.
Rectangular (\cref{fig:critical_antenna_location_cartesian}) and circular arrays (\cref{fig:critical_antenna_location_polar}) are considered. The critical antennas (circled in black) are identified using \eqref{eq:critical_antenna_element} and are shown for various given source and tentative locations. For infinite continuous arrays, the local extrema of $\mathrm{k}_i^{g}$ can be found by setting the second derivative of the phase function $\phi$ to zero, i.e.,
\begin{equation}
\mathcal{C}_{i}(\tilde{\boldsymbol{x}}_s, \boldsymbol{x}_s) \subseteq 
\mathcal{Z}_{\mathrm{ex}, i} \triangleq 
\left\{ \boldsymbol{z} :
\frac{\partial \mathrm{k}_{i}^{g}(\boldsymbol{z}; \tilde{\boldsymbol{x}}_s, \boldsymbol{x}_s)}{\partial z_i} =
\frac{\partial^2 \phi_g (\boldsymbol{z}; \tilde{\boldsymbol{x}}_s, \boldsymbol{x}_s)}{\partial z_i^2} = 0
\right\}.
\label{eq:phi_second}
\end{equation}
The loci of these extrema $\mathcal{Z}_{\mathrm{ex}, i}$ are indicated by the dashed red lines, which clearly intersect the arrays at the positions of the critical antennas.

\subsubsection{Linear Array} 
\label{sec:critical_antenna_location_linear}

Appendix~\ref{appendix:critical_antenna_location_linear} shows that, for a linear array, the extrema loci $\mathcal{Z}_{\mathrm{ex}, i}$ can be approximated by hyperbolas. The corresponding CAEs therefore lie asymptotically on a straight line (black dashed line in \cref{fig:critical_antenna_location_cartesian}), whose analytical expression is provided in the appendix. If a discrete array intersects this line, the antenna achieving the \textit{global} maximum of $\mathrm{k}_i^{g}$ for a given candidate point $\tilde{\boldsymbol{x}}_s$ is already present; hence, extending the array along the same geometry cannot increase this maximum.
%If the array intersects this line, the antenna corresponding to the \textit{global} maximum for the candidate point \( \tilde{\boldsymbol{x}}_s \)--i.e., the one achieving the largest spatial frequency within the entire array--is already included. Consequently, adding more antennas along the same geometry (e.g., extending a linear array) will not further increase this maximum frequency.
%Moreover, the line indicates how the global CAE shifts when the array is moved away from the source (here, the CAE shifts to the right when moving from the black to the light-gray array), allowing the search region for the global maximum to be constrained.

However, the maximum spatial frequency $\mathrm{K}_x$ differs between the two array configurations. It attains larger values when the source $\boldsymbol{x}_s$ is closer to the array (black array), as can be seen on the below graph. As illustrated in Figure~\ref{fig:chirp_illustration_cartesian}, this behaviour arises from a larger angular separation between the two vectors in~\eqref{eq:local_spatial_frequency_g} (in red and blue) when the source is closer to the antennas, which increases their difference and thus the spatial frequency vector $\boldsymbol{\mathrm{k}}^{g}$ (in green) grows. Since $\mathrm{k}_x^{g}$ is the $x$-component of this vector, this resulting maximum $\mathrm{K}_x$ is larger for closer sources. Because $\mathrm{K}_x$ appears directly in the non-aliasing conditions \eqref{eq:aliasing_conditions}, a larger value increases the risk of violating these conditions, and therefore increases the likelihood of aliasing.

\subsubsection{Circular Array}
\label{sec:critical_antenna_location_polar}

In Figure~\ref{fig:critical_antenna_location_polar}, the same analysis is performed for complete $360^{\circ}$ circular arrays. Finding the loci of the extrema shown in dashed red is more complex in this case and is carried out numerically by solving \eqref{eq:phi_second}. 

Comparing the maximum spatial frequency $\mathrm{K}_\theta$ (circled in the below graphs), we observe that the effect of array radius on $\mathrm{K}_\theta$ strongly depends on the $(\tilde{\boldsymbol{x}}_s, \boldsymbol{x}_s)$ pair. When the source is centered (second column), $\mathrm{K}_\theta$ is independent of the array radius. For off-centered sources (first and third columns), $\mathrm{K}_\theta$ decreases (first column) or increases (third column) with the radius. As illustrated in \cref{fig:chirp_illustration_polar}, this is due to the varying angular mismatch between the local spatial frequencies (green) and the antenna positions (black) with radius, depending on source and tentative locations.

As shown in the top row, the loci of the extrema (shown in dashed red) strongly depend on the source $\boldsymbol{x}_s$ and tentative locations $\tilde{\boldsymbol{x}}_s$. 
For a centered source (second column), the locus of extrema reduces to the line orthogonal to $\boldsymbol{x}_s - \tilde{\boldsymbol{x}}_s$ and passing through $\tilde{\boldsymbol{x}}_s$, yielding two critical antennas with the same value $\mathrm{K}_\theta$. 
Off-centered configurations lead to more intricate loci with no closed-form shape.

However, as shown in \cite{icassp}, under an individual FF approximation for each antenna element--valid when the array radius $r_z >> \| \boldsymbol{x}_s \|$ and $r_z >> \| \tilde{\boldsymbol{x}}_s \|$--the phase function $\phi_g(\boldsymbol{z}; \tilde{\boldsymbol{x}}_s, \boldsymbol{x}_s)$ reduces to a function of the vector difference $\boldsymbol{x}_s - \tilde{\boldsymbol{x}}_s$:
\begin{equation}
\phi_g(\boldsymbol{z}; \tilde{\boldsymbol{x}}_s, \boldsymbol{x}_s) \approx 
k r_{s \tilde{s}} \cos(\theta_z - \theta_{s \tilde{s}}),
\label{eq:phi_g_taylor_difference}
\end{equation}
where $r_{s \tilde{s}} = \| \boldsymbol{x}_s - \tilde{\boldsymbol{x}}_s \|$ and $\theta_{s \tilde{s}} = \angle (\boldsymbol{x}_s - \tilde{\boldsymbol{x}}_s) $. Nullifying the second derivative of \eqref{eq:phi_g_taylor_difference} yields
\begin{align}
\frac{\partial^2 \phi_g(\boldsymbol{z}; \tilde{\boldsymbol{x}}_s, \boldsymbol{x}_s)}{\partial z_\theta^2} 
&= - k \, r_{s \tilde{s}} \, \cos(\theta_z - \theta_{s \tilde{s}}) = 0 \\
&\Leftrightarrow \theta_z = \theta_{s \tilde{s}} + \frac{\pi}{2} + n \pi, \quad n \in \mathbb{Z}, 
\end{align}
which shows that, under the individual FF approximation, a full circular array always contains two critical antennas for any tentative location $\tilde{\boldsymbol{x}}_s$. These co-critic antennas lie at the intersections of the array with the line orthogonal to $\boldsymbol{x}_s - \tilde{\boldsymbol{x}}_s$ (and passing through the origin) and are separated by $\pi$ radians. This is clearly seen in the last column of \cref{fig:critical_antenna_location_polar}, where the array radius is significantly larger than the distances $\| \boldsymbol{x}_s \|$ and $\| \tilde{\boldsymbol{x}}_s \|$, with $\boldsymbol{x}_s=(-10; -10)\lambda$ and  $\tilde{\boldsymbol{x}}_s=(10; 10)\lambda$. In this individual FF case, the maximum spatial frequency $K_\theta$ becomes independent of the array radius, as expected from \eqref{eq:phi_g_taylor_difference}, which depends only on the vector difference $\boldsymbol{x}_s - \tilde{\boldsymbol{x}}_s$. \\

The previous analysis focused on the global maxima of the local spatial frequencies $\mathrm{k}_i^{g}$ along the array. In practice, arrays are finite or cover only a limited angular aperture, so only a portion of $\mathrm{k}_i^{g}$’s variation is sampled. As a result, the critical antennas may not coincide with the global maxima. In such cases, they are identified by evaluating $\mathrm{k}_i^{g}$ at the discrete antenna positions and selecting the maximum, as in \eqref{eq:critical_antenna_element}.

Conversely, when the global maximum of $\mathrm{k}_i^{g}$ lies within the array aperture, the critical antennas remain unchanged even if additional elements are added along the same geometry. The CAE framework thus allows assessing how antenna addition or removal affects the achievable AFR and identifying cases where the inclusion principle of \cref{prop:inclusion_principle_2} holds with equality.
\Cref{prop:lemme_ant_critic_removal} formalises the impact on the AFR of removing antennas from a given array, based on the critical antenna elements defined in \cref{def:critical_antenna_element}.

\begin{proposition}[Removal of Antennas]\label{prop:lemme_ant_critic_removal}
Let $\mathcal{Z}^{(2)} \subseteq \mathcal{Z}^{(1)}$ be two antenna sets and let $\mathcal{S}^{(k)}$ denotes the AFR induced by $\mathcal{Z}^{(k)}$. 
For every candidate point $\tilde{\boldsymbol{x}}_s \in \partial \mathcal{S}^{(1)}(\boldsymbol{x}_s)$ 
and for every dimension $i$, if the CAEs satisfy
\begin{equation}
    \mathcal{C}_{ i}^{(1)} (\tilde{\boldsymbol{x}}_s, \boldsymbol{x}_s) \cap  \mathcal{Z}^{(2)} \neq \emptyset,
\end{equation}
then the AFRs coincide, i.e.
\begin{equation}
    \mathcal{S}^{(1)}(\boldsymbol{x}_s) = \mathcal{S}^{(2)}(\boldsymbol{x}_s).
\end{equation}
Otherwise, if this condition is violated for at least one $\tilde{\boldsymbol{x}}_s$ or dimension $i$, we have a strict inclusion:
\begin{equation}
    \mathcal{S}^{(1)}(\boldsymbol{x}_s) \subset \mathcal{S}^{(2)}(\boldsymbol{x}_s).
\end{equation}
\end{proposition}

\begin{proof}
%The result follows directly from the definitions of the AFR and the CAEs. 
If, for every $\tilde{\boldsymbol{x}}_s \in \partial\mathcal{S}^{(1)}(\boldsymbol{x}_s)$ and every dimension $i$, 
the reduced set $\mathcal{Z}^{(2)}$ contains at least one CAE of $\mathcal{Z}^{(1)}$, then the maximum spatial frequency
\[
\mathrm{K^{(1)}}_i(\tilde{\boldsymbol{x}}_s,\boldsymbol{x}_s)
= \max_{\boldsymbol{z}\in\mathcal{Z}^{(1)}} 
\big| \mathrm{k}_i^{g}(\boldsymbol{z};\tilde{\boldsymbol{x}}_s,\boldsymbol{x}_s) \big|
\]
is preserved for all $\tilde{\boldsymbol{x}}_s$ when the maximisation is restricted to $\mathcal{Z}^{(2)}$.
Hence, the aliasing constraints remain unchanged and 
$\mathcal{S}^{(1)}(\boldsymbol{x}_s)=\mathcal{S}^{(2)}(\boldsymbol{x}_s)$.
Conversely, if there exist $\tilde{\boldsymbol{x}}_s$ and a dimension $i$ such that 
$
 \mathcal{C}_{ i}^{(1)} (\tilde{\boldsymbol{x}}_s, \boldsymbol{x}_s)
\cap \mathcal{Z}^{(2)}=\emptyset,
$
then critical antennas in $\mathcal{Z}^{(2)}$ yield strictly smaller values of 
$\lvert \mathrm{K}_i^{g}\rvert$ than those attained in $\mathcal{Z}^{(1)}$ since $\mathcal{Z}^{(2)} \subseteq \mathcal{Z}^{(1)}$, relaxing the aliasing condition and enlarging the AFR. 
Consequently, $
\mathcal{S}^{(1)}(\boldsymbol{x}_s)\subset \mathcal{S}^{(2)}(\boldsymbol{x}_s).
$
\end{proof}

\noindent This proposition gives a practical criterion to determine the impact of antenna removal on the AFR.
If the reduced antenna set retains at least one critical antenna for every candidate point on the boundary of the AFR, the AFR remains unchanged. Otherwise, removing antennas strictly enlarges the AFR.

Conversely, \cref{prop:lemme_ant_critic_addition} formalises the impact on the AFR of adding antennas to a given array.

\begin{proposition}[Addition of Antennas]\label{prop:lemme_ant_critic_addition}
Let \( \mathcal{Z}^{(1)} \subseteq \mathcal{Z}^{(2)} \) be two antenna sets.  
For every candidate point \( \tilde{\boldsymbol{x}}_s \in \partial \mathcal{S}^{(1)}(\boldsymbol{x}_s) \) and every dimension \( i \), suppose that the corresponding sets of critical antenna elements satisfy
\begin{equation}
    \mathcal{C}_{ i}^{(1)}(\tilde{\boldsymbol{x}}_s, \boldsymbol{x}_s) 
    \cap 
    \mathcal{C}_{i}^{(2)}(\tilde{\boldsymbol{x}}_s, \boldsymbol{x}_s) 
    \neq \emptyset.
\end{equation}
Then the AFRs are identical, i.e.,
\begin{equation}
    \mathcal{S}^{(1)}(\boldsymbol{x}_s) = \mathcal{S}^{(2)}(\boldsymbol{x}_s).
\end{equation}
Otherwise, if there exists at least one candidate point \( \tilde{\boldsymbol{x}}_s \) and one dimension \( i \) for which the above intersection is empty, the AFR strictly shrinks with the larger antenna set:
\begin{equation}
    \mathcal{S}^{(2)}(\boldsymbol{x}_s) \subset \mathcal{S}^{(1)}(\boldsymbol{x}_s).
\end{equation}
\end{proposition}

\begin{proof}  
If every boundary point \( \tilde{\boldsymbol{x}}_s \in \partial \mathcal{S}^{(1)}(\boldsymbol{x}_s) \) and every dimension \( i \) share at least one critical antenna element between \( \mathcal{Z}^{(1)} \) and \( \mathcal{Z}^{(2)} \), then the maximum spatial frequencies coincide, i.e. \( \mathrm{K}_i^{(1)} = \mathrm{K}_i^{(2)} \). Hence, the aliasing conditions in \eqref{eq:aliasing_conditions} are unchanged and $
\mathcal{S}^{(1)}(\boldsymbol{x}_s) = \mathcal{S}^{(2)}(\boldsymbol{x}_s).
$ 
Otherwise, if for some \( (\tilde{\boldsymbol{x}}_s,i) \), no common CAE exists, then the larger set \( \mathcal{Z}^{(2)} \) introduces new critical antenna elements absent from \( \mathcal{Z}^{(1)} \). This implies that the maximum spatial frequency strictly increases since \( \mathcal{Z}^{(1)} \subseteq \mathcal{Z}^{(2)} \), i.e.  
$
\mathrm{K}_{i}^{(2)}(\tilde{\boldsymbol{x}}_s, \boldsymbol{x}_s) > \mathrm{K}_{i}^{(1)}(\tilde{\boldsymbol{x}}_s, \boldsymbol{x}_s).
$  
As a result, the aliasing conditions in \eqref{eq:aliasing_conditions} become more restrictive, leading to a reduction of the AFR:  
$
\mathcal{S}^{(2)}(\boldsymbol{x}_s) \subset \mathcal{S}^{(1)}(\boldsymbol{x}_s).
$  
\end{proof}

\noindent In other words, \cref{prop:lemme_ant_critic_addition} states that if a newly added antenna becomes critical for a point of the AFR contour $\partial \mathcal{S}$ --replacing all previously critical elements for this point--the maximum spatial frequency $\mathrm{K}_i$ increases, tightening the aliasing constraints and reducing the AFR. Otherwise, $\mathrm{K}_i$ remains unchanged, and the AFR is preserved.

Propositions \ref{prop:lemme_ant_critic_removal} and \ref{prop:lemme_ant_critic_addition} provide practical criteria, based on critical antenna elements, to assess how adding or removing antennas affects the AFR. They enable analysis of the influence of system parameters--array geometry, inter-element spacing, and source location--on the AFR, as conducted in the following sections.

%%%%%%%%%%% Cartesian Array Analysis %%%%%%%%%%%

\section{Rectangular Array Analysis}
\label{sec:point_scatterer}

This section investigates the impact of system parameters on the AFR and resolution for rectangular arrays, using the framework of \cref{sec:geometrical_tools}. The results highlight the interplay and trade-offs between these metrics as functions of array configuration. Numerical results illustrate the effects, with local spatial frequencies $\boldsymbol{\mathrm{k}}^{h}$ and $\boldsymbol{\mathrm{k}}^{g}$ along $x$- and $y$-axes given by \eqref{eq:local_spatial_frequency_h} and \eqref{eq:local_spatial_frequency_g}, respectively.

\subsection{Safe Antenna Spacing and Far-Field Compliance}
\label{sec:safe_ant_spacing_linear}

\cref{prop:safe_spacing_cartesian} provides a condition on the antenna spacing to avoid aliasing in rectangular array geometries.

\begin{proposition}\label{prop:safe_spacing_cartesian}
    For a rectangular antenna array geometry, a sufficient condition to avoid aliasing for all $\boldsymbol{\tilde{x}}_s$ is that the antenna spacing $\Delta_i$ along the $i$-th direction satisfies
    \begin{equation}
        \Delta_i \leq \frac{\lambda}{2},
    \label{eq:safe_spacing_cartesian}
    \end{equation}
\end{proposition}

\begin{proof}
Aliasing-free sampling requires $\mathrm{K}_i(\tilde{\boldsymbol{x}}_s,\boldsymbol{x}_s)\le 2\pi/\Delta_i$ \eqref{eq:aliasing_conditions}.  
From \eqref{eq:local_spatial_frequency_g}, $\mathrm{K}_i(\tilde{\boldsymbol{x}}_s,\boldsymbol{x}_s)=2k$ at its maximum possible value, so the safe spacing along the $i$-th direction is
$\Delta_i = 2\pi/\mathrm{K}_i = \pi/k = \lambda/{2}. $ 
\end{proof}

This condition, derived from the chirp-based near-field approximation, reduces to the classical far-field non-aliasing constraint of a maximum inter-element spacing of $\lambda/2$ \cite{li_sparse_2025}. In practice, however--particularly at high carrier frequencies or under hardware constraints--this requirement may be difficult to satisfy, making the determination of the AFR essential.

Furthermore, as shown in Appendix~\ref{appendix:ff_conditions_cartesian}, applying the chirp-based near-field framework to the far-field regime recovers the standard aliasing structure based on directional cosines. This confirms that the far-field model is a special case of the general near-field formulation developed in this article.

\subsection{Impact of System Parameters}

In \cref{fig:impact_parameters_cart}, the influence of various system parameters on both AFR and resolution is illustrated for a point source. 

Each subfigure contains:
\begin{itemize}
    \item The AFR boundary $\partial \mathcal{S}(\boldsymbol{x}_s)$ in yellow. The interior corresponds to candidate points $\tilde{\boldsymbol{x}}_s$ satisfying the aliasing conditions in \eqref{eq:aliasing_conditions} using the spatial frequency estimates of \eqref{eq:Ki_chirp}. Outside this region, aliasing artefacts are expected.
    \item The magenta contour of the resolution region $\mathcal{R}(\boldsymbol{x}_s)$, defined as
    \begin{equation} 
    \mathcal{R}(\boldsymbol{x}_s) \triangleq \left\{ \tilde{\boldsymbol{x}}_s : |\tilde{x}_{s,i} - x_{s,i}| \leq \delta_i(\boldsymbol{x}_s)/2 , \, \, \forall i \right\}, 
    \end{equation}
    where $\delta_i(\boldsymbol{x}_s)$ is defined in \eqref{eq:resolution}. Smaller regions reflect improved resolution and better focusing. 
    \item The boundary of the intersection of the axis-wise NCZs, $\partial(\mathcal{C}_x \cap \mathcal{C}_y)$, is shown as a dashed curve in the corresponding array colour. Antennas added in this region do not contribute to resolution improvement in any direction. 
    \item When relevant, AFR-critical antennas are circled in white.
\end{itemize}

Comparing numerical results with the chirp-based AFR predictions shows that the theoretical AFR matches the observed artefact-free area of the $AF$ for the different configurations. Inside the AFR, the $AF$ response matches that of the non-aliased case. Outside it, aliasing yields spurious high-magnitude responses that can be mistaken for genuine sources, degrading localisation reliability.

In all configurations, the AFR remains centered around the true source $\boldsymbol{x}_s. $%, indicating that the likelihood of aliasing increases when the candidate location $\tilde{\boldsymbol{x}}_s$ moves away from the true source position $\boldsymbol{x}_s$. 
As the reconstruction point \( \tilde{\boldsymbol{x}}_s \) moves away from \( \boldsymbol{x}_s \)\footnote{Specifically, in a direction not aligned with \( \nabla_{\boldsymbol{z}} \| \boldsymbol{x}_s-\boldsymbol{z}_{} \| \).}, the wave vectors \( k \nabla_{\boldsymbol{z}} \|  \tilde{\boldsymbol{x}}_s - \boldsymbol{z} \| \) and \( k \nabla_{\boldsymbol{z}} \|  \boldsymbol{x}_s - \boldsymbol{z} \| \) become increasingly misaligned (see Figure~\ref{fig:chirp_illustration}). Their difference--the local wavenumber vector \( \boldsymbol{\mathrm{k}}^{g}(\tilde{\boldsymbol{x}}_s, \boldsymbol{x}_s) \)--thus grows in magnitude, increasing the local spatial frequencies and the maximum frequency $\mathrm{K}_i$. Larger $\mathrm{K}_i$ makes the aliasing condition \eqref{eq:aliasing_conditions} more difficult to satisfy, explaining why aliasing becomes more likely as the reconstruction point departs from the source.

\begin{figure*}
\centering 

\includegraphics[width=1\linewidth, trim=0 45 0 45, clip]{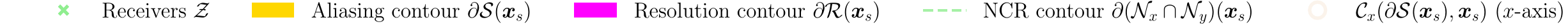}

\vspace{0.3cm}

% 1ère Subfigure
\begin{subfigure}{0.47\linewidth}
\centering 
\includegraphics[width=0.95\linewidth, trim=20 382 160 20, clip]{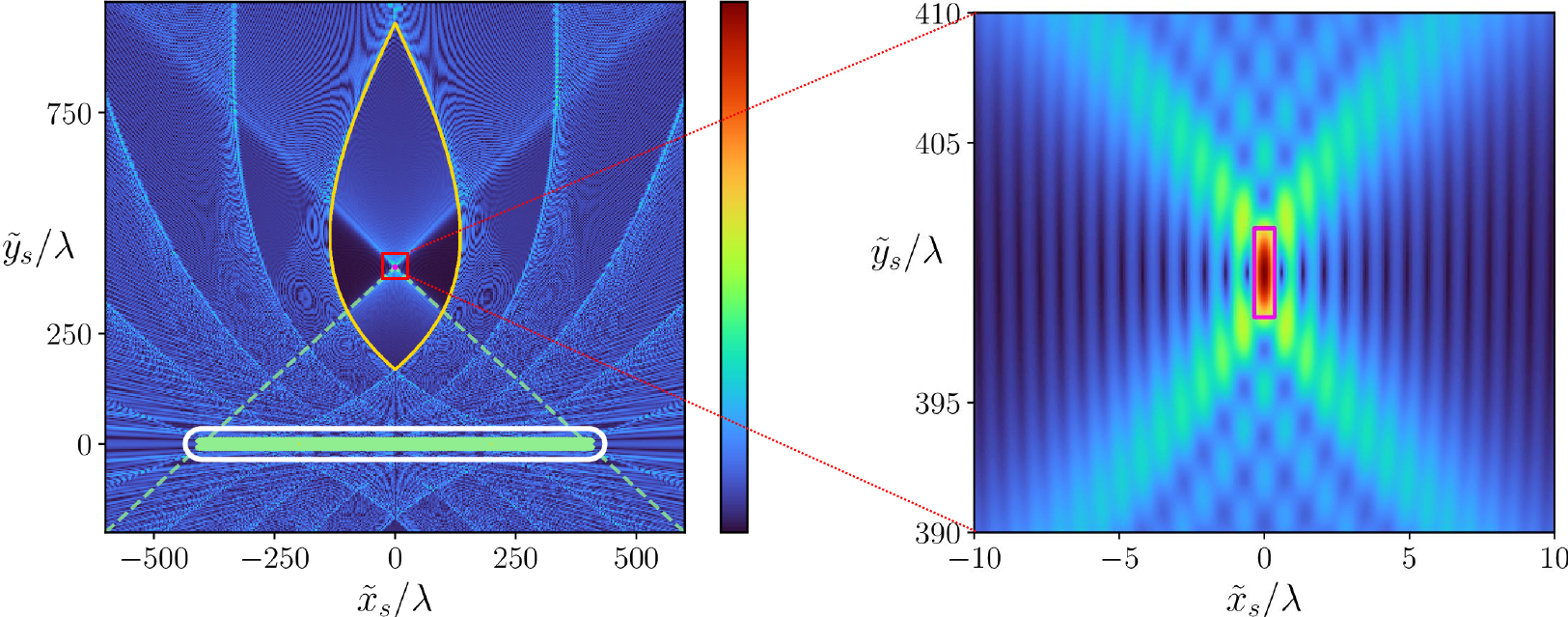}
\caption{Source location: $\boldsymbol{x}_s = (0;400)\lambda$. Array length: $D_x=800\lambda$. Number of antennas: $N=256$. }
\label{fig:impact_a}
\end{subfigure}
\hfill
% 2ème Subfigure
\begin{subfigure}{0.47\linewidth}
\centering 
\includegraphics[width=0.95\linewidth, trim=20 382 160 20, clip]{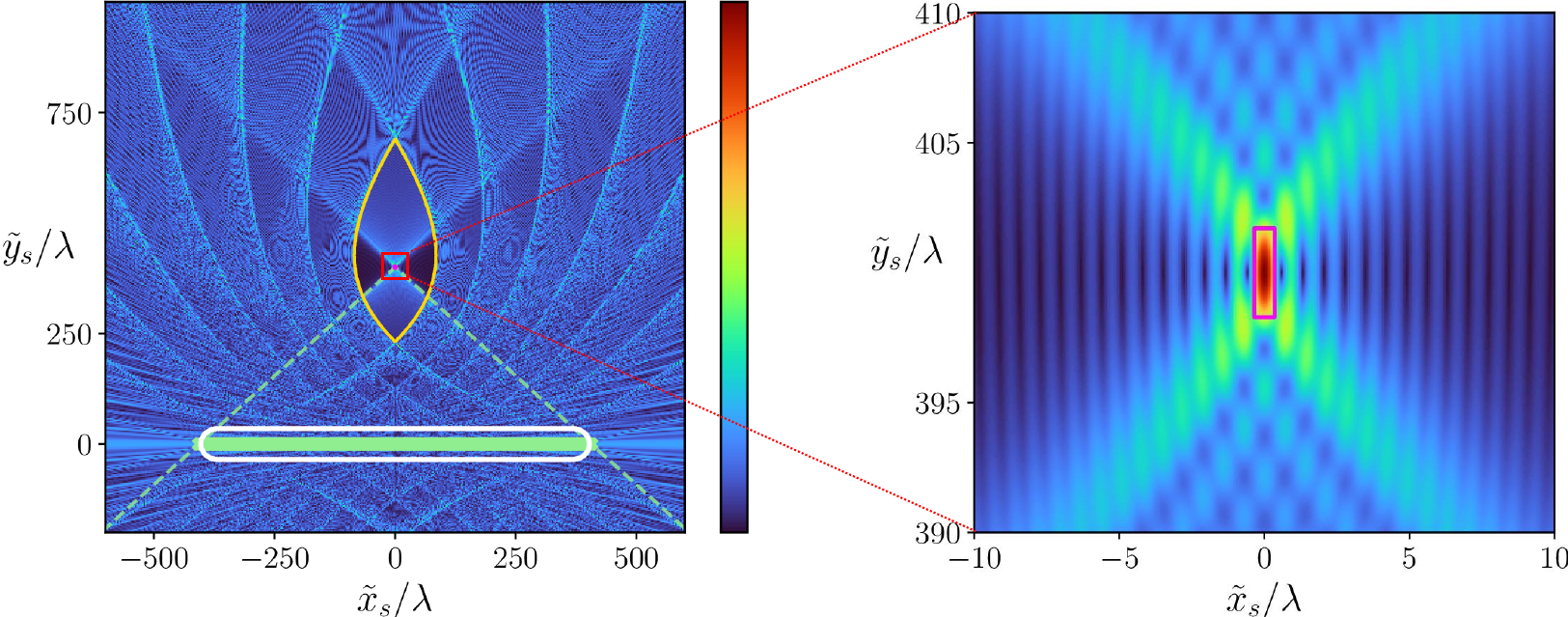}
\caption{Source location: $\boldsymbol{x}_s = (0;400)\lambda$. Array length: $D_x=800\lambda$. Number of antennas: $N=168$.}
\label{fig:impact_b}
\end{subfigure}

\vspace{0.35cm}

% 3ème Subfigure
\begin{subfigure}{0.47\linewidth}
\centering 
\includegraphics[width=0.95\linewidth, trim=20 382 160 20, clip]{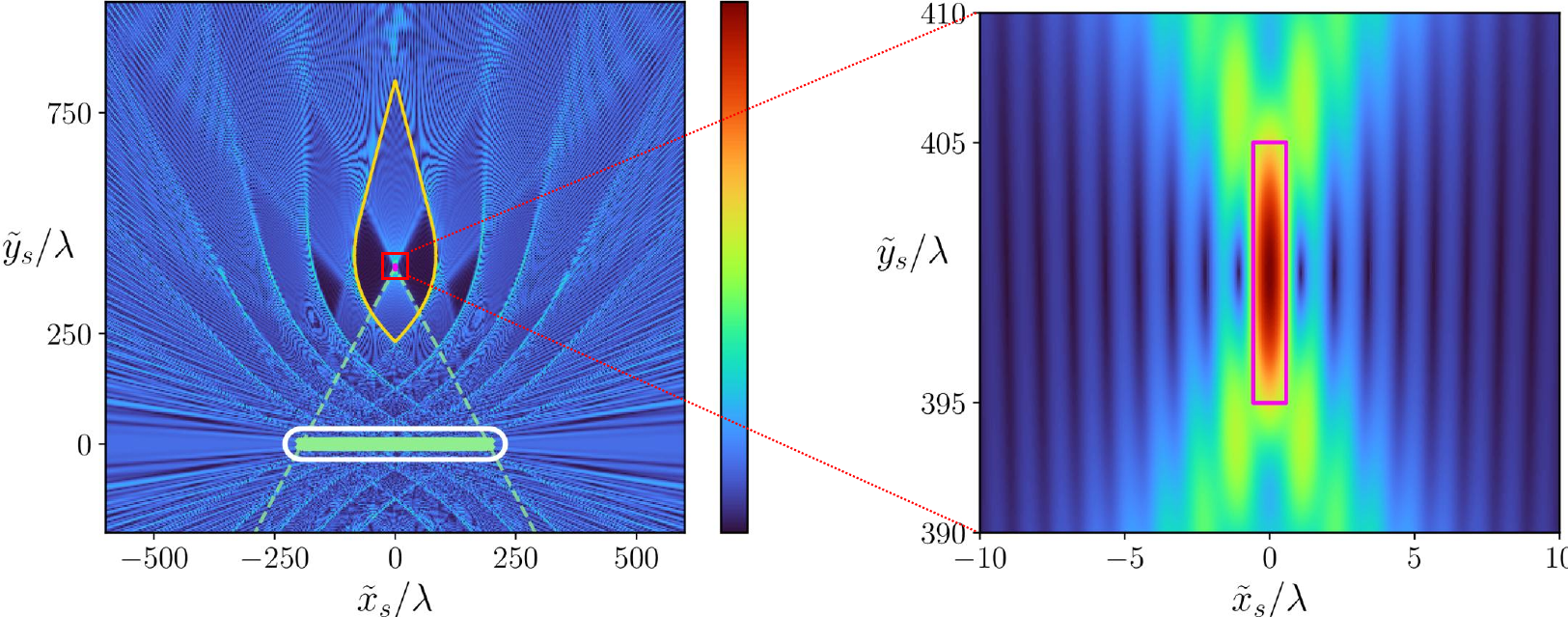}
\caption{Source location: $\boldsymbol{x}_s = (0;400)\lambda$. Array length: $D_x=400\lambda$. Number of antennas: $N=84$. }
\label{fig:impact_c}
\end{subfigure}
\hfill
% 4ème Subfigure
\begin{subfigure}{0.47\linewidth}
\centering 
\includegraphics[width=0.95\linewidth, trim=20 382 160 20, clip]{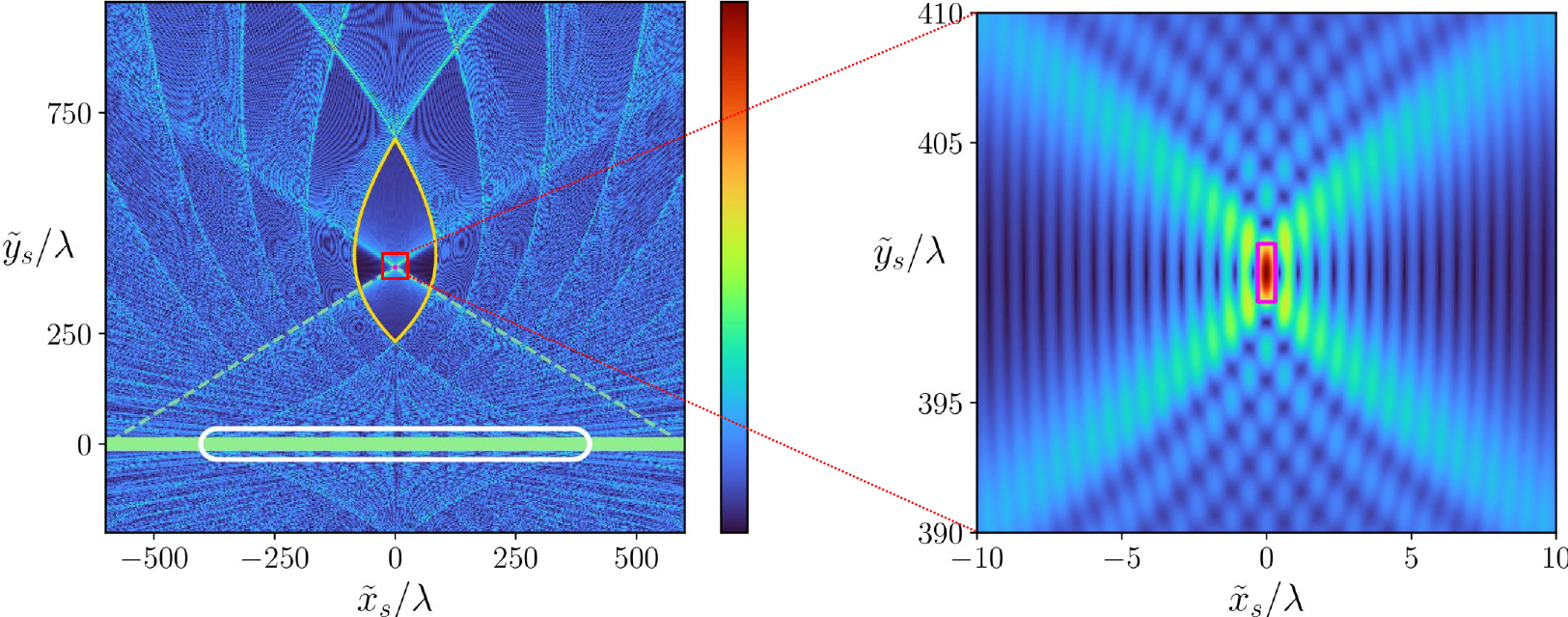}
\caption{Source location: $\boldsymbol{x}_s = (0;400)\lambda$. Array length: $D_x=1200\lambda$. Number of antennas: $N=248$.}
\label{fig:impact_d}
\end{subfigure}

\vspace{0.35cm}

% 5ème Subfigure
\begin{subfigure}{0.47\linewidth}
\centering 
\includegraphics[width=0.95\linewidth, trim=20 382 160 20, clip]{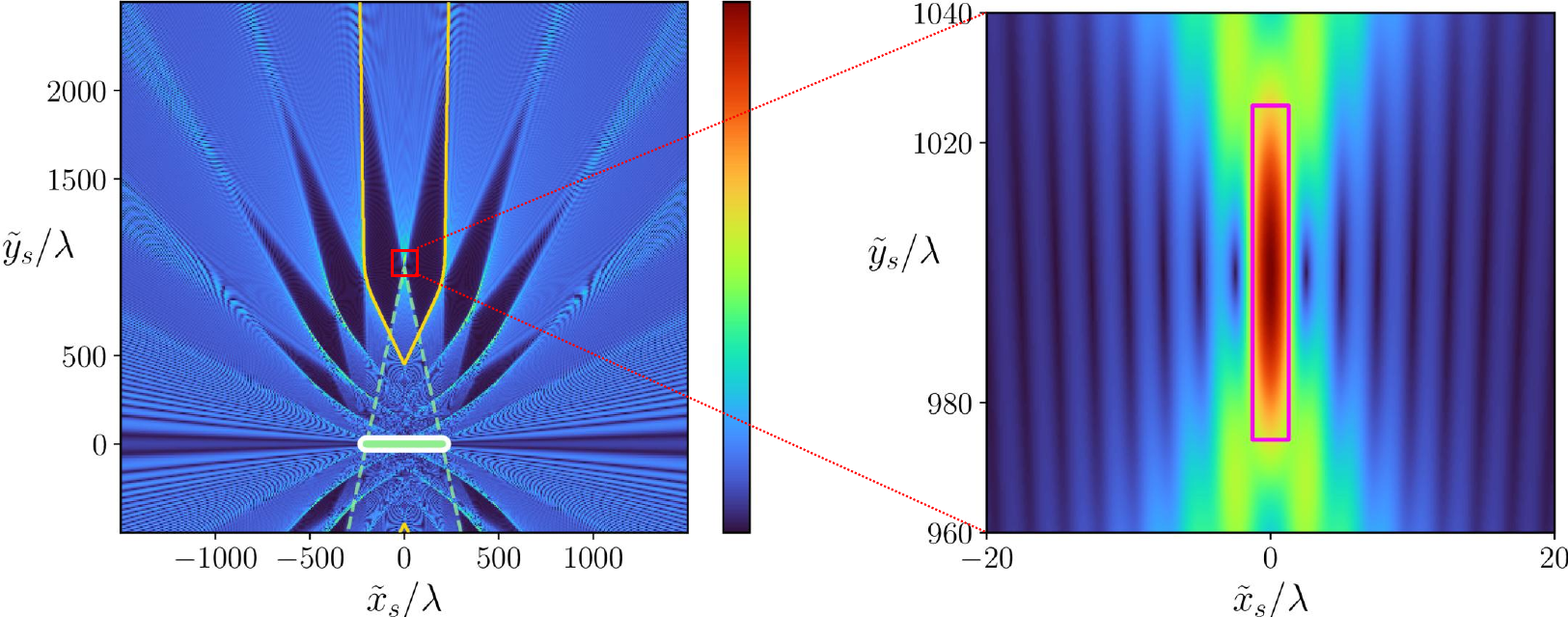}
\caption{Source location: $\boldsymbol{x}_s = (0;1000)\lambda$. Array length: $D_x=400\lambda$. Number of antennas: $N=84$. }
\label{fig:impact_e}
\end{subfigure}
\hfill
% 6ème Subfigure
\begin{subfigure}{0.47\linewidth}
\centering 
\includegraphics[width=0.95\linewidth, trim=20 382 160 20, clip]{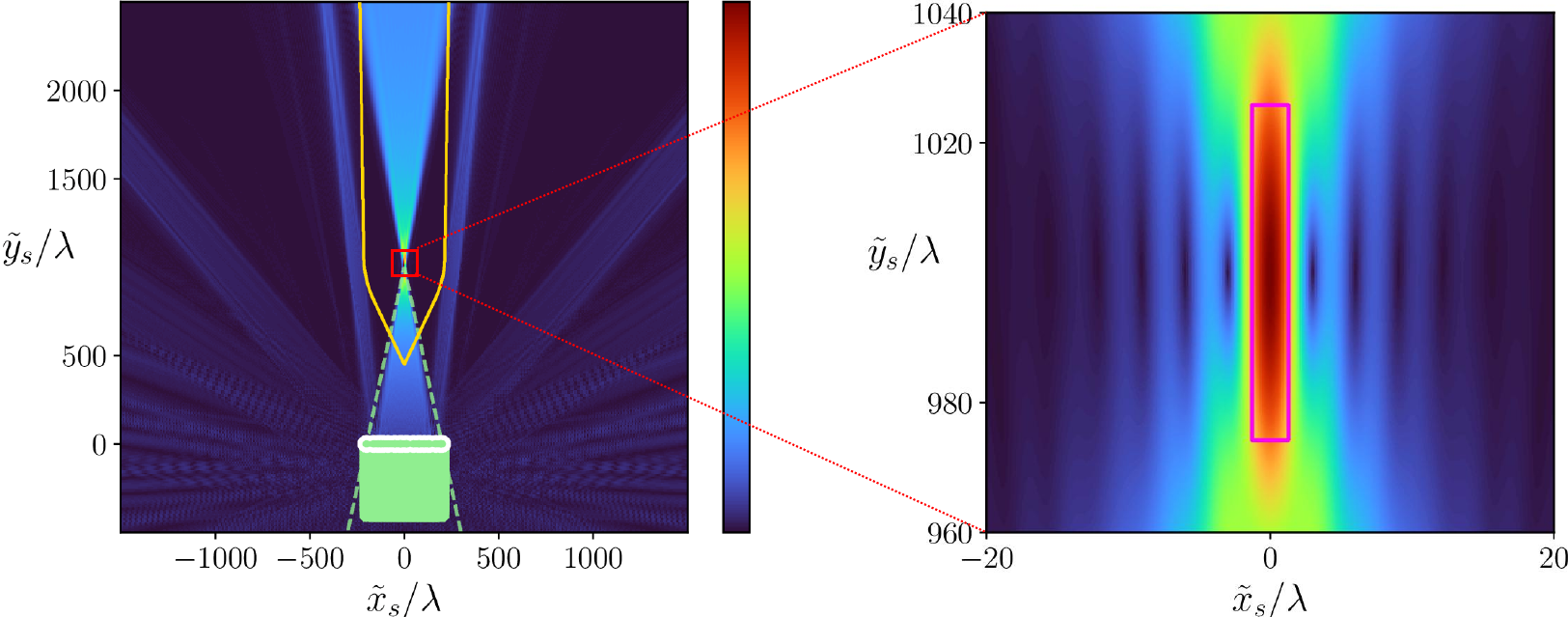}
\caption{Source location: $\boldsymbol{x}_s = (0;1000)\lambda$. Array dimensions: $D_x=400\lambda$, $D_y=400\lambda$. Number of antennas: $N=84\times84$. }
\label{fig:impact_f}
\end{subfigure}

\vspace{0.35cm}

% 7ème Subfigure
\begin{subfigure}{0.47\linewidth}
\centering 
\includegraphics[width=0.95\linewidth, trim=20 382 160 20, clip]{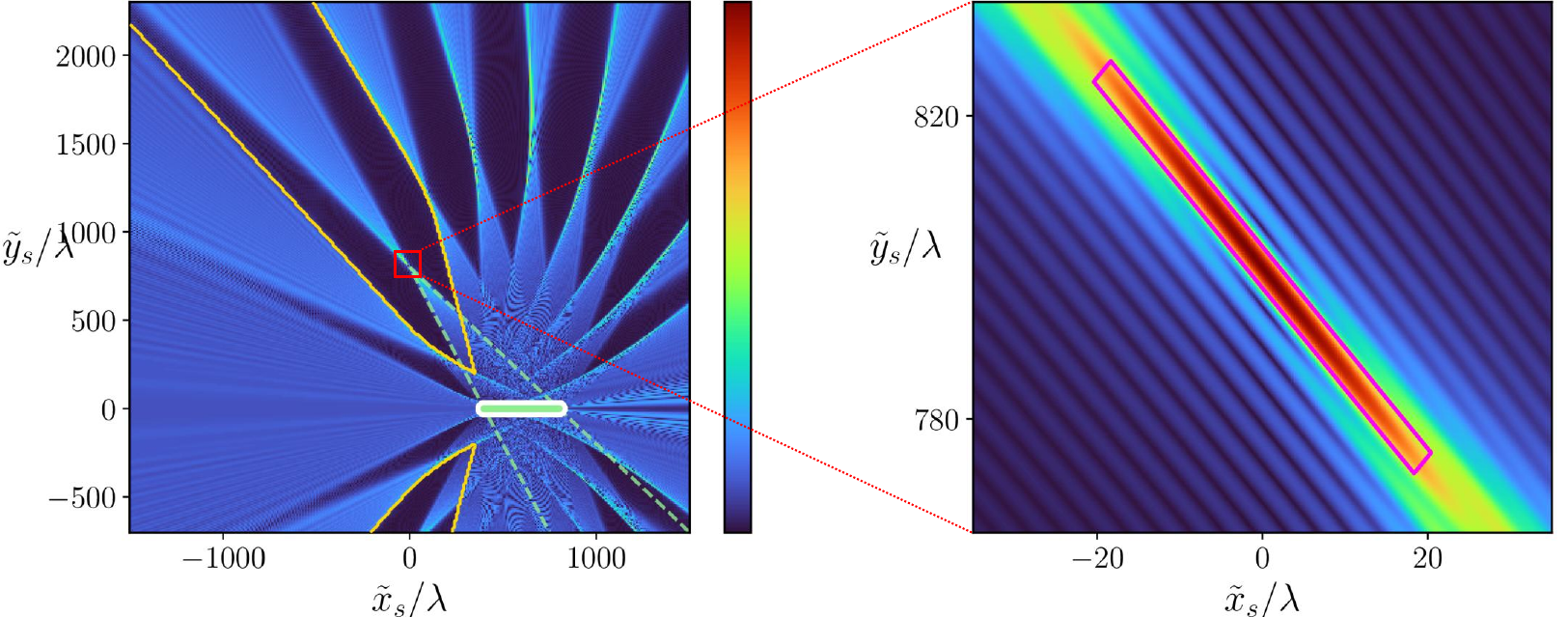}
\caption{Source location: $\boldsymbol{x}_s = (0;1000)\lambda$. Array length: $D_x=400\lambda$. Number of antennas: $N=84$. Array shifted by $600\lambda$ along $x$. }
\label{fig:impact_g}
\end{subfigure}
\hfill
% 8ème Subfigure
\begin{subfigure}{0.47\linewidth}
\centering 
\includegraphics[width=0.95\linewidth, trim=20 382 160 20, clip]{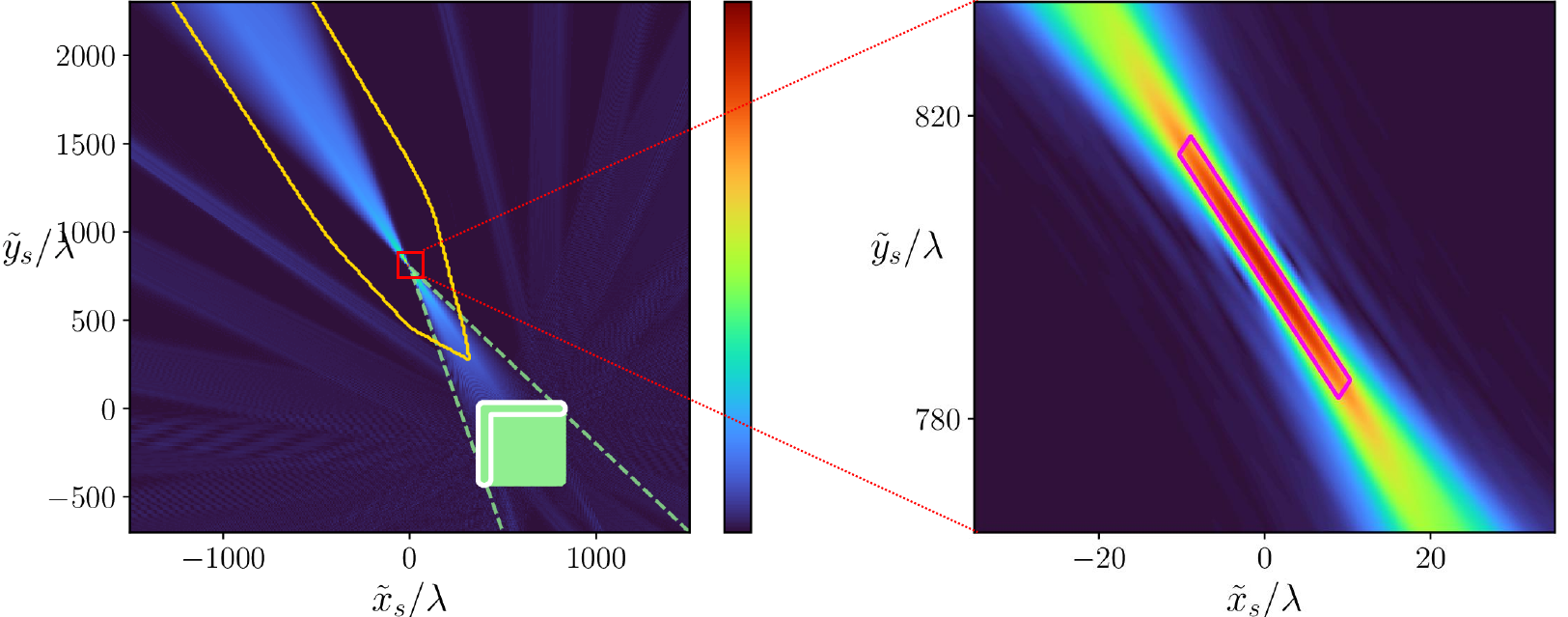}
\caption{Source location: $\boldsymbol{x}_s = (0;1000)\lambda$. Array dimensions: $D_x=400\lambda$, $D_y=400\lambda$. Antennas: $N=84\times84$. Array center: $(600;-200)\lambda$. }
\label{fig:impact_h}
\end{subfigure}

\caption{Impact of system parameters on the $AF$ for a point source with rectangular arrays. }

\label{fig:impact_parameters_cart}
\end{figure*}

The trade-offs between resolution and AFR size are examined next by varying each system parameter independently while keeping all others fixed. The following key observations can be made.

\subsubsection{Antenna Spacing}

Comparing Figures \ref{fig:impact_a} and \ref{fig:impact_b}, which differ only in antenna spacing ( \( N=256 \) in \cref{fig:impact_a} vs. \( N=168 \) in \cref{fig:impact_b} over the same aperture \( D_x=800\lambda \) ), shows that larger spacing yields a significantly smaller AFR. This occurs because the spatial spectrum \(G\) replicates more densely with a period \(2\pi/\Delta_x\) that decreases with larger spacing $\Delta_x$, raising the risk of aliasing as per~\eqref{eq:aliasing_conditions}. 

Nevertheless, because the aperture remains fixed, the resolution is unchanged, as confirmed by identical NCZ boundaries and resolution region boundaries$\partial \mathcal{R}(\boldsymbol{x}_s)$. This observation is consistent with~\cite{ahmed2014electronic} and can be clearly seen in the zoomed-in plots. This illustrates the following interplay: increasing antenna spacing preserves resolution but reduces the AFR, increasing aliasing risk.

\subsubsection{Array Aperture}

Figures~\ref{fig:impact_b}, \ref{fig:impact_c} and~\ref{fig:impact_d} share the same antenna spacing and source location $\boldsymbol{x}_s$ but differ in array aperture: \( D_x=800\lambda \), \(=400\lambda \), and \(=1200\lambda \), respectively. 

First, increasing the array aperture improves resolution: antennas are added outside the NCZ and thus extends the spatial frequency bandwidth $B_i$ in \eqref{eq:bandwidth_definition} for at least one axis, thereby reducing the resolution cell size.

For the AFR, when the antenna spacing is fixed, its size is governed by the maximum spatial frequency \( \mathrm{K}_x \), not by spectral repetition. 
From \cref{fig:impact_c} to \cref{fig:impact_b}, the AFR shrinks as the newly added antennas become critical at some points along the AFR boundary, increasing \(\mathrm{K}_x\) and tightening the aliasing constraints (see \cref{prop:lemme_ant_critic_addition}).  
In contrast, from \cref{fig:impact_b} to \cref{fig:impact_d}, although the array is extended, the newly added antennas never become critical on the AFR boundary. As a result, \( \mathrm{K}_x \) remains unchanged, and the AFR is preserved. 
Similarly, removing antennas from \cref{fig:impact_d} to \cref{fig:impact_b} does not enlarge the AFR because the removed elements were not critical either (see \cref{prop:lemme_ant_critic_removal}).

Thus, increasing the array aperture always improves resolution, whereas the AFR may shrink or remain unchanged depending on whether the added antennas become critical. Hence, resolution improvement and AFR size are not necessarily coupled, and trade-offs may arise. Propositions \ref{prop:lemme_ant_critic_addition} and \ref{prop:lemme_ant_critic_removal} provide practical criteria to predict these effects.

\subsubsection{Relative source distance}

Comparing Figures \ref{fig:impact_c} and \ref{fig:impact_e}, with the same array configuration, but different source locations (\( \boldsymbol{x}_s = (0;400)\lambda \) vs. \( \boldsymbol{x}_s = (0;1000)\lambda \)), shows that moving the source closer to the array reduces the AFR. As explained in \cref{sec:critical_antenna_location_linear} and illustrated in Figures \ref{fig:chirp_illustration_cartesian} and \ref{fig:critical_antenna_location_cartesian}, this is because shorter source distances to the antennas increase the angular separation between the vectors defining the local spatial frequencies $\boldsymbol{\mathrm{k}}^{g}$, thus increasing the maximum frequency $\mathrm{K}_x$ and tightening the aliasing conditions in \eqref{eq:aliasing_conditions}.

At the same time, bringing the source closer improves resolution, as the spatial frequency bandwidth \( \mathrm{B}_i \) increases. Nearby sources are sampled over a wider range of spatial frequencies (wider non-contributive zones), leading to finer detail in the ambiguity function. The resolution region thus shrinks when moving from \cref{fig:impact_e} to \cref{fig:impact_c}\footnote{The $y$-axis scales differ, with \cref{fig:impact_e} using a larger scale.}. 

The source distance thus introduces a strict trade-off between AFR size and resolution: while closer sources enhance the resolution, they reduce the AFR.

\subsubsection{Array Dimensionality}

Figures \ref{fig:impact_e} and \ref{fig:impact_f} compare a one-dimensional ($d=1$, 1D) linear array along \( x \) (\cref{fig:impact_e}) with a two-dimensional ($d=2$, 2D) planar array in the \( xy \)-plane (\cref{fig:impact_f}), keeping the source location and antenna spacing fixed. The additional antennas in the 2D case lie within the NCZ $\mathcal{N}(\boldsymbol{x}_s)$ of the 1D case and thus do not improve resolution, so the resolution regions are identical\footnote{Spatial frequencies $\boldsymbol{\mathrm{k}}^{h}$ are evaluated here in a rotated coordinate system aligned with the beam axis, allowing analysis along physically meaningful directions--parallel and orthogonal to the beam.}. Removing these antennas when going from 2D to 1D does not affect the AFR, as they are non-critical for its boundary, in agreement with \cref{prop:lemme_ant_critic_removal}, which is confirmed by the identical AFR boundaries in both plots. 

In contrast, Figures \ref{fig:impact_g} and \ref{fig:impact_h} show that extending a 1D array to 2D can improve resolution when the added antennas lie outside the NCZ, increasing the spatial frequency bandwidth \( B_i \) for at least one direction. In this case, the AFR shrinks, as the added antennas become critical at some AFR boundary points, increasing \(K_i\) and tightening the aliasing constraints, in agreement with \cref{prop:lemme_ant_critic_addition}. 

Thus, increasing array dimensionality can either leave both AFR and resolution unchanged, or enhance the resolution while reducing the AFR, highlighting a trade-off. 
At the same time, 2D arrays more completely sample the spatial frequencies of the source, mitigating aliasing artefacts, which explains their reduced prominence compared to 1D cases. 

\begin{table}
    \centering
    \renewcommand{\arraystretch}{1.3}
    \begin{tabular}{lccccc}
    \hline
    \textbf{Effect} &
    \begin{tabular}[c]{@{}c@{}}\textbf{Spacing} \\ \textbf{(\(\Delta_i\))} \(\uparrow\)\end{tabular} &
    \begin{tabular}[c]{@{}c@{}}\textbf{Array} \\ \textbf{Aperture (\(D_i\))} \(\uparrow\) \end{tabular} &
    \begin{tabular}[c]{@{}c@{}}\textbf{Relative} \\ \textbf{distance} \(\uparrow\) \end{tabular} &
    \begin{tabular}[c]{@{}c@{}}\textbf{Array} \\ \textbf{Dim.} \(\uparrow\)\end{tabular} \\
    \hline
    Spectrum repetition       & \(\uparrow\) & --          & --          & -- \\
    Max frequency \(K_i\)     & --           & \(\uparrow\) or $=$ & \(\downarrow\) & \(\uparrow\) or $=$ \\
    Aliasing artefacts        & -- & -- & -- & \textcolor{green!60!black}{\(\downarrow\)} \\ \hline
    Alias-free region size    & \textcolor{red}{\(\downarrow\)} & \textcolor{red}{\(\downarrow\)} or $=$ & \textcolor{green!60!black}{\(\uparrow\)} & \textcolor{red}{\(\downarrow\)} or $=$ \\
    Resolution region $\mathcal{R}$ size     & --  & \textcolor{green!60!black}{\(\downarrow\)} & \textcolor{red}{\(\uparrow\)} & \textcolor{green!60!black}{\(\downarrow\)} or $=$ \\
    \hline
\end{tabular}
\caption{Impact of system parameters on the ambiguity function for a rectangular array geometry.}
\label{tab:parameter_effects}
\end{table}

\subsubsection{\textbf{Summary}}

Table~\ref{tab:parameter_effects} summarises how key system parameters affect the ambiguity function for a rectangular array. Arrows indicate the direction of change, with colour coding showing whether the effect is beneficial (green) or detrimental (red). Key observations include:
\begin{itemize}
    \item \textbf{Antenna spacing} ($\Delta_i$) : Larger spacing increases spectrum repetition and reduces the AFR, without affecting resolution.
    \item \textbf{Array aperture} ($D_i$) : Larger aperture improves resolution. AFR may shrink or stay constant depending on critical antennas.
    \item \textbf{Relative source distance}: Moving the source farther enlarges the AFR but reduces resolution.
    \item \textbf{Array dimensionality}: Extending from 1D to 2D can enhance resolution and mitigate aliasing artefacts, but the AFR may shrink or stay constant.
\end{itemize}

Overall, these results highlight a trade-off between resolution and AFR size: parameters that improve resolution (e.g., larger aperture, closer source, higher dimensionality) may reduce the AFR, while others (e.g., greater source distance) improve AFR at the cost of resolution. Propositions \ref{prop:lemme_ant_critic_removal} and \ref{prop:lemme_ant_critic_addition} provide practical criteria to determine when these changes are strict or when the AFR remains unchanged. %, based on critical antenna locations.

%%%%%%%%%%% Polar Array Analysis %%%%%%%%%%%

\section{Circular Array Analysis}
\label{sec:polar_arrays}

This section extends the AFR–resolution trade-off analysis to circular array geometries, examining the impact of system parameters through numerical results.

\subsection{Polar Spatial Frequencies}

The phase function $\phi_g$ can be written as 
\begin{equation}
     k \| \boldsymbol{z} - \boldsymbol{x}_s \| =  k \sqrt{r_z^2 + r_s^2 - 2 r_z r_s \cos(\theta_z - \theta_s)}, 
\label{eq:polar_distance}
\end{equation}
using the law of cosine \cite{euclides_euclids_2008}, where $(r_z, \theta_z)$ and $(r_s, \theta_s)$ denote the polar coordinates of $\boldsymbol{z}$ and $\boldsymbol{x}_s$ respectively.
The local vector \( \boldsymbol{\mathrm{k}}^{g} \) along a given axis \( i \in \{r, \theta\} \) can then be expressed as
\begin{equation}
\mathrm{k}_{r}(\boldsymbol{z} ; \tilde{\boldsymbol{x}}_s, \boldsymbol{x}_s) =
k \left( 
\frac{r_z - \tilde{r}_s \cos(\theta_z - \tilde{\theta}_s)}{\| \boldsymbol{z} - \tilde{\boldsymbol{x}}_s \|}
-
\frac{r_z - r_s \cos(\theta_x - \theta_s)}{\| \boldsymbol{z} - \boldsymbol{x}_s \|}
\right)
\label{eq:k_r_polar}
\end{equation}
for the radial direction, or for the angular one as 
\begin{equation}
\mathrm{k}_{\theta}(\boldsymbol{z} ; \tilde{\boldsymbol{x}}_s, \boldsymbol{x}_s) =
k \left(
\frac{r_z \tilde{r}_s \sin(\theta_z - \tilde{\theta}_s)}{\| \boldsymbol{z} - \tilde{\boldsymbol{x}}_s \|}
-
\frac{r_z r_s \sin(\theta_z - \theta_s)}{\| \boldsymbol{z} - \boldsymbol{x}_s \|}, 
\right). 
\label{eq:k_theta_polar}
\end{equation}

\subsection{Safe Antenna Spacing}

Proposition \ref{prop:safe_spacing_polar} gives a condition on the angular spacing of antennas in circular arrays to avoid aliasing. Note that, in the radial direction, the analysis recovers the classical $\lambda/2$ sampling requirement, consistent with the linear array case presented in \cref{sec:safe_ant_spacing_linear}. 

\begin{proposition}\label{prop:safe_spacing_polar}
    For a 1D circular array of radius $R$ ($r_z = R$, $\forall \boldsymbol{z} \in \mathcal{Z}$) and a source inside the array ($r_s \le R$), aliasing is avoided for all test points $\tilde{\boldsymbol{x}}_s$ within the array ($\tilde{r}_s \leq R$) if
    \begin{equation}
        \Delta_\theta \leq \frac{\lambda}{2R}.
    \label{eq:safe_spacing_polar}
    \end{equation}
\end{proposition}

\begin{proof}
    Aliasing is avoided if $\forall \tilde{\boldsymbol{x}}_s$ \(\mathrm{K}_\theta(\tilde{\boldsymbol{x}}_s,\boldsymbol{x}_s)\le 2\pi/\Delta_\theta \) (cf. \eqref{eq:aliasing_conditions}). From \eqref{eq:k_theta_polar}, the maximum angular spatial frequency occurs at $\theta_z - \tilde{\theta}_s = \pi/2$ and $\theta_z - \theta_s = -\pi/2$: 
    \begin{equation}
    \mathrm{K}_\theta (\tilde{\boldsymbol{x}}_s,\boldsymbol{x}_s) = k \left( \frac{R \tilde{r}_s}{\sqrt{R^2 + \tilde{r}_s^2}} + \frac{R r_s}{\sqrt{R^2 + r_s^2}} \right) \leq 2kR,  
    \end{equation} 
    where the inequality holds for all $r_s, \tilde{r}_s \le R$. We have $\mathrm{K}_\theta (\tilde{\boldsymbol{x}}_s,\boldsymbol{x}_s) \leq 2kR$ and thus $\Delta_\theta \leq \frac{\pi}{kR} = \frac{\lambda}{2R}$. 
\end{proof}

As $R$ increases, the local spatial frequencies vary more rapidly between neighboring antennas, requiring finer angular sampling to avoid aliasing.

\subsection{Parameter Analysis}

\begin{figure*}
\centering 

\includegraphics[width=1\linewidth, trim=0 35 0 35, clip]{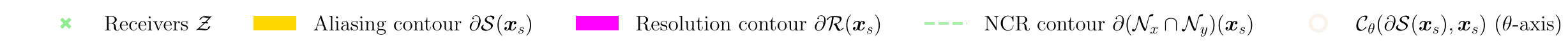}

\vspace{0.3cm}

% 1ère Subfigure
\begin{subfigure}{0.47\linewidth}
\centering 
\includegraphics[width=\linewidth, trim=20 392 160 20, clip]{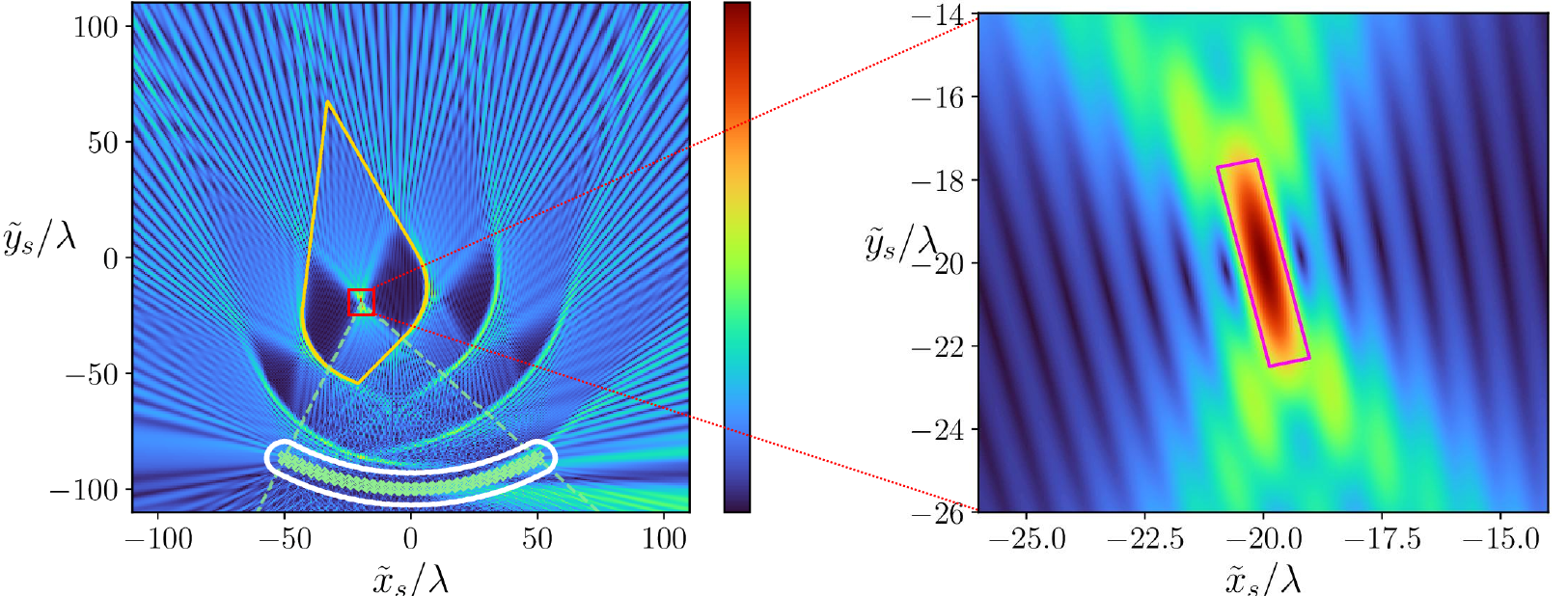}
\caption{Source location: $\boldsymbol{x}_s = (-20;-20)\lambda$. Array aperture: $60^{\circ}$. Number of antennas: $N=32$. Array radius: $R=100 \lambda$. }
\label{fig:impact_circular_a}
\end{subfigure}
\hfill
% 2ème Subfigure
\begin{subfigure}{0.47\linewidth}
\centering 
\includegraphics[width=\linewidth, trim=20 392 160 20, clip]{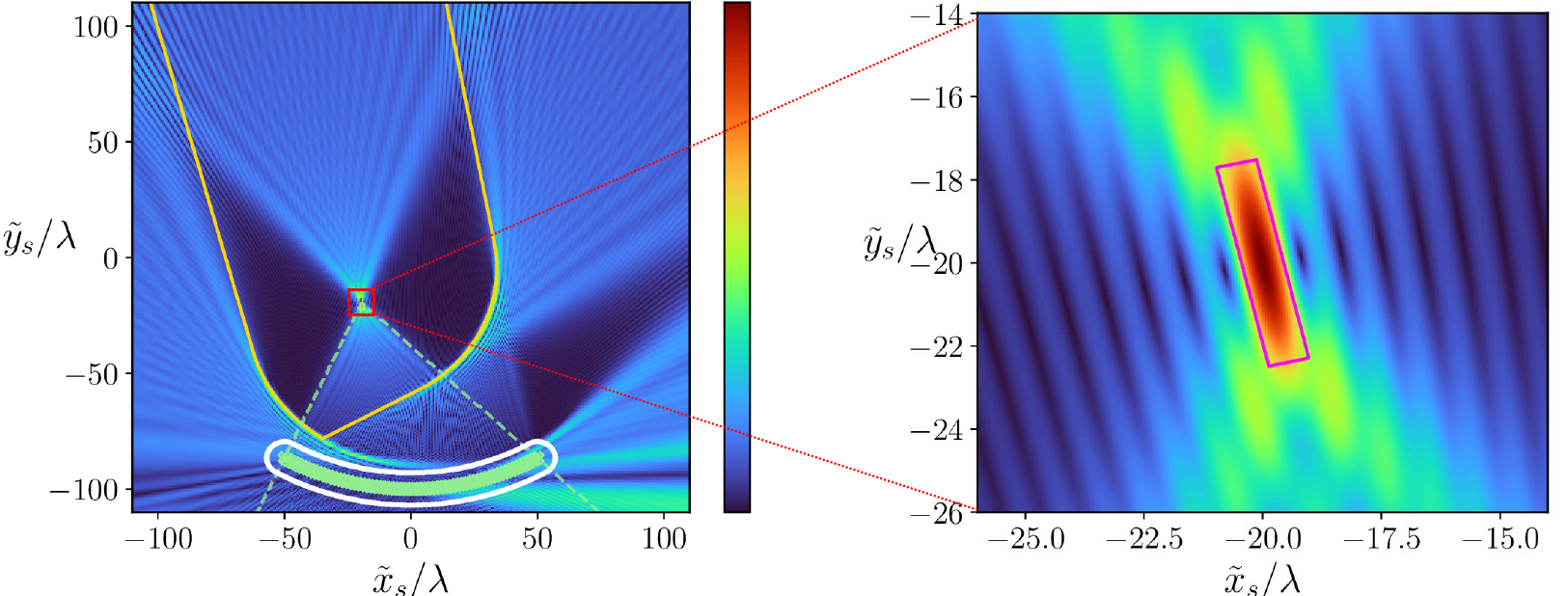}
\caption{Source location: $\boldsymbol{x}_s = (-20;-20)\lambda$. Array aperture: $60^{\circ}$. Number of antennas: $N=64$.  Array radius: $R=100 \lambda$.}
\label{fig:impact_circular_b}
\end{subfigure}

\vspace{0.4cm}

% 3ème Subfigure
\begin{subfigure}{0.47\linewidth}
\centering 
\includegraphics[width=\linewidth, trim=20 392 160 20, clip]{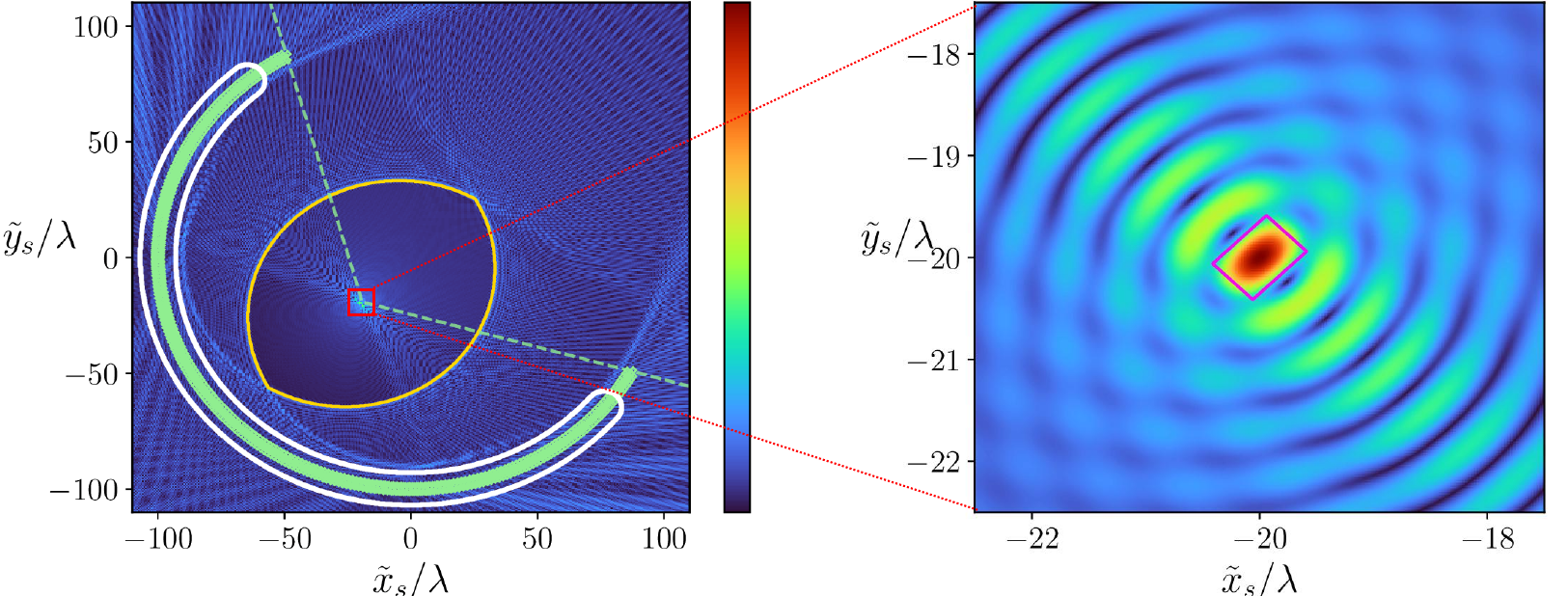}
\caption{Source location: $\boldsymbol{x}_s = (-20;-20)\lambda$. Array aperture: $210^{\circ}$. Number of antennas: $N=224$.  Array radius: $R=100 \lambda$.}
\label{fig:impact_circular_c}
\end{subfigure}
\hfill
% 4ème Subfigure
\begin{subfigure}{0.47\linewidth}
\centering 
\includegraphics[width=\linewidth, trim=20 392 160 20, clip]{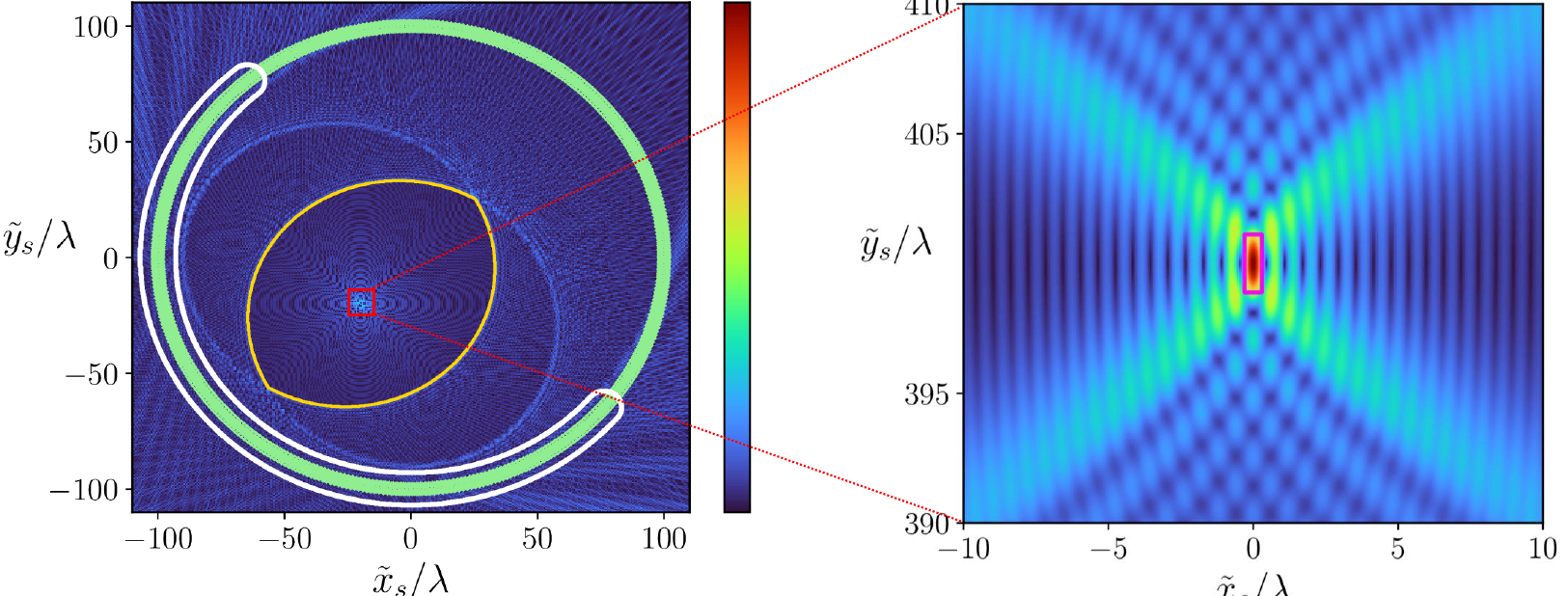}
\caption{Source location: $\boldsymbol{x}_s = (-20;-20)\lambda$. Array aperture: $360^{\circ}$. Number of antennas: $N=384$.  Array radius: $R=100 \lambda$. }
\label{fig:impact_circular_d}
\end{subfigure}

\vspace{0.4cm}

% 5ème Subfigure
\begin{subfigure}{0.47\linewidth}
\centering 
\includegraphics[width=\linewidth, trim=20 392 160 20, clip]{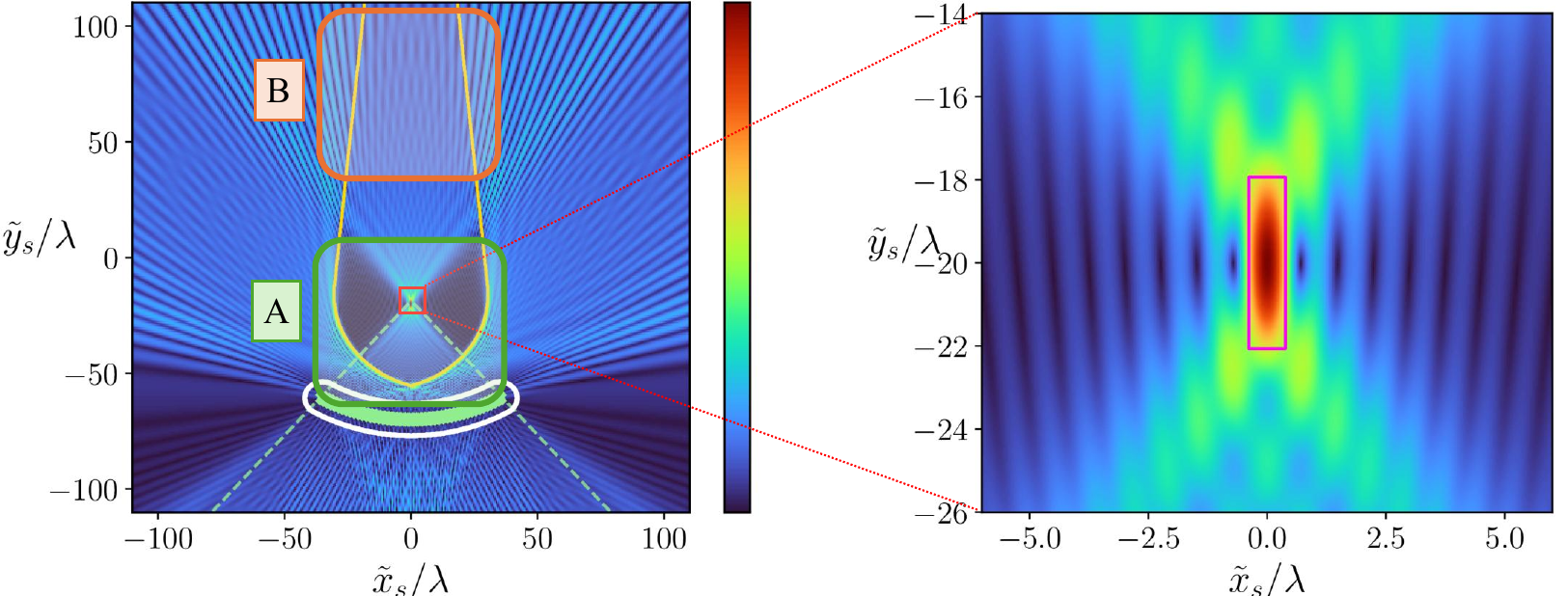}
\caption{Source location: $\boldsymbol{x}_s = (0;-20)\lambda$. Array aperture: $60^{\circ}$. Number of antennas: $N=45$.  Array radius: $R=70 \lambda$. }
\label{fig:impact_circular_e}
\end{subfigure}
\hfill
% 6ème Subfigure
\begin{subfigure}{0.47\linewidth}
\centering 
\includegraphics[width=\linewidth, trim=20 392 160 20, clip]{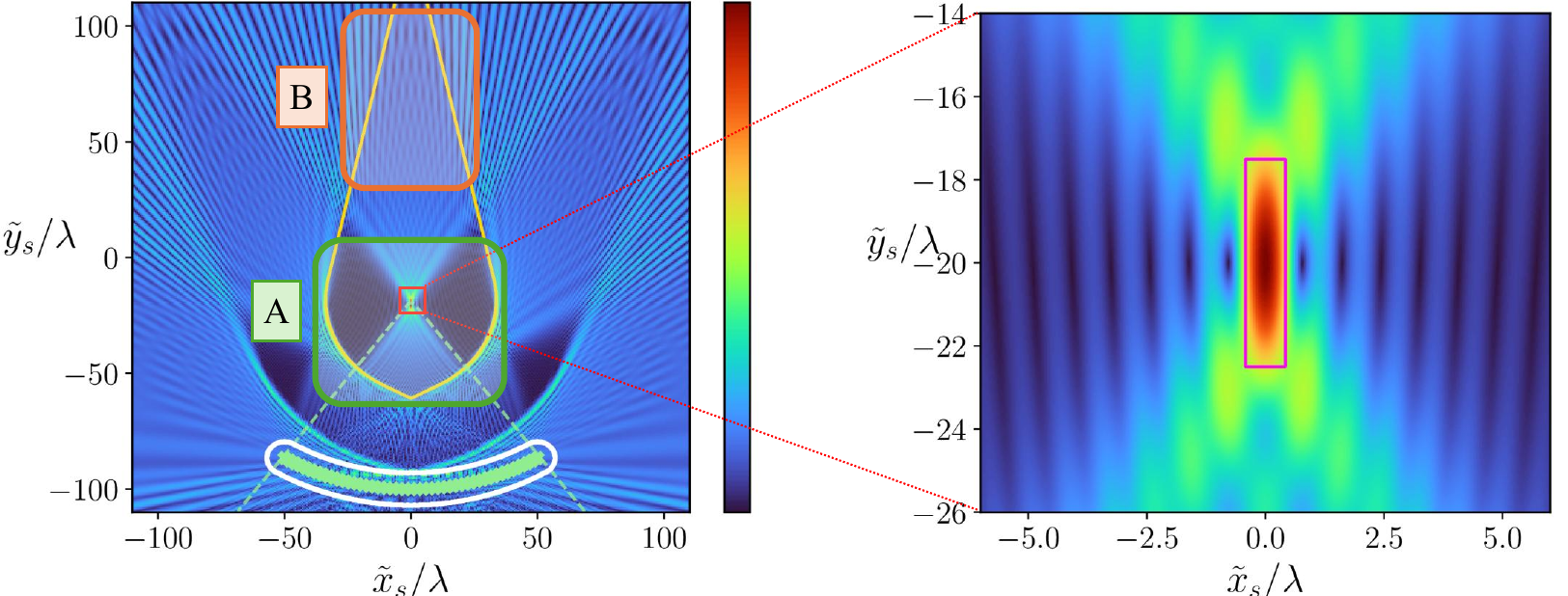}
\caption{Source location: $\boldsymbol{x}_s = (0;-20)\lambda$. Array aperture: $60^{\circ}$. Number of antennas: $N=45$.  Array radius: $R=100 \lambda$.  }
\label{fig:impact_circular_f}
\end{subfigure}

\vspace{0.4cm}

% 7ème Subfigure
\begin{subfigure}{0.47\linewidth}
\centering 
\includegraphics[width=\linewidth, trim=20 392 160 20, clip]{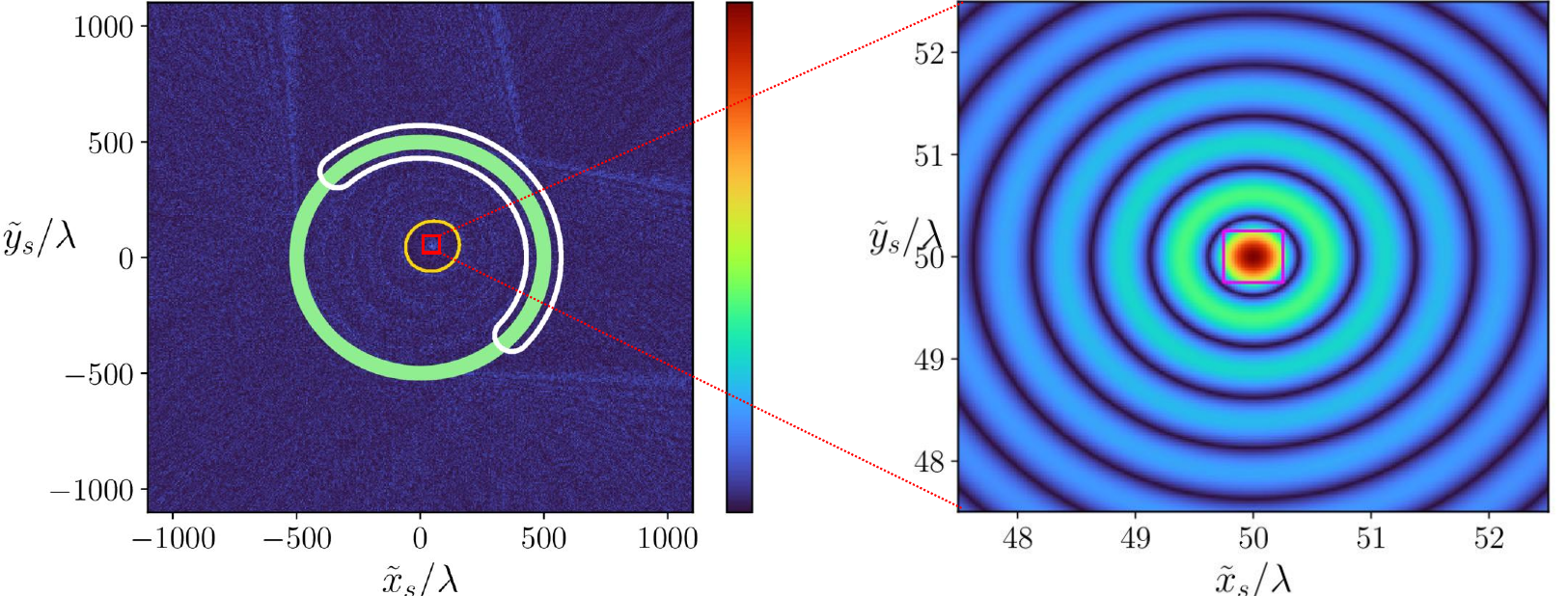}
\caption{Source location: $\boldsymbol{x}_s = (50;50)\lambda$. Array aperture: $360^{\circ}$. Number of antennas: $N=720$.  Array radius: $R=500 \lambda$.  }
\label{fig:impact_circular_g}
\end{subfigure}
\hfill
% 8ème Subfigure
\begin{subfigure}{0.47\linewidth}
\centering 
\includegraphics[width=\linewidth, trim=20 392 160 20, clip]{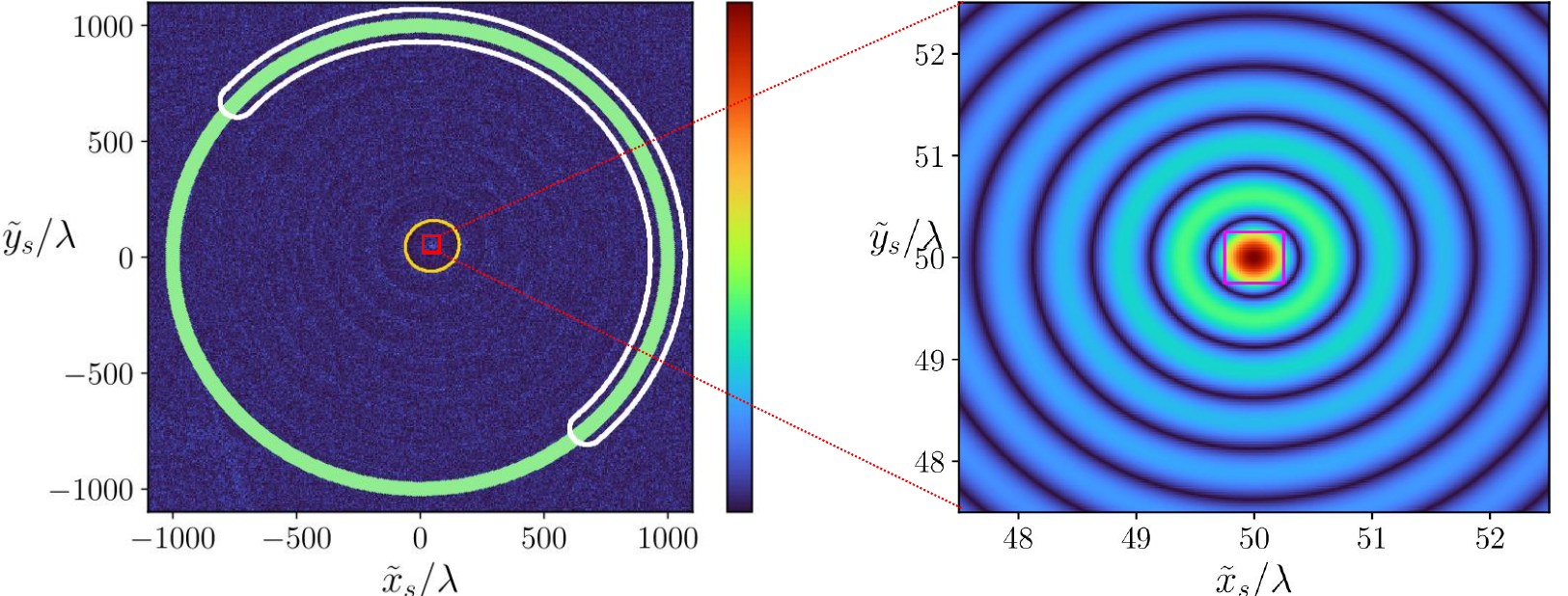}
\caption{Source location: $\boldsymbol{x}_s = (50;50)\lambda$. Array aperture: $360^{\circ}$. Number of antennas: $N=720$.  Array radius: $R=1000 \lambda$. }
\label{fig:impact_circular_h}
\end{subfigure}

\vspace{0.4cm}

% 9ème Subfigure
\begin{subfigure}{0.47\linewidth}
\centering 
\includegraphics[width=\linewidth, trim=20 392 160 20, clip]{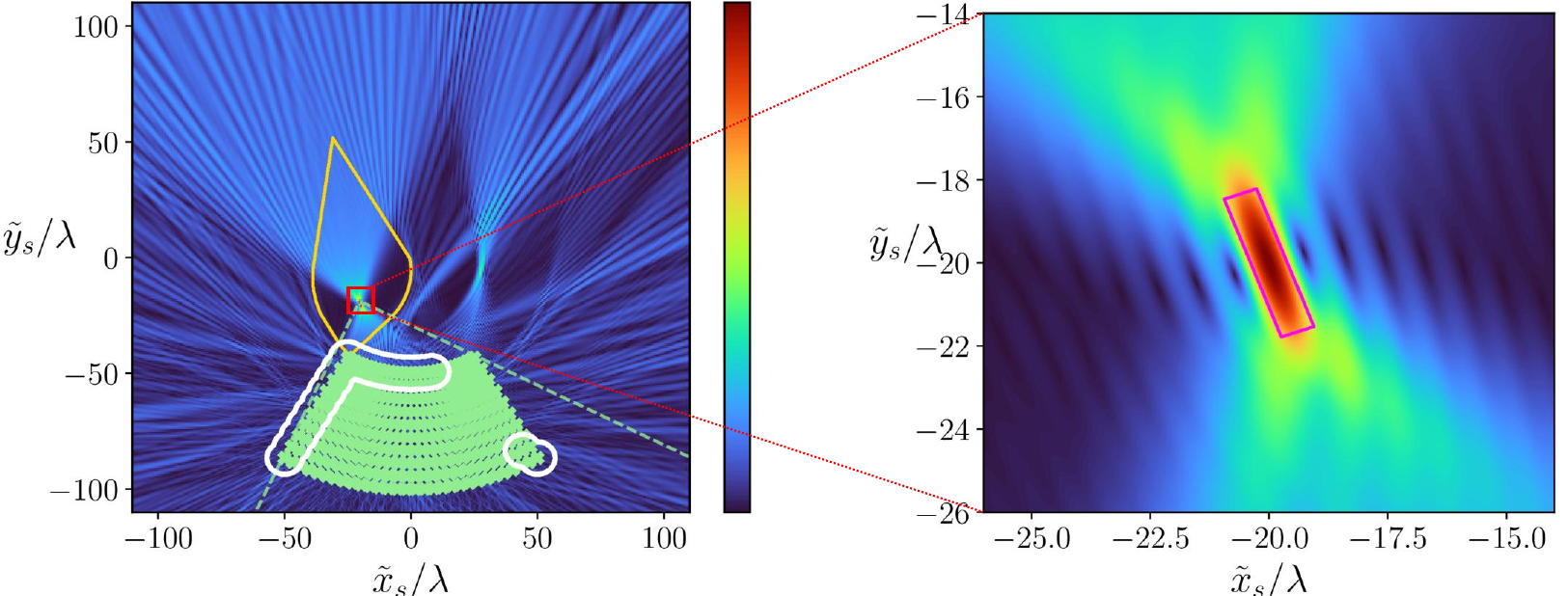}
\caption{Source location: $\boldsymbol{x}_s = (-20;-20)\lambda$. Array aperture: $60^{\circ}$. Number of antennas: $N_\theta \times N_r =32 \times 10$.  Array radius: $R= [50;100] \lambda$.  }
\label{fig:impact_circular_i}
\end{subfigure}
\hfill
% 10ème Subfigure
\begin{subfigure}{0.47\linewidth}
\centering 
\includegraphics[width=\linewidth, trim=20 392 160 20, clip]{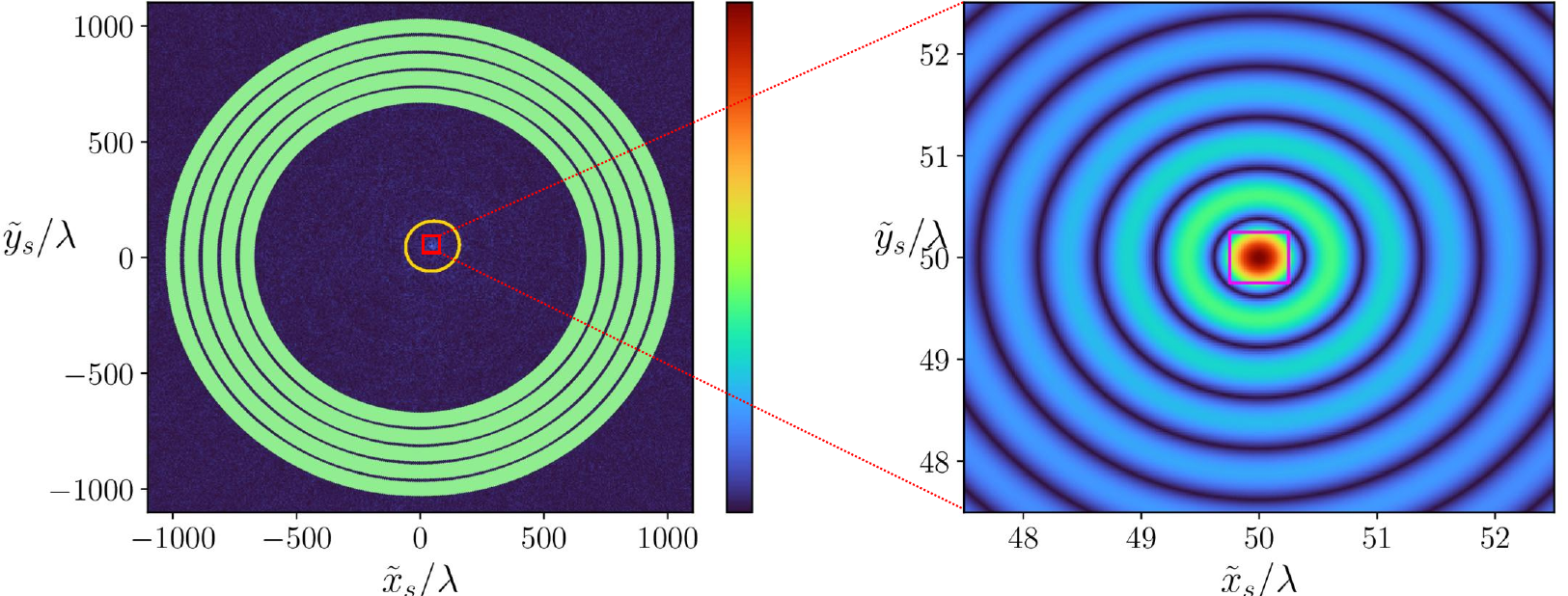}
\caption{Source location: $\boldsymbol{x}_s = (50;50)\lambda$. Array aperture: $360^{\circ}$. Number of antennas: $N_\theta \times N_r =720 \times 5$.  Array radius: $R= [700;1000] \lambda$. }
\label{fig:impact_circular_j}
\end{subfigure}

\caption{Impact of system parameters on the $AF$ for a point source with circular arrays. }

\label{fig:impact_parameters_polar}
\end{figure*}

Figure~\ref{fig:impact_parameters_polar} illustrates the impact of system parameters on the AFR for a point source in circular arrays. The legend at the top follows that of Figure~\ref{fig:impact_parameters_cart}, with adjustments for circular geometry. For clarity, the spatial frequencies $\boldsymbol{\mathrm{k}}^{h}$ used to assess resolution are evaluated in a rectangular system aligned with the beam, ensuring that the resolution region $\mathcal{R}(\boldsymbol{x}_s)$ accurately reflects resolving capability along and across the beam. The following observations describe the parameter impacts. 

\subsubsection{Angular Spacing} 

Figures \ref{fig:impact_circular_a} and \ref{fig:impact_circular_b} differ only in angular antenna spacing, with \( N=32 \) and \( N=64 \) antennas, respectively, over the same array aperture of \( 60^{\circ} \). 
The NCZ contours remain identical, indicating that resolution along both directions is unchanged due to constant aperture, as seen in the zoomed-in plots.
However, increasing angular spacing (reducing the number of antennas) significantly decreases the AFR. As in rectangular arrays, the spatial spectrum \(G\) repeats more densely with a period \(2\pi/( \Delta_\theta)\) that decreases with larger spacing $\Delta_\theta$, raising the risk of aliasing according to~\eqref{eq:aliasing_conditions}.

\subsubsection{Array Aperture}

From \cref{fig:impact_circular_b} to \cref{fig:impact_circular_c} and \cref{fig:impact_circular_d}, the array aperture increases from \( 60^{\circ} \)  to \( 210^{\circ} \) and then to \( 360^{\circ} \), with constant angular spacing and source location. Resolution improves as antennas are added beyond the non-contributive zones, increasing spatial frequency bandwidth $\mathrm{B}_i$ and reducing the size of the resolution region $\mathcal{R}(\boldsymbol{x}_s)$ delineated in magenta. 

For the AFR, expanding from $60^\circ$ to $210^\circ$, shrinks it because new antennas become critical at some AFR boundary points (\cref{prop:lemme_ant_critic_addition}).
Further expansion to $360^\circ$ adds no new critical antennas, leaving the AFR unchanged. Conversely, from $360^\circ$ to $210^\circ$, non-critical antennas can be removed without affecting the AFR (\cref{prop:lemme_ant_critic_removal}), but enlarges the resolution cell.

Increasing the array aperture thus reduces the AFR up to a certain limit, while consistently enhancing resolution, highlighting possible trade-off situations. 

\subsubsection{Array Radius}
\label{sec:impact_circular_array_radius}

The effect of array radius is more complex. Figures \ref{fig:impact_circular_e} and \ref{fig:impact_circular_f} show the same source, aperture, and number of antennas, but radii \( R=70\lambda \)  and \( R=100\lambda \). 

Increasing the radius expands the AFR in zone A (green), while it shrinks in zone B (orange). This behaviour arises from the dependence of the maximal angular spatial frequency \( \mathrm{K}_\theta \) on the considered pair \( (\tilde{\boldsymbol{x}}_s, \boldsymbol{x}_s) \) and the antenna positions \( \boldsymbol{z} \), as discussed in \cref{sec:critical_antenna_location_polar}. As a result, a larger radius can either increase or decrease \( \mathrm{K}_\theta \) at different points along the AFR boundary, producing the observed non-uniform evolution of the AFR. Moreover, \cref{sec:critical_antenna_location_polar} established that, for a full circular array, \( \mathrm{K}_\theta \) becomes independent of $R$ once the radius is sufficiently large. This is confirmed by comparing Figures~\ref{fig:impact_circular_g} and \ref{fig:impact_circular_h}, where the AFRs are identical despite increasing the radius from \( 500\lambda \) to \( 1000\lambda \).

Regarding resolution: 
\begin{itemize}
    \item from \cref{fig:impact_circular_e} to \cref{fig:impact_circular_f}, the resolution region is enlarged with greater radius (worsened resolution); 
    \item from \cref{fig:impact_circular_g} to \cref{fig:impact_circular_h}, it remains the same;
    \item in \cref{fig:impact_circular_i}, the NCZ is enlarged (improved resolution) on its left border by the last row of antennas at greater radius.
\end{itemize}

\subsubsection{Array Dimensionality} 

Comparing \cref{fig:impact_circular_i} with \cref{fig:impact_circular_a}, $N_r=10$ radial positions are added between \( R=50\lambda \) and \( R=100\lambda \), keeping $N_\theta=32$ over a $60^\circ$ angular aperture. This 2D layout introduces multiple radii, so the AFR–radius behaviour discussed in \cref{sec:impact_circular_array_radius} must be considered for all of them. For each region, the smallest AFR among those associated with each radius is retained.

As seen in \cref{def:afr}, the AFR is obtained as the intersection of the directional AFRs associated with each dimension. Consequently, adding a radial dimension introduces an additional constraint and thus may reduce the overall AFR--which is the case here.
Besides, as shown in \cref{fig:chirp_illustration_polar}, the local radial frequency \( \mathrm{k}_r^{g} \) increases at smaller radii (the corresponding critical antenna is at the smallest radius). Hence, the smaller the inner radius of the array, the larger the maximum radial frequency \( \mathrm{K}_r \), and the stricter the aliasing constraints along this direction. 

The additional antennas lie outside the NCZ $\mathcal{N}(\boldsymbol{x}_s)$ of the 1D case, thereby increasing the spatial frequency bandwidth \( B_i \) in at least one direction and improving resolution, as shown by the smaller resolution cell in the zoomed-in plot.

However, comparing \cref{fig:impact_circular_j} with \ref{fig:impact_circular_h}, $N_r=5$ radial positions are added, while keeping $N_\theta=720$ angular positions with a full $360^{\circ}$ aperture. Here, the radii are sufficiently large so that \( \mathrm{K}_\theta \) is independent of the radius, as explained in \cref{sec:impact_circular_array_radius}. Since the radial spacing $\Delta_r$ is also small enough to avoid additional aliasing constraints, the AFR remains identical in the two figures. With the aperture unchanged, the resolution regions are identical as well.

\subsubsection{Summary}

Table~\ref{tab:parameter_effects_polar} summarises the impacts of various system parameters on resolution and AFR for circular arrays. Each parameter is varied independently, with the legend matching Table~\ref{tab:parameter_effects}. Compared to rectangular arrays, the AFR shape differs and the impact of source-array distance (i.e., array radius) is more intricate, depending on both the AF region and the direction (radial or angular), leading to more complex trade-off situations. 

\begin{table}
    \centering
    \renewcommand{\arraystretch}{1.3}
    \begin{tabular}{lccccc}
    \hline
    \textbf{Effect} &
    \begin{tabular}[c]{@{}c@{}}\textbf{Spacing} \\ \textbf{(\(\Delta\))} \(\uparrow\)\end{tabular} &
    \begin{tabular}[c]{@{}c@{}}\textbf{Array} \\ \textbf{Aperture} \(\uparrow\)\end{tabular} &
    \begin{tabular}[c]{@{}c@{}}\textbf{Array} \\ \textbf{radius (\(R\))} \(\uparrow\)\end{tabular} &
    \begin{tabular}[c]{@{}c@{}}\textbf{Array} \\ \textbf{Dim.} \(\uparrow\)\end{tabular} \\
    \hline
    Spectrum repetition       & \(\uparrow\) & -- & -- & -- \\
    Max frequency \(K_i\)     & -- & \(\uparrow\) or $=$ & \(\downarrow\) ($r$) - \(\downarrow\) \&/or $=$ \&/or \(\uparrow\) ($\theta$) & \(\uparrow\) or $=$ \\
    Aliasing artefacts        & -- & -- & -- & \textcolor{green!60!black}{\(\downarrow\)} \\
    \hline
    Alias-free region size    & \textcolor{red}{\(\downarrow\)} & \textcolor{red}{\(\downarrow\)} or $=$ & \textcolor{red}{\(\downarrow\)} (r) - \textcolor{red}{\(\downarrow\)} \&/or $=$ \&/or \textcolor{green!60!black}{\(\uparrow\)} & \textcolor{red}{\(\downarrow\)} or $=$ \\
    Resolution region $\mathcal{R}$ size   & -- & \textcolor{green!60!black}{\(\downarrow\)} & \textcolor{red}{\(\downarrow\)} or = or \textcolor{green!60!black}{\(\downarrow\)} & \textcolor{green!60!black}{\(\downarrow\)} or = \\
    \hline
    \end{tabular}%
    \caption{Impact of system parameters on the AF for a circular array. $r$ and $\theta$ denote the radial and angular directions.}
    \label{tab:parameter_effects_polar}
\end{table}

%%%%%%%%%%% Conclusion %%%%%%%%%%%

\section{Conclusion}
\label{sec:conclusion}

In this paper, we  have introduced a joint framework for analysing the impacts of resolution, aliasing and their trade-off on the ambiguity function of monostatic systems with spatially distributed antenna arrays. By defining critical antennas and the non-contributive zone, we have derived explicit conditions for the aliasing-free region (AFR) and resolution limits as functions of array geometry, source position, and array dimensionality. The analysis clarifies how these parameters shape the aliasing–resolution trade-off.
We have showed that increasing array size or dimensionality can enhance resolution without shrinking the AFR, provided that added antennas remain non-critical. In contrast, antenna spacing and source distance directly constrain the AFR and must be carefully controlled, while parameters such as array radius in circular configurations have strong configuration-dependent effects.
Numerical results confirm the theoretical predictions and illustrate the practical relevance of the proposed framework.

Future work will extend this analysis to noisy scenarios, bistatic configurations, non-uniform array geometries, and will explore the incorporation of the derived aliasing and resolution criteria into array design and optimisation frameworks.

%%%%%%%%%%% Appendix %%%%%%%%%%%

\appendix
%\section*{Appendix} 
\setcounter{section}{0} % remettre à zéro
\renewcommand{\thesection}{\Alph{section}} % sections = A, B, C...

\subsection{Loci of Critical Antennas for Infinite Linear Arrays}
\label{appendix:critical_antenna_location_linear}

We seek the antenna positions \( \boldsymbol{z} = (x,y) \) such that \( k_x(\boldsymbol{z} ; \tilde{\boldsymbol{x}}_s, \boldsymbol{x}_s) \) is maximised. Using the chirp-based approach, this leads to 
\begin{equation}
    \frac{\partial^2 \phi_g(\boldsymbol{z} ; \tilde{\boldsymbol{x}}_s, \boldsymbol{x}_s)}{\partial x^2} = 0,
\end{equation}
which yields, after substitution and rearrangements,
\begin{equation}
(\tilde{y}_s - y )^{4/3} ((x_s -x)^2 + (y_s - y)^2) 
= (y_s - y )^{4/3} ((\tilde{x}_s -x)^2 + (\tilde{y}_s - y)^2)
\end{equation}

Defining \( A \triangleq y - y_s \) and \( B \triangleq y - \tilde{y}_s \), this leads to a quadratic equation in \( x \):
\begin{equation}
\begin{split}
    &\left(A^{4/3}-B^{4/3}\right) x^2
    - 2x\left(A^{4/3} \tilde{x}_s - B^{4/3} x_s\right) \\ 
    &\quad + \left[ A^{4/3}\left(\tilde{x}_s^2 + B^2\right)
    - B^{4/3}\left(x_s^2 + A^2\right) \right] = 0.
\end{split}
\end{equation}

This last equation can be rewritten in the canonical form
\begin{equation}
    (x - x_0(y))^2 + \frac{\tilde{c}(y)}{\tilde{a}} = 0,
\end{equation}
where, by identification,
\begin{equation}
\left\{
\begin{aligned}
    \tilde{a} &= A^{4/3} - B^{4/3}, \\[2pt]
    x_0(y) &= \frac{A^{4/3} \tilde{x}_s - B^{4/3} x_s}{A^{4/3} - B^{4/3}}, \\[2pt]
    \tilde{c}(y) &= A^{4/3}\!\left(\tilde{x}_s^2 + B^2\right) - B^{4/3}\!\left(x_s^2 + A^2\right) \\
    &\quad - \frac{\big(A^{4/3} \tilde{x}_s - B^{4/3} x_s\big)^2}{A^{4/3} - B^{4/3}}.
\end{aligned}
\right.
\end{equation}

Assuming \( A \approx B \) (i.e., \( B = A + \epsilon \) with \( \epsilon = y_s - \tilde{y}_s \) small), a first-order Taylor expansion gives
$
    B^{4/3} \approx A^{4/3} + 4/3 A^{1/3} \epsilon,
$
leading to
\begin{equation}
    x_0(y) \approx A u_1 + u_0, \text{ with }
    u_1 = -\frac{3}{4} \frac{\tilde{x}_s - x_s}{\epsilon}, \quad
    u_0 = x_s,
\end{equation}
and
\begin{equation}
    - \frac{\tilde{c}(y)}{\tilde{a}} \approx A^2 u_2 + A u_3,
\end{equation}
where
\begin{equation}
\left\{
\begin{aligned}
    u_2 &= \frac{9}{16} \frac{(\tilde{x}_s - x_s)^2}{\epsilon^2} + \frac{1}{2}, \\[2pt]
    u_3 &= \frac{{x}_s^2 - \tilde{x}_s^2 - \epsilon^2}{-4\epsilon/3} - \frac{3}{2} x_s \frac{\tilde{x}_s - x_s}{\epsilon}.
\end{aligned}
\right.
\end{equation}

Thus, the loci of critical antennas can be approximated as
\begin{equation}
    (x - (A u_1 + u_0))^2 + A^2 u_2 + A u_3 = 0.
\end{equation}

Rearranging terms leads to the conic section form
\begin{equation}
    a x^2 + b x A + c A^2 + d x + e A + f = 0,
\label{eq:conic_section}
\end{equation}
with coefficients
\begin{equation}
\left\{
\begin{aligned}
    a &= 1, & b &= -2 u_1, & c &= u_1^2 - u_2,\\
    d &= -2 u_0, & e &= -2 u_0 u_1 - u_3, & f &= u_0^2, 
\end{aligned}
\right.
\end{equation}

whose nature is determined by the discriminant \( D = b^2 - 4ac \) \cite{korn_mathematical_2013, granino_a_korn_theresa_m_korn_math_2006, zelator_integral_2009}. 
For our parameters (\(a = 1\), \(b = -2u_1\), \(c = u_1^2 - u_2 = -1/2\)), we have \(D = 4 u_1^2 + 2 > 0\), indicating a hyperbola.  
The corresponding asymptotes can be derived as described in Appendix~\ref{appendix:asymptotes_hyperbola}.

\subsection{Asymptotes of a Hyperbola}
\label{appendix:asymptotes_hyperbola}

For a conic
\begin{equation}
    a x^2 + b x y + c y^2 + d x + e y + f = 0,
\end{equation}
with $D=b^2-4ac>0$, the asymptotes $y = m x + p$ satisfy \cite{korn_mathematical_2013, granino_a_korn_theresa_m_korn_math_2006, zelator_integral_2009}
\begin{equation}
\left\{
\begin{aligned}
    a + b m + c m^2 &= 0, \\[3pt]
    b p + 2 c m p + e m + d &= 0.
\end{aligned}
\right.
\end{equation}
Thus,
\begin{equation}
    m = \frac{-b \pm \sqrt{D}}{2c}, \qquad
    p = -\frac{e m + d}{b + 2 c m}.
\end{equation}

\subsection{Directional Cosine for Far-Field Conditions in Rectangular Arrays}
\label{appendix:ff_conditions_cartesian}

Let the source and tentative source be written as
\begin{equation}
\boldsymbol{x}_s = D(\cos\theta_s,\sin\theta_s), \qquad
\tilde{\boldsymbol{x}}_s = \tilde{D}(\cos\theta_s,\sin\theta_s).
\end{equation}
In the far field ($D,\tilde{D}\gg\|\boldsymbol{z}\|$), following a first order Taylor expansion, 
$
\|\boldsymbol{z}-\boldsymbol{x}_s\|
\approx D - \boldsymbol{z}\cdot\boldsymbol{s},
$
which implies
\begin{equation}
    \frac{\boldsymbol{z} - \boldsymbol{x}_s}{\| \boldsymbol{z} - \boldsymbol{x}_s \|} \approx \frac{\boldsymbol{z} - D \boldsymbol{s}}{D - \boldsymbol{z} \cdot \boldsymbol{s}} \approx -\boldsymbol{s},
\end{equation}
Applying the same approximation to $\tilde{\boldsymbol{x}}_s$ gives
$
    \boldsymbol{\mathrm{k}}^{g}(\boldsymbol{z}) \approx k(\boldsymbol{s}-\tilde{\boldsymbol{s}}),
$
showing that, for rectangular arrays in FF conditions, $\boldsymbol{\mathrm{k}}^{g}$ becomes independent of $\boldsymbol{z}$ and depends only on the difference in directional cosines.

%%%%%%%%%%% References %%%%%%%%%%%

\bibliographystyle{ieeetr}
\bibliography{biblio}

\end{document}